\documentclass[a4paper, twoside, 11pt]{article}
\usepackage{geometry}
\usepackage{amsmath, amsthm, amssymb, dsfont, mathrsfs, BOONDOX-cal, upgreek}
\usepackage{graphicx, caption, float, multirow}
\usepackage{chngcntr, etoolbox}
\usepackage[numbers]{natbib}

\usepackage[T1]{fontenc}

\usepackage{hyperref}
\usepackage[dvipsnames]{xcolor}
\hypersetup{colorlinks=true, linkcolor=NavyBlue, citecolor=NavyBlue}

\usepackage[title]{appendix}

\usepackage{fancyhdr}
\usepackage{txfonts}
\usepackage[T1]{fontenc}

\allowdisplaybreaks

\counterwithin*{equation}{section}
\counterwithin*{equation}{subsection}
\renewcommand\theequation{\ifnumgreater{\value{subsection}}{0}{\thesubsection.}{\thesection.}\arabic{equation}}

\newtheoremstyle{theorem}
  {11pt}
  {11pt}
  {\sl}
  {}
  {\bf}
  {. }
  { }
  {}
\theoremstyle{theorem}

\newtheorem{theorem}{Theorem}[section]
\newtheorem{proposition}[theorem]{Proposition}
\newtheorem{lemma}[theorem]{Lemma}

\newtheorem{remark}[theorem]{Remark}
\newtheorem{definition}[theorem]{Definition}

\newtheorem{assumption}[theorem]{Assumption}

\begin{document}

\title{\textbf{Insider and stealth trading with dynamic legal risk}}
\author{Bixing Qiao\thanks{Department of Mathematics, University of Southern California. Email: bqiao@usc.edu} \and Weixuan Xia\thanks{Department of Mathematics, University of Southern California. Email: weixuanx@usc.edu}}
\date{2026}
\maketitle

\begin{abstract}
  The present paper investigates how insiders strategically navigate ongoing legal risk while leveraging stealth trading within a continuous-time Kyle-type framework. Legal enforcement operates concurrently with trading, which dynamic can be adversely obscured by a large surrounding population of noise traders. While surveillance intensity responds directly to the insider's trading intensity, triggering a random prosecution time, the resulting legal sanctions encompass both strategy-focused criminal penalties and profit-dependent civil penalties. Employing a new impact-neutral measure change, equilibrium analysis shows that even after achieving stealth, the insider internalizes regulatory exposure, and enforcement can significantly shape equilibrium trading strategies. The associated limiting equilibria yield a rich set of outcomes, with three key insights for regulatory impact: (i) under committed regulatory scrutiny, the insider trades a time-varying function of the discrepancy between the asset's fundamental value and its market price, and trading may intensify indefinitely near the end of the trading horizon as legal risk recedes; (ii) merely raising penalties as an advantageous selection cost proves ineffective in offsetting declines in regulatory diligence; (iii) criminal penalties remain essential for deterring aggressive insider trading, as they impose critical temporal constraints on trading intensity not achievable through civil penalties alone. \medskip
\end{abstract}

\newcommand{\dd}{{\rm d}}
\newcommand{\pd}{\partial}
\newcommand{\PP}{\mathbb{P}}
\newcommand{\E}{\mathbb{E}}
\newcommand{\Var}{\mathrm{Var}}
\newcommand{\1}{\mathds{1}}
\newcommand{\supp}{\mathrm{supp}}
\newcommand{\sgn}{\mathrm{sgn}}
\renewcommand{\Omega}{\varOmega}
\renewcommand{\Theta}{\varTheta}
\renewcommand{\Pi}{\varPi}
\renewcommand{\Lambda}{\varLambda}
\renewcommand{\Phi}{\varPhi}
\renewcommand{\Psi}{\varPsi}
\renewcommand{\Gamma}{\varGamma}
\renewcommand{\Delta}{\varDelta}

\def\a{\alpha}
\def\b{\beta}
\def\g{\gamma}
\def\d{\delta}
\def\e{\epsilon}
\def\z{\zeta}
\def\k{\kappa}
\def\l{\lambda}
\def\m{\mu}
\def\n{\nu}
\def\si{\sigma}
\def\t{\tau}
\def\f{\varphi}
\def\th{\theta}
\def\o{\omega}
\def\h{\widehat}
\def\G{\Gamma}
\def\D{\Delta}
\def\Th{\Theta}
\def\L{\Lambda}
\def\Si{\Sigma}
\def\F{\Phi}
\def\tl{\tilde}

\def\cA{{\cal A}}
\def\cB{{\cal B}}
\def\cC{{\cal C}}
\def\cD{{\cal D}}
\def\cE{{\cal E}}
\def\cF{{\cal F}}
\def\cG{{\cal G}}
\def\cH{{\cal H}}
\def\cI{{\cal I}}
\def\cJ{{\cal J}}
\def\cK{{\cal K}}
\def\cL{{\cal L}}
\def\cM{{\cal M}}
\def\cN{{\cal N}}
\def\cO{{\cal O}}
\def\cP{{\cal P}}
\def\cQ{{\cal Q}}
\def\cR{{\cal R}}
\def\cS{{\cal S}}
\def\cT{{\cal T}}
\def\cU{{\cal U}}
\def\cV{{\cal V}}
\def\cW{{\cal W}}
\def\cX{{\cal X}}
\def\cY{{\cal Y}}
\def\cZ{{\cal Z}}

\def\dbB{\mathbb{B}}
\def\dbC{\mathbb{C}}
\def\dbD{\mathbb{D}}
\def\dbE{\mathbb{E}}
\def\dbF{\mathbb{F}}
\def\dbG{\mathbb{G}}
\def\dbH{\mathbb{H}}
\def\dbI{\mathbb{I}}
\def\dbJ{\mathbb{J}}
\def\dbK{\mathbb{K}}
\def\dbL{\mathbb{L}}
\def\dbM{\mathbb{M}}
\def\dbN{\mathbb{N}}
\def\dbP{\mathbb{P}}
\def\dbR{\mathbb{R}}
\def\dbS{\mathbb{S}}
\def\dbT{\mathbb{T}}
\def\dbQ{\mathbb{Q}}
\def\dbZ{\mathbb{Z}}

\medskip

\section{Introduction}\label{S:1}

The celebrated Kyle model, introduced in Kyle's pioneering work in \citeyear{K85}, provides a foundational framework for understanding strategic informed trading and price impact. In its canonical form, the model depicts a strategic game between a privately informed trader and a competitive market maker, with noise traders -- modeled as a single representative agent -- acting as exogenous participants. In particular, the insider exploits the presence of noise traders and strategically trades a risky asset, taking advantage of his private information about its value to maximize expected gains. The resulting equilibrium carries important implications on the balance between profit and concealment, establishing a formal link between market liquidity and information asymmetry.

Following Back's (\citeyear{B92}) seminal extension of the Kyle model to a continuous-time framework, or what is now also known as the Kyle--Back model (see, e.g., Ma, Sun, and Zhou \citeyear{MSZ18}), a substantial body of theoretical and empirical literature has emerged and provided significant generalizations to explore rich market implications of the original model. For instance, these include dynamic asymmetric information that evolves over time (see, e.g., Back and Peterson \citeyear{BP98}, Back, Crotty, and Li \citeyear{BCL18}, \c{C}etin \citeyear{C18}, Ma and Tan \citeyear{MT22}, and Ekren, Mostowski, and \v{Z}itkovi\'{c} \citeyear{EMZ25}), risk-averse preferences of the insider (see Bose and Ekren \citeyear{BE23}, \citeyear{BE24}), a random deadline for publicly announcing the asset value (\text{cf.} Back and Baruch \citeyear{BB04}, Caldentey and Stacchetti \citeyear{CS10}, \c{C}etin \citeyear{C18}, Corcuera, Di Nunno, and Fajardo \citeyear{CDNF19}, and Qiu and Zhou \citeyear{QZ25b}), and the consideration of stochastic liquidity (see, e.g., Collin-Dufresne and Fos \citeyear{C-DF16}, and Collin-Dufresne, Fos, and Muravyev \citeyear{C-DFM21}, and Ekren, Mostowski, and \v{Z}itkovi\'{c} \citeyear{EMZ25}) or long-short memory (\text{cf.} Biagini \text{et al.} \citeyear{BHM-BO12} and Yang, He, and Huang \citeyear{YHH21}).

Notwithstanding the interesting extensions noted above, in the present paper we examine the insider trading problem within the Kyle--Back model framework featuring \textsl{dynamic legal risk}. Notably, Kyle-type models with legal risk have become a topic of considerable interest in recent literature, which account for the (illegal) insider's motive for mitigating legal penalties in addition to maximizing profit. Such legal penalties are paramount to addressing the fundamental difference between illegal insider trading and normal, legal informed trading\footnote{The Securities and Exchange Commission (SEC) clearly distinguishes between illegal insider trading and the legal reporting of transactions by corporate insiders; see \href{https://www.investor.gov/introduction-investing/investing-basics/glossary/insider-trading}{this Fast Answers page}.}
and can be regarded as additional advantageous selection costs associated with insider trading. Existing works have provided useful insights into the optimal design of regulatory actions (DeMarzo, Fishman, and Hagerty \citeyear{DFH98} and Carr\'{e}, Collin-Dufresne, and Gabriel \citeyear{CC-DG22}) and the deterrent effect of legal enforcement on insider trading (Frino \text{et al.} \citeyear{FSWZ13}, Kacperczyk and Pagnotta \citeyear{KP24}, \c{C}etin \citeyear{C25}, and Ma, Xia, and Zhang \citeyear{MXZ25}). In particular, as the first to incorporate legal risk into a continuous-time Kyle model, \c{C}etin (\citeyear{C25}) revealed the suboptimality for the insider to align the market price of the risky asset with its fundamental value by the terminal time.

More specifically, consideration of legal risk introduces an additional agent to the standard model setting, namely an exogenous regulator. While the noise traders remain liquidity providers in the market, the regulator launches investigations into the (illegal) insider's trading behavior with some generally state-dependent probability. Then, depending on his actual trading strategies and illicit profit, the insider is successfully prosecuted with a certain (conditional) probability. In this framework, the insider internalizes legal risk into his objective function by contemplating an endogenous penalty cost (advantageous selection) over his expected gains from trade (Kacperczyk and Pagnotta \citeyear{KP24} and Ma, Xia, and Zhang \citeyear{MXZ25}). Such strategic conduct has profound implications for the insider's optimal strategy, for the threat of prosecution with severe consequences readily precludes arbitrarily large trades.

In illegal insider trading, a recently recognized and important consideration is ``stealth trading,'' a strategy by which the insider conceals his identity within a large quantity of surrounding trades, concentrating on medium-sized trades, leading to the so-called ``camouflage effect.'' In this regard, Ma, Xia, and Zhang (\citeyear{MXZ25}) are the first to establish a Kyle-type model framework to coalesce stealth trading into insider trading for an examination of the intrinsic link between the insider's trade size choices and the liquidity provided by noise traders, quantifying the scale of such trade sizes by a stealth index ($\gamma$). Through comprehensive empirical experiments, the study also confirmed implications of the stealth trading hypothesis (Barclay and Warner \citeyear{BW93} and Chakravarty \citeyear{C01}), trade size clustering (\text{cf.} Chang, Pinegar, and Schachter \citeyear{CPS97}, Alexander and Peterson \citeyear{AC07}, and Fei and Xia \citeyear{FX24}), and diminished price impact due to the legal deterrence (\text{cf.} Patel and Put\c{n}in\v{s} \citeyear{PP21} and Kacperczyk and Pagnotta \citeyear{KP24}).

In the present paper, we shall integrate stealth trading into the analysis of dynamic legal risk as well. Such an operation not only easily justifies the use of a Gaussian process\footnote{In the existing literature, a Brownian motion is often adopted for simplicity (see, e.g., Bose and Ekren \citeyear{BE23} and \c{C}etin \citeyear{C25}). The use of a general Gaussian process presupposes that noise traders trade independently -- without systematic interaction -- while allowing for time-dependent variations in their trades.} for modeling noise traders' order flow (in light of the Central Limit Theorem) but also helps simplify the equilibrium structure significantly. The limiting equilibrium, which exhibits negligible price impact, is well-suited to examine how the insider balances the competing objectives of profit maximization and concealment to avoid detection and prosecution. In this limiting case, somewhat remarkably, the corresponding game problem ``deforms'' into a deterministic control problem, enabling us to establish solutions with standard methods. The basic intuition is that as price impact diminishes, the insider seeks to optimize his legal risk-adjusted expected profit in an intrinsically certain, ``risk-free'' manner, in which penalties are still internalized, albeit with no explicit dependence on the market price of the asset.

On a technical level, our approach makes use of the recently proposed ``weak formulation'' for the Kyle--Back model (Qiao and Zhang \citeyear{QZ25a}), which transforms the net order flow into an exogenous diffusion term by means of a measure change, greatly facilitating subsequent equilibrium analysis and permitting general path-dependence properties. The validity of this formulation relies on the presence of a cost functional, which can be directly connected to legal penalties. In this paper, we shall refer to it as an ``alternative formulation'' due to its equivalence to the standard formulation when the insider's strategies are square-integrable in time -- the latter property governed by legal risk.

Noticeably, \c{C}etin (\citeyear{C25}) restricts legal risk to the terminal horizon, meaning that detection (and prosecution) of the insider can only happen at a single point in time, in spite of continuous trading. While this assumption is inevitable in a one-period framework (as done in Carr\'{e}, Collin-Dufresne, and Gabriel \citeyear{CC-DG22} and Ma, Xia, and Zhang \citeyear{MXZ25}), it presents a notable limitation in continuous time, where regulatory investigation and trading occurs concurrently (e.g., compare the two-period model in Kacperczyk and Pagnotta \citeyear{KP24}; see also Frino \text{et al.} \citeyear{FSWZ13}), adjusting dynamically for each other. Just as an example, if the insider executes an extremely aggressive trade before the terminal time, regulators can be alerted well before public announcement of the asset price, preemptively halting further trades by the insider.

To account for legal risk in a dynamic fashion, in this paper we propose an intensity-based approach, inspired from the credit modeling literature (see, e.g., Gourieroux, Monfort, and Polimenis \citeyear{GMP06}). In essence, this approach specifies an intensity process that governs the regulator's investigation scheme, subject to temporal variation. Comprehensibly, the intensity process should depend on the insider's concurrent trades and should increase as the insider begins to trade more aggressively. Besides, it also varies with the population size of noise traders because the latter directly complicates the investigation process.\footnote{We refer to Ma, Xia, and Zhang (\citeyear{MXZ25}) \text{Sect.} 1 for a detailed discussion on this aspect.} The intensity process structurally enables the introduction of a random prosecution time (a ``default'' time in a sense),\footnote{This random time feature is technically comparable to the random announcement deadline considered in the literature (\c{C}etin \citeyear{C18}, Corcuera, Di Nunno, and Fajardo \citeyear{CDNF19}, etc.), though it is tied to the structure of legal risk in this paper.} at which point the insider is barred from further trading, his private information immediately published, and the asset value fully disclosed.

Moreover, we come up with a partially convex bivariate function to encompass a broad spectrum of legal penalty structures. For example, this significantly extends the quadratic, purely strategy-dependent penalty used in \c{C}etin (\citeyear{C25}) \text{Sect.} 2. The resulting functional form is capable of covering very general penalty structures, as considered in Ma, Xia and Zhang (\citeyear{MXZ25}), including polynomial, or endogenous, profit-dependent penalties (\text{cf.} Shin \citeyear{S96}, Carr\'{e}, Collin-Dufresne, and Gabriel \citeyear{CC-DG22} \text{Sect.} 3.4, and Kacperczyk and Pagnotta \citeyear{KP24} \text{Sect.} II). More specifically, the bivariate function is constructed by applying an aggregation function to a civil penalty component and a criminal penalty component, offering sufficient generality under practical circumstances. Furthermore, with dynamic legal risk, penalty imposition must account for the entire history of trades rather than (necessarily) focusing on behavior at any single time point. To model this, we introduce an instantaneous penalty rate function featuring a concentration parameter gauging the penalty intensity -- more precisely, whether adjudication is distributed evenly over time or concentrated around distinctly aggressive trades.

With the above considerations, the resulting equilibria (limiting or not) are essentially path-dependent, leaning on the insider's entire trade history, which feature poses significant technical challenges. Using the aforementioned alternative formulation, we are able to partially characterize a finite-population equilibrium in terms of a coupled system of backward stochastic differential equations (BSDEs), hence warranting the derivation of necessary conditions. In the limiting case, the decoupling of the same system allows us to view it as deterministic optimization and subsequently analyze it in full capacity. Analysis of the limiting equilibria also recovers the stealth index choice concluded from Ma, Xia and Zhang (\citeyear{MXZ25}) (\text{Sect.} 3 and \text{Sect.} 5), with joint dependence on the detection mechanism and penalty composition.

Most important, the limiting equilibria that we are bound to construct, by leveraging the intrinsic nexus between (illegal) insider trading and stealth trading, accommodate dynamic legal investigations and facilitate an in-depth examination of their impact on informed trading strategies. Economically, it elucidates how specific detection mechanisms and penalty structures can induce time-varying -- and at times locally explosive -- trading behavior, further yielding a sequence of temporally rich outcomes. In more detail, while the stealth level of insiders is quantified in much the same fashion as in Ma, Xia, and Zhang (\citeyear{MXZ25}), in terms of the stealth index ($\gamma$), which is governed by three coefficients (specifically, $\beta$, $\eta$, and $\alpha$) linked to the population-size elasticity of regulatory efficacy and the penalization intensity for large trades, the temporal dimension is in large part dictated by a distinct set of parameters ($p$, $b$, and $c$). These additional parameters capture the historical concentration of penalties (whether regulators prioritize singular aggressive trades or an average measure of misconduct) as well as the penalty multiplier (applied directly to civil penalties and indirectly to criminal ones) as implemented under the \textsc{Insider Trading and Securities Fraud Enforcement Act of 1988} (ITSFEA). Understandably, such nuanced analysis would not be feasible within a static or terminal horizon-based framework, as it necessitates explicit modeling of the temporal variation in legal enforcement mechanisms.

In this aspect, it is also critical to recognize that for an (illegal) insider, legal risk concerns should far outweigh price impact. Indeed, by exploiting considerable noise trading and adequately ``camouflaging,'' the insider can associate trades with significantly reduced price impact (Kacperczyk and Pagnotta \citeyear{KP24} and Ma, Xia, and Zhang \citeyear{MXZ25}); however, legal risk remains an irreducible and non-diversifiable component, as our results demonstrate.\footnote{On the other hand, price impact can well arise from legal insider trading where properly reported transactions (via SEC Form 4), not based on material non-public information, contribute legitimately to price discovery. This scenario lies outside the scope of our study, which focuses on the legal risk associated with illicit trading.} This highlights a central benefit of the mentioned limiting equilibria, along with four key insights. First, under committed investigation effort, the insider optimally employs a time-varying function of the price discrepancy, which is not necessarily a constant multiple of the latter (in contrast with a key finding of \c{C}etin \citeyear{C25}). Second, as the terminal horizon approaches, the diminishing scope for legal action provides the insider with an incentive to maintain -- or even intensify, significantly -- trading aggressiveness, which outcome is comparable to Back (\citeyear{B92}), Caldentey and Stacchetti (\citeyear{CS10}), and \c{C}etin (\citeyear{C18}) in the absence of legal risk. Third, na\"{\i}vely escalating penalties cannot substitute for regulatory diligence in that a decline in enforcement incentivizes more aggressive trading; this is at bottom to raise the cost of advantageous selection without satisfactorily countering the insider's speculative calculus when perceived detection risk is low. Fourth, strategy-focused criminal penalties are much more efficacious in deterring insider trading than civil charges alone, for the latter create scope for varied temporal trade allocations, which will be illustrated by establishing the existence of infinitely many equilibria, whereas the former yields a uniquely constrained equilibrium.

In a nutshell, the present paper contributes to the existing literature on insider trading and private information through three major innovations. First, we develop a Kyle--Back-type model framework for studying insider trading with legal risk, where legal risk can manifest concurrently with insider trading. Such dynamism is paramount to considering profit-dependent penalties in continuous time, significantly overcoming the limitations of strategy-focused penalties typically reserved for criminal cases (under the ITSFEA). In particular, as noted in \c{C}etin (\citeyear{C25}), profit-dependent penalization is infeasible within a static detection framework, which assumes that detection occurs at a single point in time. However, after introducing a dynamic investigation scheme, we shall see that such a penalty structure is not only implementable but also yields significant economic insights -- primarily offering a flexible and responsive deterrent mechanism for insider trading. In this sense, this paper provides a dynamic extension of the one-period and two-period models studied in Carr\'{e}, Collin-Dufresne, and Gabriel (\citeyear{CC-DG22}) and Kacperczyk and Pagnotta (\citeyear{KP24}), respectively, while completely covering the choice of penalty functions therein.

Second, the aforementioned alternative formulation, while standard in the stochastic control literature (including Qiao and Zhang \citeyear{QZ25a}), is a rather newfangled technique in the economics literature, and we expect this paper to provide a meticulous demonstration of how this reformulation works to facilitate equilibrium analysis. Notably, with its compatibility with the consideration of legal risk, such a reformulation leads to the economic interpretation of the underlying change of measure as ``impact neutralization,'' under which the insider fully accounts for market expectations of concurrent order flows.

Third, as noted before, we shall show that under stealth trading, the equilibrium greatly simplifies into an inherently deterministic form, a feature in turn made transparent by the alternative formulation. This deterministic form conveys essentially the same messages about an insider's inter-temporal trading behaviors in the presence of legal risk, due to internalization of the endogenous legal risk. In slightly more technical terms, the limiting equilibrium is an $\epsilon$-equilibrium in any finite-population setting, the fidelity of which approximation is further bounded by the design of the regulatory investigation and penalties, offering practical guidance for policymakers while preserving analytical tractability.

The rest of this paper is organized as follows. Section \ref{S:2} formulates the insider trading problem with legal risk, explaining how to model dynamic legal risk in continuous time, subject to flexible detection mechanisms and penalty structures. In particular, Section \ref{S:2.3} presents the alternative formulation and details how it leads to an equivalent definition of the standard equilibrium, while Section \ref{S:2.4} formalizes the notion of the limiting equilibrium under stealth trading. Section \ref{S:3} provides characterizations of the standard, finite-population equilibrium and the associated limiting equilibrium, relying on the alternative formulation. Section \ref{S:4} then analyzes several explicitly solvable scenarios within the limiting framework, unveiling profound connections between equilibrium trading strategies and key elements of regulatory design, in conjunction with additional insights for regulatory policy. Conclusions are drawn in Section \ref{S:5}, and all mathematical proofs are given in Appendices \ref{A}, \ref{B}, and \ref{C}.

\medskip

\section{Problem formulation}\label{S:2}

We introduce the insider trading problem following the standard Kyle--Back model setting (e.g., Back \citeyear{B92} and Ma, Sun, and Zhou \citeyear{MSZ18}). Let $(\Omega,\mathcal{F},\PP;\mathbb{F}\equiv(\mathscr{F}_{t})_{t\in[0,T]})$ (with $\mathscr{F}_{T}=\mathcal{F}$, $T>0$ fixed) be a filtered probability space satisfying the usual conditions, which supports the uncertainty structure of a market that offers to trade a single risky asset. Trading takes place in continuous time $t\in[0,T]$, and the risk-free interest rate is exogenously set to $0$. The fundamental value $V$ of the risky asset will be publicly announced at time $T$ to be $v\in\mathcal{V}:=\supp(V)\subseteq\mathbb{R}$, hence $\mathcal{F}$-measurable, and it has a finite variance. The market price of the asset is denoted by $S$, to be determined by equilibrium analysis. There are three types of economic agents participating in trading the asset.

\begin{itemize}
  \item{
  A risk-neutral \textsl{insider (trader)} (\textbf{IT}) who fully accesses $V$ at time $0$ at no cost and trades strategically with a cumulative order $\Theta\equiv(\Theta_{t})$, where $\Theta_{t}=\Theta(V;t,S_{[0,t]})$, $t\in[0,T]$;\footnote{We use the notation $\jmath_{[0,t]}$ as an abbreviation of $(\jmath_{s})_{s\in[0,t]}$ for a generic stochastic process $\jmath$.}
  }
  \item{
  A total of $N$ (almost) identical \textsl{noise traders} (\textbf{NT}s) who are unaware of either $V$ or $S$, while trading independently from one another and non-strategically, with a total cumulative order $\sqrt{N}\int^{t}_{0}\sigma_{s}\dd B_{s}$, $t\in[0,T]$, where $\sigma:\mathbb{R}_{+}\mapsto[\underline{\sigma},\infty)$ is a deterministic function, with $\underline{\sigma}>0$, and $B\equiv(B_{t})$ is an $\mathbb{F}$-standard Brownian motion that is independent of $V$;
  }
  \item{
  A \textsl{market maker} (\textbf{MM}) who only observes the aggregate net order flow
  \begin{equation}\label{Q.def}
    Q_{t}:=\sqrt{N}\int^{t}_{0}\sigma_{s}\dd B_{s}+\Theta_{t},\quad t\in[0,T],
  \end{equation}
  and sets a pricing rule $P$ to determine the market price $S_{t}=P(t,Q_{[0,t]})$.
  }
\end{itemize}

By using the pricing rule, the IT's cumulative order can be alternatively written as a function of the aggregate order, namely
\begin{equation*}
  \Theta_{t}=\Theta(V;t,Q_{[0,t]}),\quad t\in[0,T],
\end{equation*}
which is valid either when the correspondence between $S$ and $Q$ is one-to-one or when the IT can also access $Q$ (see, e.g., Qiao and Zhang \citeyear{QZ25a} \text{Sect.} 2.1). These assumptions are nonrestrictive -- technical verifiability aside, for most securities, as trading volume is public by design, it is not difficult for the IT to construct dynamic estimates of $Q$. Another important feature is the population size $N$ of NTs, which is essential for incorporating stealth trading and for justifying the finite-dimensional Gaussian distribution of the process $\int^{\cdot}_{0}\sigma_{s}\dd B_{s}$ (Ma, Xia, and Zhang \citeyear{MXZ25} \text{Sect.} 2) tracking the cumulative order of an average NT. This process can generally differ from a scaled Brownian motion to allow for possible serial dependence in normal trades (\text{cf.} Collin-Dufresne and Fos \citeyear{C-DF16} \text{Sect.} 2, Ma, Sun, and Zhou \citeyear{MSZ18} \text{Sect.} 2, and Ekren, Mostowski, and \v{Z}itkovi\'{c} \citeyear{EMZ25} \text{Sect.} 2.2), with the deterministic function $\sigma$ measuring the \textsl{noise trading intensity};\footnote{The assumption that $\sigma$ is deterministic basically ensures that the resulting stochastic integral is Gaussian, aligning with the Fredholm representation of Gaussian processes under rather weak conditions on their covariance functions; see, e.g., Sottinen and Viitasaari (\citeyear{SV16}). On the other hand, with stealth trading, even when individual trades can have stochastic volatility ($\sigma$), their aggregate determinism persists for $N$ sufficient large.} we refer to Collin-Dufresne and Fos (\citeyear{C-DF15}) for related empirical evidence.

Under insider trading penalty consideration, we assume that the IT's cumulative order is absolutely continuous in time (\text{cf.} Back \citeyear{B92} and \c{C}etin \citeyear{C25}):
\begin{equation}\label{th}
  \Theta_{t}=\int^{t}_{0}\theta_{s}\dd s=\int^{t}_{0}\theta(V;s,Q_{[0,s]})\dd s,
\end{equation}
where $\theta\equiv(\theta_{t})$ measures the \textsl{insider trading intensity}, which, when treated as a functional, specifies the \textsl{insider trading strategy}.

\medskip

\subsection{Dynamic legal risk}\label{S:2.1}

The IT faces legal risk with a \textsl{regulator} (\textbf{RG}) present in the market, and as much as trading occurs in continuous time, prosecution of insider trading is conducted in a \textsl{dynamic} manner. Specifically, let $\tau$ be the time when the RG successfully detects and prosecutes the IT. Adopting an intensity-based approach (Gourieroux, Monfort, and Polimenis \citeyear{GMP06}), we model $\tau$ as the first jump time of an $\mathbb{F}$-adapted doubly stochastic Poisson process (\text{a.k.a.} Cox process) $M_{\Lambda}\equiv(M_{\Lambda_{t}})$, i.e.,
\begin{equation}\label{tau.def}
  \tau:=\inf\{t\geq0:\;M_{\Lambda_{t}}=1\},
\end{equation}
where $\Lambda\equiv(\Lambda_{t})$ is the associated cumulative intensity process. We shall assume that $\Lambda$ has $\PP$-\text{a.s.} absolutely continuous sample paths, taking the form
\begin{equation}\label{hr}
  \Lambda_{t}=\int^{t}_{0}\lambda_{s}\dd s,\quad t\in[0,T].
\end{equation}
Then, for any $t>0$, conditional on $\Lambda_{t}$, $M_{\Lambda_{t}}$ is Poisson-distributed with parameter $\Lambda_{t}$, also independent from $B_{t}$ and $V$. The intensity process $\lambda\equiv(\lambda_{t})$, which represents surveillance intensity, is referred to as the \textsl{hazard rate} -- from the IT's perspective. As noted before, with the RG's access to the IT's brokerage account, $\lambda$ should depend on the IT's trade history $\theta_{[0,\cdot]}$, and the two should be directly related as well. Thus, we have the functional forms
\begin{equation}\label{de.s}
  \lambda_{t}=\lambda(t,N^{-\beta}\theta_{[0,t]}),\quad\Lambda_{t}=\int^{t}_{0}\lambda(s,N^{-\beta}\theta_{[0,s]})\dd s,\quad t\in[0,T],
\end{equation}
where $\theta$ is as shown in (\ref{th}). Importantly, the term $N^{-\beta}$, for a power coefficient $\beta\geq0$, reflects how the population size of the NTs can obscure detection of insider trading. By inflating trading noise, a large population can significantly complicate regulatory actions to differentiate insider trading activity from normal activity (Ma, Xia, and Zhang \citeyear{MXZ25} \text{Sect.} 2.1), which are typically conducted on a case-by-case basis.\footnote{In this respect, the complication is on account of the number $N$ of NTs rather than their average trading intensity $\sigma$. Indeed, considering an extreme scenario with only one highly capitalized NT who trades aggressively, then identifying the IT becomes rather straightforward, as there are only two trading entities to analyze.} The larger the power coefficient $\beta\geq0$, the higher level of stealth the IT can acquire.

At time $T$, if the IT has \textsl{not yet} been detected, then his gains from trade are equal to
\begin{equation*}
  \Theta_{T}(V-S_{T-})+\int^{T-}_{0}\Theta_{t}\dd S_{t}=\int^{T}_{0}\theta_{t}(V-S_{t})\dd t,
\end{equation*}
conditional on his trading strategy $\theta$ from (\ref{th}). However, once detected and prosecuted (at time $\tau\leq T$), the IT will not only lose all of his profit due to disgorgement but is also subject to legal penalties, inherently a form of advantageous selection costs, given by some penalty function $\Pi$.

As considered in Ma, Xia, and Zhang (\citeyear{MXZ25}) \text{Sect.} 2.2, we adopt a general composition, under which $\Pi$ depends on both the IT's trade history $\theta_{[0,\tau]\!]}$ and his profit gained (or losses avoided) $\int^{\tau}_{0}(V-S_{s})\theta_{s}\dd s$, corresponding to a criminal penalty channel and a civil penalty channel, respectively. As aforementioned, in continuous time, we postulate that in such a situation, the fundamental value $V$ becomes public information immediately following successful prosecution, with trading of the asset halted (leastways for the IT) immediately after time $\tau$.

\begin{remark}\label{rem:1}
In this standard formulation, the market filtration $\mathbb{F}$ can be viewed as the $\PP$-augmented natural filtration generated by $B$ and $M_{\Lambda}$ subject to an enlargement by $\upsigma(V)$ at the time of disclosing $V$. More specifically, the natural filtration of $M_{\Lambda}$ is understood to be the time-stopped one, $(\upsigma(M_{\Lambda_{[0,t]}}))_{t\in[0,T]}$, and is exactly the RG's filtration, which satisfies the usual conditions as long as $\Lambda$ is continuous, as guaranteed by (\ref{hr}). The IT's filtration, under suitable conditions on the pricing rule $P$ (see \c{C}etin \citeyear{C25} \text{Sect.} 2 and also (\ref{PL2})), can be taken as $(\upsigma(\mathscr{F}_{t}\cup\upsigma(V)))_{t\in[0,T]}$ (by observing $V$). It is also clear that $\tau$ is an $\mathbb{F}$-stopping time, and so $V$ is further measurable with respect to the stopping-time $\upsigma$-field $\mathscr{F}_{\tau\wedge T}$. The MM's filtration is that generated by observed $Q$ alone, as he cannot distinguish between orders coming from the NTs ($\sqrt{N}\int^{\cdot}_{0}\sigma_{s}\dd B_{s}$) and those from the IT ($\Theta$).
\end{remark}

A general form of the total penalty imposed at time $\tau$ of prosecution is
\begin{align}\label{pe.s1}
  \Pi_{\tau}&=\Pi\bigg(\theta_{[0,\tau]\!]},\int^{\tau}_{0}\theta_{s}(V-S_{s})\dd s\bigg)\mathds{1}_{\{\tau\leq T\}} \nonumber\\
  &=\bigg(\int^{\tau}_{0}\theta_{s}(V-S_{s})\dd s+\Pi_{\rm a}\bigg(\theta_{[0,\tau]\!]},\int^{\tau}_{0}\theta_{s}(V-S_{s})\dd s\bigg)\bigg)\mathds{1}_{\{\tau\leq T\}}.
\end{align}
In the second equality of (\ref{pe.s1}), the first component (given nonnegativity) signifies disgorgement of the IT's illegal profit, ensuring unprofitability, while $\Pi_{\rm a}$ is a functional covering additional penalties, which could contain (strategy-dependent) criminal charges through the first argument and (profit-dependent) civil charges through the second. Intuitively, as a bivariate functional, $\Pi_{\rm a}$ (and hence $\Pi$) should be increasing (not necessarily strictly) in both arguments.

\begin{remark}\label{rem:2}
The penalty structure (\ref{pe.s1}) encompasses all penalty types (linear, quadratic, profit-multiple, etc.) that have been discussed in the literature, and it is beneficial to use an aggregation function to combine (additional) criminal and civil penalties,
\begin{equation}\label{pe.s2}
  \Pi_{\rm a}\bigg(\theta_{[0,\tau]\!]},\int^{\tau}_{0}\theta_{s}(V-S_{s})\dd s\bigg)=W\bigg(\Pi_{0}(\theta_{[0,\tau]\!]}),c\int^{\tau}_{0}\theta_{s}(V-S_{s})\dd s\bigg),
\end{equation}
where $\Pi_{0}$ is the \textsl{criminal penalty functional} and $c\geq0$ is a \textsl{penalty multiplier} attached to the IT's illegal profit. The aggregation function provides a flexible mechanism for composing penalties based on the actual circumstances -- common choices include summation ($W(x_1,x_2)=x_1+\max\{x_2,0\}$), product ($W(x_1,x_2)=x_1\max\{x_2,0\}$), and maximum ($W(x_1,x_2)=\max\{x_1,x_2,0\}$).
\end{remark}

\medskip

\subsection{Profit maximization and equilibrium}\label{S:2.2}

Knowing the fundamental value $V=v$, the IT takes the pricing rule $P$ from the MM as given and designs a trading strategy $\theta^{v}$ to maximize his expected profit from trade. With $S=P(\cdot,Q_{[0,\cdot]})$, the IT's objective function can be written as
\begin{equation}\label{IT.o1}
  J(P;\theta^{v},v):=\E\bigg[\int^{\tau\wedge T}_{0}\theta^{v}_{t}(v-P(t,Q^{\theta^{v}}_{[0,t]}))\dd t-\mathds{1}_{\{\tau\leq T\}}\Pi\bigg(\theta^{v}_{[0,\tau]\!]},\int^{\tau}_{0}\theta^{v}_{t}(v-P(t,Q^{\theta^{v}}_{[0,t]}))\dd t\bigg)\bigg].
\end{equation}
Inside the expectation in (\ref{IT.o1}), the first component represents the IT's illegal profit conditional on not being detected or prosecuted, and the second component corresponds to the total penalty upon successful prosecution. Besides, (\ref{Q.def}) and (\ref{th}) imply that for each $\theta^{v}_{t}=\theta(v;t,Q_{[0,t]})$, $Q^{\theta^{v}}$ is determined by the relation
\begin{equation}\label{Q.e1}
  Q^{\theta^{v}}_{t}=\sqrt{N}\int^{t}_{0}\sigma_{s}\dd B_{s}+\int^{t}_{0}\theta(v;s,Q^{\theta^{v}}_{[0,s]})\dd s,\quad t\in[0,T].
\end{equation}
The notations $Q^{\theta^{v}}_{[0,t]}$ and $\theta^{v}_{[0,t]}$ are understood similarly; see Footnote 6.

\begin{proposition}\label{pro:1}
Given $V=v$, the IT's objective function in (\ref{IT.o1}) can be recast as
\begin{equation}\label{IT.o2}
  J(P;\theta^{v},v)=\E\bigg[\int^{T}_{0}e^{-\Lambda^{\theta^{v}}_{t}}\theta^{v}_{t}(v-P(t,Q^{\theta^{v}}_{[0,t]}))\dd t-\int^{T}_{0}\lambda^{\theta^{v}}_{t}e^{-\Lambda^{\theta^{v}}_{t}} \Pi_{\rm a}\bigg(\theta^{v}_{[0,t]},\int^{t}_{0}\theta^{v}_{s}(v-P(s,Q^{\theta^{v}}_{[0,s]}))\dd s\bigg)\dd t\bigg],
\end{equation}
where $\lambda^{\theta^{v}}$ and $\Lambda^{\theta^{v}}$ are as shown in (\ref{de.s}) with $\theta=\theta^{v}$.
\end{proposition}

Proposition \ref{pro:1} shows that the IT's illegal profit and total penalty can be rewritten as their probability-weighted averages over the entire trading horizon $[0,T]$. For $t\in[0,T]$ and $v\in\mathcal{V}$ given, as the IT follows a trading strategy $\theta^{v}$, his instantaneous survival probability is exactly equal to $e^{-\Lambda^{\theta^v}_{t}}$, inversely related to the hazard rate $\lambda^{\theta^{v}}$; equivalently, he is exposed to a chance $1-e^{-\Lambda^{\theta^v}_{t}}$ of being detected and prosecuted.

\begin{remark}\label{rem:3}
Based on (\ref{IT.o2}) and (\ref{pe.s2}), we recover the scenario studied in \c{C}etin (\citeyear{C25}) by taking $T=1$ and $\Lambda_{t}=-\log(1-p)\mathds{1}_{\{t=T\}}$, $t\in[0,1]$, where $p\in(0,1)$, with $N=1$, $W=\mathrm{id}$, and $\Pi_{0}(\theta_{[0,t]})=\int^{t}_{0}\theta^{2}_{s}\dd s$. Then, detection / prosecution of insider trading only happens at time $1$ with a constant probability $p$ and quadratic, purely strategy-dependent penalties. In this case, the process $\Lambda$ is not continuous but deterministic, which still ensures that $\mathbb{F}$ satisfies the usual conditions. Going forward we mainly work with (\ref{IT.o2}) for writing BSDEs.
\end{remark}

At the same time, given the IT's trading strategy $\theta$ (see \eqref{th}), the MM sets a rational pricing rule $P^{\theta}$ for break-even (see Kyle \citeyear{K85} \text{Sect.} 2), i.e.,
\begin{equation}\label{MM.p1}
  P^{\theta}(t,Q^{\theta^{V}}_{[0,t]})=\E[V|\upsigma(Q^{\theta^{V}}_{[0,t]})],\quad t\in[0,T],
\end{equation}
where recall that $\theta^{V}=\theta(V;t,Q_{[0,t]})$.

\begin{definition}\label{eq.def1}
For any fixed $N\geq1$, an equilibrium is a pair $(\theta^{\star}_{N},P^{\star}_{N})$ of trading strategy and pricing rule such that:\smallskip\\
(i) Given $P^{\star}_{N}$, $\theta^{\star}_{N}$ maximizes the IT's expected profit (\ref{IT.o2}), i.e.,
\begin{equation*}
  J(P^{\star}_{N};\theta^{\star v}_{N},v)=\sup_{\theta^{v}}J(P^{\star}_{N};\theta^{v},v),\quad v\in\mathcal{V};
\end{equation*}
(ii) Given $\theta^{\star}_{N}$, the MM's pricing rule is rational, with $P^{\star}_{N}=P^{\theta^{\star}_{N}}$ as in (\ref{MM.p1}).
\end{definition}

\medskip

\subsection{An alternative formulation}\label{S:2.3}

Much of the existing literature on Kyle--Back models follows a dynamic programming (or PDE) approach and relies heavily on certain Markovian structures of the market price process $S$ (e.g., Back \citeyear{B92}). However, as noted in Qiao and Zhang (\citeyear{QZ25a}), these properties do not necessarily hold for general equilibria, which underlays our pursuit of an alternative formulation. In spirit, this formulation modifies the standard formulation in Section \ref{S:2.2} in such a way that the aggregate order flow $Q$ becomes exogenous and no longer depends on $\Theta$. The central idea is to construct an equivalent \textsl{impact-neutral measure} (see (\ref{mc})) under which the IT's price impact is fully absorbed into market expectations and, in doing so, allows his trading strategy to be specified solely in terms of the order flow $\sqrt{N}\int^{\cdot}_{0}\sigma_{s}\dd B_{s}$ from the NTs.

In the presence of legal risk, we first reasonably assume that the trading strategies are exponentially square-integrable regardless of the population size $N$, apart from ensuring that the objective function in \eqref{IT.o2} is always well-defined. % More precisely, we define the set of admissible trading strategies (admissibility set) for the IT as: \mathbb{L}^{2}_{\wp_{\mathbb{F}}}(\Omega\times[0,T];\mathbb{R})$

\begin{definition}\label{as.def}
A trading strategy $\theta$ as in (\ref{th}) is admissible if it satisfies: \smallskip\\
(i) Novikov's condition,
\begin{equation}\label{wf.adm}
  \E\big[e^{1/2\int_0^T(\th_t^v)^2\dd t}\big]<\infty,\quad v\in\mathcal{V},
\end{equation}
with $\th_t^v=\th\big(v;t,\big(\sqrt{N}\int^{\cdot}_{0}\sigma_{s}\dd B_{s}\big)_{[0,t]}\big)$; \\
(ii) square-integrable total penalty,
\begin{equation}\label{fL2}
  \E\bigg[\int_0^T\Pi_{\rm a}\bigg(\theta_{[0,t]}^v,\int^{t}_{0}\theta_{s}^v(v-S_{s})\dd s\bigg)^2\dd t\bigg]<\infty,\quad v\in\mathcal{V};
\end{equation}
(iii) square-integrable pricing rule,
\begin{equation}\label{PL2}
  \E\bigg[\int_0^T\big(P(t,Q_{[0,t]}^{\th^V})\big)^2\dd t\bigg]<\infty,
\end{equation}
whenever the relation (\ref{Q.e1}) is determinant. \\
%for a constant $C_{N,v}$ depending only on $N$ and $v$
For any given $v\in\cV$, the admissibility set $\cA$ contains all $\theta^{v}$ such that $\theta$ is admissible.
\end{definition}

%Note that $\mathbb{F}$ is generally smaller than the IT's filtration under the standard formulation. %{\color{red} The first condition is for the admissibility of changing of measures which will be introduced right away. The second condition if for the well-posedness of the BSDE when applying the stochastic maximal principle after changing the measure as well as existence of the limiting model after taking $N$ to infinity. Thus in both condition $\theta$ is a function of path of $B$. Condition (iii) is to guaranteed the integrability of pricing rule with respect to $V$ in the original formulation, thus here $\th^V=\th(V;\cdot,Q_{[0,\cdot]})$. The integrability of pricing rule is weak formulation follows naturally.}
In Definition \ref{as.def}, conditions (i) and (ii), pertaining to the form $\th^v_t=\th\big(v;t,\big(\sqrt{N}\int^{\cdot}_{0}\sigma_{s}\dd B_{s}\big)_{[0,t]}\big)$, $v\in\mathcal{V}$ and $t\in[0,T]$, with $N\geq1$ fixed, are imposed to ensure the viability and well-posedness of the change of measure and the subsequent BSDEs (see Proposition \ref{pro:3}). Under these conditions, $\theta^{v}$ automatically has square-integrable values and is progressively measurable in both variables. In contrast, condition (iii) mainly concerns the square-integrability of the pricing rule with respect to $V$ under the standard formulation, hence there $\th^V=\th(V;\cdot,Q_{[0,\cdot]})$ as in \eqref{Q.e1}, from which its integrability under the alternative formulation follows easily. All three conditions in Definition \ref{as.def} are quite natural, with each admitting a direct parallel to \c{C}etin (\citeyear{C25}) \text{Def.} 2.2 in the context of quadratic penalties.
% Condition (i) is crucial for defining the exponential martingale which will be introduced right after, technical functionalities to guarantee the well-posedness of the characterization next aside. % Condition (ii) and (iii) are mainly to guarantee the well-posedness of the FBSDE system, with Condition (iii) additionally necessary for the well-posedness of the equilibrium pricing rule under the alternative formulation.
For every $\theta^{v}$ as in conditions (i) and (ii), let us introduce the density process
\begin{equation}\label{dp}
  X^{\theta^{v}}_{t}:=\mathcal{E}\bigg(\frac{1}{\sqrt{N}}\int^{\cdot}_{0}\frac{\theta^{v}_{s}}{\sigma_{s}}\dd B_{s}\bigg)_{t}=\mathcal{E}\bigg(\frac{1}{\sqrt{N}}\int^{\cdot}_{0}\frac{\theta\big(v;s,\big(\sqrt{N}\int^{\cdot}_{0}\sigma_{s}\dd B_{s}\big)\big)}{\sigma_{s}}\dd B_{s}\bigg)_{t},\quad t\in[0,T],
\end{equation}
where $\mathcal{E}$ is the Dol\'{e}ans--Dade exponential and which solves the SDE
\begin{equation*}
  X^{\theta^{v}}_{t}=1+\frac{1}{\sqrt{N}}\int^{t}_{0}\frac{\theta^{v}_{s}X^{\theta^{v}}_{s}}{\sigma_{s}}\dd B_{s}.
\end{equation*}

By (\ref{wf.adm}) and $\sigma\geq\underline{\sigma}>0$, we clearly have $\theta^{v}/\sigma\in\mathbb{L}^{2}_{\wp_{\mathbb{F}}}(\Omega\times[0,T];\mathbb{R})$, where $\wp_{\mathbb{F}}$ denotes the $\upsigma$-algebra of $\mathbb{F}$-progressively measurable processes, verifying Novikov's condition. Thus, the density process (\ref{dp}) induces an equivalent probability measure $\PP^{\theta^{v}}\sim\PP$ via the relation
\begin{equation}\label{mc}
  \PP^{\theta^{v}}[A|\mathscr{F}_{t}]=\int_{A}X^{\theta^{v}}_{t}(\omega)\dd\PP[\omega|\mathscr{F}_{t}],\quad A\in\mathcal{F},\;t\in[0,T],
\end{equation}
and the classical Cameron--Martin--Girsanov theorem implies that
\begin{equation}\label{th-sBm}
  B^{\theta^{v}}:=B-\int^{\cdot}_{0}\frac{\theta^{v}_{s}}{\sqrt{N}\sigma_{s}}\dd s
\end{equation}
is a standard Brownian motion under $\PP^{\theta^{v}}$. With (\ref{th-sBm}) and (\ref{Q.e1}), the distribution of the aggregate order flow under $\PP$ is the same as that of the NT's order flow under $\PP^{\theta^{v}}$, i.e.,
\begin{equation}\label{Q.e2}
  \mathfrak{L}_\dbP(Q^{\theta^{v}}_{t})=\mathfrak{L}_{\PP^{\theta^{v}}}\bigg(\sqrt{N}\int^{t}_{0}\sigma_{s}\dd B_{s}\bigg),\quad t\in[0,T],
\end{equation}
where the variable on the right-hand side does not involve $\theta^{v}$.

Economically, the interpretation of $\PP^{\th^v}$ from \eqref{mc} as an impact-neutral measure can be clarified in tandem with traditional measure changes in classical asset pricing theory. As shown in (\ref{Q.def}), the aggregate order flow $Q$ follows a diffusion process whose drift component is influenced by the IT's trades and whose volatility component is (primarily) driven by all the NTs' contemporaneous trades. The quotient $\theta/(\sqrt{N}\sigma)$ in constructing the density process $X$ in (\ref{dp}) represents an insider-to-noise trading intensity ratio, or essentially, an informativeness-to-variability ratio, and is analogous to a ``price of risk'' in terms of order flow imbalances. Then, $X$, acting as a stochastic discount factor, reflects the MM's relative expectations of the aggregate order flow under different trading behaviors of the IT. The IT's access to the fundamental value $V=v$ enables the equivalent measure $\PP^{\theta^{v}}$ to internalize and compensate for his price impact, effectively neutralizing the drift component of $Q$. We highlight that this type of transformation is novel to the economic literature on informed trading.

\begin{remark}\label{rem:4}
Compared with Remark \ref{rem:1}, under the alternative formulation, the MM's filtration can be viewed as the natural filtration of the $\PP$-standard Brownian motion $B$, or equivalently, the centered Gaussian process (under $\PP$), $\sqrt{N}\int^{\cdot}_{0}\sigma_{s}\dd B_{s}=Q^{\theta^{v}}$, with its canonical space considered fixed. Then, the IT's filtration is equivalent to $(\upsigma(\upsigma(B_{[0,t]},M_{\Lambda^{\theta^{v}}_{[0,t]}})\cup\upsigma(V)))_{t\in[0,T]}$, with the market filtration $(\upsigma(B_{[0,t]},M_{\Lambda^{\theta^{v}}_{[0,t]}}))_{t\in[0,T]}=\mathbb{F}$.
\end{remark}

The main consequence of the alternative formulation is concretized by the next proposition, which permits reexpressing the IT's profit as well as the MM's pricing rule under the impact-neutral measure $\PP^{\theta^{v}}$.

\begin{proposition}\label{pro:2}
The IT's objective function in (\ref{IT.o2}) is equivalent to\footnote{It is understood that $\theta^{v}_{t}=\theta\big(v;t,\big(\sqrt{N}\int^{\cdot}_{0}\sigma_{r}\dd B_{r}\big)_{[0,t]}\big)$ for $t\in[0,T]$.}
\begin{align}\label{IT.o4}
  J(P;\theta^{v},v)&=\E^{\theta^{v}}\bigg[\int^{T}_{0}e^{-\Lambda^{\theta^{v}}_{t}}\theta^{v}_{t}\bigg(v-P\bigg(t, \bigg(\sqrt{N}\int^{\cdot}_{0}\sigma_{s}\dd B_{s}\bigg)_{[0,t]}\bigg)\bigg)\dd t \nonumber\\
  &\quad-\int^{T}_{0}\lambda^{\theta^{v}}_{t}e^{-\Lambda^{\theta^{v}}_{t}} \Pi_{\rm a}\bigg(\theta^{v}_{[0,t]},\int^{t}_{0}\theta^{v}_{s}\bigg(v-P\bigg(s,\bigg(\sqrt{N}\int^{\cdot}_{0}\sigma_{r}\dd B_{r}\bigg)_{[0,s]}\bigg)\bigg)\dd s\bigg)\dd t\bigg],
\end{align}
for given $v\in\cV$ and $\theta^{v}\in\cA$ (Definition \ref{as.def}), where $\E^{\theta^{v}}$ denotes expectation under $\PP^{\theta^{v}}$. Moreover, the MM's pricing rule is (a.s.)
\begin{equation}\label{MM.p2}
  P^{\theta}_{t}=\frac{\int_{\mathcal{V}}vX^{\theta^{v}}_{t}\dd\PP[V\leq v]}{\int_{\mathcal{V}}X^{\theta^{v}}_{t}\dd\PP[V\leq v]},\quad t\in[0,T].
\end{equation}
\end{proposition}

Built upon (\ref{IT.o4}), the next proposition then provides a useful connection between the IT's objective function to a standard BSDE.

\begin{proposition}\label{pro:3}
For given $v\in\cV$ and $\theta^{v}\in\cA$, it holds that
\begin{equation*}
  J(P;\th^v,v)=Y_0^{\th^v},
\end{equation*}
for which the couple $(Y^{\th^v},Z^{\th^v})$ satisfies the following BSDE:
\begin{align}\label{BSDE1}
  Y_t^{\th^v}&=\int_t^T\bigg[\th_s^v(v-P_s^{\th})-\l_s^{\th^v}\Pi_{\rm a}\bigg(\th_{[0,s]}^v,\int_0^s\th_r^v(v-P_r^{\th})\dd r\bigg)-\l_s^{\th^v}Y_s^{\th^v}+\frac{\th^v_sZ_s^{\th^v}}{\sqrt{N}\si_s}\bigg]\dd s \nonumber\\
  &\quad-\int_t^TZ_s^{\th^v}\dd B_s,\quad t\in[0,T],
\end{align}
whose well-posedness is guaranteed by the three conditions in Definition \ref{as.def}.
\end{proposition}

In Proposition \ref{pro:3}, the $\mathbb{F}$-adapted processes $Y^{\th^v}$ and $Z^{\th^v}$ can be interpreted, respectively, as the IT's objective function in the time flow and its sensitivity with respect to the randomness source $B$ driving the NTs' cumulative order. We should note that in (\ref{BSDE1}), $Y^{\th^v}$ can involve (nonlinearly) nested integrals of $\th^v$, giving it general dependence on the paths of $\th^v$ (namely the IT's entire trade history). This path-dependence feature is rooted in the RG's prosecution dynamism and does not arise from the structure of the penalty function alone, which is in stark contrast to the scenario of static prosecution in \c{C}etin (\citeyear{C25}).

Also, as long as Novikov's condition (\ref{wf.adm}) is in force, the alternative formulation is able to cover \textsl{all} admissible strategies and coincides with the standard formulation; see also Qiao and Zhang (\citeyear{QZ25a}) Remark 2.5 (i). Thus, the equilibrium under the alternative formulation is an \textsl{equivalent equilibrium}, precisely defined below.

\begin{definition}\label{eq.def2}
For any fixed $N\geq1$, an equivalent equilibrium is a pair $(\theta^{\star}_{N},P^{\star}_{N})$ of trading strategy and pricing rule such that:\smallskip\\
(i) Given $P^{\star}_{N}$, $\theta^{\star}_{N}$ maximizes (\ref{IT.o4}) in the sense that
\begin{equation}\label{IT.op}
  J(P^{\star}_{N};\theta^{\star v}_{N},v)=\sup_{\theta^v\in\cA}J(P^{\star v}_{N};\theta^v,v),\quad v\in\cV;
\end{equation}
(ii) Given $\theta^{\star}_{N}$, $P^{\star}_{N}=P^{\theta^{\star}_{N}}$, as in (\ref{MM.p2}).
\end{definition}

From a technical perspective, the alternative formulation has also rendered the IT's profit maximization an open-loop optimization problem -- in the sense of optimization with respect to functions of time and the Brownian motion only -- as far as the equivalent equilibrium is concerned. Indeed, based on Definition \ref{eq.def1}, for each $v\in\cV$, $\th^{v}$ is originally postulated as a function of $Q_{[0,\cdot]}$, hence forming a closed loop. In comparison, in light of Proposition \ref{pro:3}, Definition \ref{eq.def2} turns to the interpretation $\th^{v}=\th\big(v;t,\big(\sqrt{N}\int_0^\cdot\si_s\dd B_s\big)_{[0,t]})$, $t\in[0,T]$, which forms an open loop as there is a one-to-one correspondence between such $\th^v$ and $\vartheta\in\cA$ with $\vartheta(t,B_{[0,t]})=\th\big(v;t,\big(\sqrt{N}\int_0^\cdot\si_s\dd B_s\big)_{[0,t]}\big)$.

\medskip

\subsection{Limiting equilibrium}\label{S:2.4}

According to Definition \ref{eq.def2}, solving the equilibrium for finite $N\geq1$ requires considering a coupled system for $(\theta^{\star}_N,P^{\star}_N)$, which can be written into a generally infinite-dimensional (depending on the cardinality of $\cV$) system of coupled forward-backward stochastic differential equations (FBSDEs) in the continuous-time setting; see Qiao and Zhang (\citeyear{QZ25a}) \text{Thm.} 3.2. Instead of revolving around such coupled systems, whose global analysis is rather difficult and which can well have multiple solutions (e.g., Ma and Yong \citeyear{MY91}), we place the emphasis on the limiting case $N\to\infty$, also in line with stealth trading. As mentioned in the one-period model of Ma, Xia, and Zhang (\citeyear{MXZ25}), this consideration leads to diminished price informativeness and is able to decouple $P^{\star}_N$ from $\theta^{\star}_N$, thus significantly enhancing the tractability of the resulting equilibrium, apart from strong empirical support (\text{cf.} Meulbroek \citeyear{M92} and Frino \text{et al.} \citeyear{FSWZ13}).

In more detail, the IT exercises stealth trading by appropriately scaling his trade size with respect to the population size, governed by a \textsl{stealth index} $\gamma$, i.e., $\tilde{\theta}=N^{-\gamma}\theta$. In equilibrium, the original strategy $\theta^{\star}_N$ is expected to have the asymptotic property $\lim_{N\to\infty}(\theta^{\star v}_{N}/(N^{\gamma}\tilde{\theta}^{\star v}))=1$ for any $v\in\cV$,
%\footnote{Note that such a limit satisfies the admissibility constraint in (\ref{wf.adm}) due to the dependency of $C_{N}$ on $N$.}
where $\tilde{\theta}^{\star}$ is a presumably existent limiting equilibrium strategy (with real values). The choice of $\gamma$ is mutually determined by the intensity of regulatory investigation ($\beta$) as well as the functional form of legal penalties ($\Pi$), which will be explained afterwards. Following Ma, Xia, and Zhang (\citeyear{MXZ25}) \text{Sect.} 3, for stealth trading, it is reasonable to restrict attention to the value range $[0,1/2)\ni\gamma$.

In what follows, let $\tilde{P}$ denote the corresponding limiting pricing rule of the MM, which shares the same properties of $P$ as shown in (\ref{IT.o4}). Similar properties of $\theta$ apply to $\tilde{\theta}$ as well, with
\begin{equation}\label{wf.adm.t}
  \tilde{\theta}^{v}\in\cA,\quad v\in\mathcal{V}.
\end{equation}
%where the constant $\tl C_{v}>0$ only depends on $v$.
Throughout this paper, the tilde symbol is reserved for scaled quantities. To formalize the notion of limiting equilibrium under the alternative formulation, note that with $\gamma<1/2$, putting $\theta^{v}=N^{\gamma}\tilde{\theta}^{v}$ for $v\in\mathcal{V}$, the limit of (\ref{dp}) is (a.s.)
\begin{equation}\label{dp.l}
  \lim_{N\to\infty}X^{N^{\gamma}\tilde{\theta}^{v}}_{t}=\lim_{N\to\infty}\mathcal{E}\bigg(N^{\gamma-1/2}\int^{\cdot}_{0} \frac{\tilde{\theta}^{v}_{s}}{\sigma_{s}}\dd B_{s}\bigg)_{t}=1,\quad t\in[0,T],
\end{equation}
which implies that $\PP^{N^{\gamma}\tilde{\theta}_{v}}\to\PP$ (strongly) as $N\to\infty$, i.e., the alternatively formulation coincides with the original standard formulation in the $N$-limit. In this case, by dominated convergence, the MM's pricing rule (\ref{MM.p2}) simplifies into a constant
\begin{equation}\label{MM.lp}
  \tilde{P}_{\gamma}=\E[V],\quad\PP\text{-a.s.}
\end{equation}
Thus, for any $\tilde{\theta}\in\cA$, the scaled objective function for the IT can be defined independently from the pricing rule, i.e., with (\ref{MM.lp}),
\begin{equation}\label{IT.lo}
  \tilde{J}_{\gamma}(\tilde{\theta}^v,v):=\lim_{N\to\infty}(N^{-\gamma}J(\tilde{P}_{\gamma};N^{\gamma}\tilde{\theta}^v,v)),\quad v\in\mathcal{V},
\end{equation}
provided that the limit exists. Here, note that for $t\in[0,T]$, $N^{\gamma}\tilde{\theta}^v_t=N^{\gamma}\tilde{\theta}(v;t,B_{[0,t]})$ is generally different from $\th^{v}_t=\th\big(v;t,\big(\sqrt{N}\int_0^\cdot\si_s\dd B_s\big)_{[0,t]})$.

To guarantee that the stealth index $\gamma$ is effective, in the sense that IT's stealth results from a sophisticated choice of trade sizes in lieu of a complete abstention from trading, it is important to also require that both the limiting equilibrium strategy and the associated objective function value be nonzero.
%Additionally, according to (\ref{wf.adm.t}), the limiting equilibrium strategy should not constantly attain its bounds, namely $\pm\tilde{C}_{v}$, with general dependence on $v$, as this would imply that the admissibility constraint is overly restrictive.
This leads us to the next definition for the limiting equilibrium paired with Definition \ref{eq.def2}.

\begin{definition}\label{leq.def}
A $\gamma$-limiting equilibrium with $\gamma\in[0,1/2)$ is a pair $(\tilde{\theta}^{\star}_{\gamma},\tilde{P}^{\star}_{\gamma})$ of (scaled) trading strategy and pricing rule such that:\footnote{The zero element $\mathbf{0}$ is understood in the sense of equivalence class for (\ref{wf.adm.t}), recalling Definition \ref{as.def}.} \smallskip\\
(i) $\tilde{\theta}^{\star}_{\gamma}$ maximizes (\ref{IT.lo}), i.e.,
\begin{equation}\label{IT.lop}
  \tilde{J}_{\gamma}(\tilde{\theta}^{\star v}_\g,v)=\sup_{\tilde{\theta}^v\in\cA\setminus\{\mathbf{0}\}}\tilde{J}_{\gamma}(\tilde{\theta}^v,v)\neq0,\quad v\in\mathcal{V};
\end{equation}
(ii) $\tilde{P}^{\star}_{\gamma}=\E[V]$ as in (\ref{MM.lp}).
\end{definition}

\begin{remark}\label{rem:5}
In the extreme case $\gamma=0$, the population size $N$ has no effect on the IT's trade size, with $\theta=\tilde{\theta}$, which is at the same level as an average NT. With no price impact ($\tilde{P}_{\gamma}=\E[V]$), insider trading achieves maximal stealth. On the other hand, by forcing $\gamma=1/2$, one has $X^{N^{\gamma}\tilde{\theta}^{v}}=X^{\theta^{v}}$ in (\ref{dp.l}), and $N$ has no practical effect on the alternative formulation. In this case, the IT trades on the same scale as all NTs combined (hence no stealth), with usual price impact, similar to what is considered in the existing literature (\text{cf.} Carr\'{e}, Collin-Dufresne, and Gabriel \citeyear{CC-DG22} and \c{C}etin \citeyear{C25}).
\end{remark}

Based on Definition \ref{leq.def}, for any $v\in\cV$, $\tl\th^v$ is decoupled from the (constant) pricing rule, and $\tl\th^v$ is treated as a functional of $B$ only (with no more dependence on $N$). Hence, the IT's profit maximization \eqref{IT.lop} for the $\gamma$-limiting equilibrium also leads to open-loop optimization, analogous to the case of Definition \ref{eq.def2} regarding the (finite-$N$) equivalent equilibrium.

\medskip

\section{Equilibrium characterization}\label{S:3}

In this section, we analyze the (equivalent) equilibrium in Definition \ref{eq.def2} for any fixed $N\geq1$ %(recalling its equivalence to the original equilibrium (Definition \ref{eq.def1}))
as well as its limiting counterpart in Definition \ref{leq.def}. To facilitate analysis, we introduce a few technical assumptions.

\begin{assumption}\label{as:1}
Under Remark \ref{rem:2}, the aggregation function $W\in\mathcal{C}^{(1,1)}(\mathbb{R}^{2}_{+};\mathbb{R}_{+})$, and its partial derivatives $W^{(1,0)}$ and $W^{(0,1)}$ are uniformly bounded, i.e., there exists a constant $K>0$ such that $|W^{(1,0)}(x,y)|\leq K$ and $|W^{(0,1)}(x,y)|\leq K$ for all $(x,y)\in\mathbb{R}^{2}_{+}$.
Furthermore, the criminal penalty functional (at time $\tau$) takes the form
\begin{equation}\label{crpe.s}
  \Pi_{0}(\theta_{[0,\tau]\!]})=\bigg(\int^{\tau}_{0}(\varpi_{0}(\theta_{s}))^{p}\dd s\bigg)^{1/p},
\end{equation}
for $p\geq1$ and some function $\varpi_{0}\in\mathcal{C}(\mathbb{R};\mathbb{R}_{+})$ that is decreasing on $\mathbb{R}_{-}$ and increasing on $\mathbb{R}_{+}$.
\end{assumption}

In Assumption \ref{as:1}, while the differentiability and uniformly bounded derivatives conditions on $W$ are natural and nonrestrictive, the $\mathbb{L}^{p}$-functional form (\ref{crpe.s}) implies that criminal charges are determined based on the IT's entire trade history, rather than at certain time points, once prosecution is successful. The coefficient $p$ measures the degree of concentration and the function $\varpi_{0}$ is interpreted as the instantaneous rate of penalty. In particular, for given $\tau$, when $p=1$, $\Pi_{0}(\theta_{[0,\tau]\!]})=\int^{\tau}_{0}\varpi_{0}(\theta_{s})\dd s$ is a linear functional of $\varpi_{0}(\theta)$; when $p=2$, $\Pi_{0}(\theta_{[0,\tau]\!]})=\big(\int^{\tau}_{0}(\varpi_{0}(\theta_{s}))^{2}\dd s\big)^{1/2}$ is a quadratic functional of $\varpi_{0}(\theta)$, which further coincides with the quadratic penalty choice in \c{C}etin (\citeyear{C25}) \text{Sect.} 2 if $\varpi=|\cdot|$ and $\tau=1$. In Theorem \ref{theorem:sup}, we shall consider the case $p\to\infty$ as well, with (\ref{crpe.s}) becoming
\begin{equation}\label{crpe.sl}
  \Pi_{0}(\theta_{[0,\tau]\!]})=\sup_{s\in[0,\tau]\!]}\varpi_{0}(\theta_{s}),
\end{equation}
in which case the penalty responds to the IT's most violent trade(s). With this limiting case, the penalty scheme (\ref{crpe.s}) offers a flexible framework that accommodates both insiders who concentrate risk in a few large trades and those who distribute it over time, hence delivering a magnitude-sensitive reinforcement of the deterrent effect.

The next two assumptions are about the functional behaviors of the hazard rate function $\lambda$, penalty function $W$, and the instantaneous penalty rate function $\varpi_{0}$ under Assumption \ref{as:1}. As discussed in Ma, Xia, and Zhang (\citeyear{MXZ25}), the assumed properties of $\lambda$ hold for most practical scenarios, including when investigations are initiated based on detected abnormal trading activity (DeMarzo, Fishman, and Hagerty \citeyear{DFH98}), while those of $\varpi_{0}$ are totally controllable by the RG, hence also nonrestrictive. In particular, in Assumption \ref{as:3}, the tail behaviors are assumed to be time-independent for simplicity, as imposing time dependence leads to no substantial differences or practical benefits from a modeling perspective. Of particular interest are the coefficients $\eta$ and $\alpha$, which respectively capture the RG's limited capacity to detect small-scale insider trading, reflecting potential incompetence in regulatory diligence, and the degree to which severe criminal penalties are imposed on large trades. The implications of these coefficients for insider trading strategies will be explored in Section \ref{S:4}.

\begin{assumption}\label{as:2}
$\lambda\in\mathcal{C}^{(0,1)}([0,T]\times(\mathbb{R}\setminus\{0\});\mathbb{R}_{+})$, with $\lambda^{(0,1)}<0$ on $\mathbb{R}_{--}$ and $\lambda^{(0,1)}>0$ on $\mathbb{R}_{++}$. Also, $\varpi_{0}\in\mathcal{C}^{1}(\mathbb{R}\setminus\{0\};\mathbb{R}_{+})$, with $\varpi'_{0}\leq0$ on $\mathbb{R}_{-}$ and $\varpi'_{0}\geq0$ on $\mathbb{R}_{+}$.
\end{assumption}

\begin{assumption}\label{as:3}
Under Assumption \ref{as:1} and Assumption \ref{as:2}, there exist constants $\eta\geq1$, $\alpha\geq1$, $\kappa>0$, $b\geq0$, and $C_{1}\geq0$ such that $\lambda(t,\imath)=\kappa|\imath|^{\eta}+o(|\imath|^{\eta})$ as $|\imath|\searrow0$ for all $t\in[0,T]$, $\varpi_{0}(\imath)=b|\imath|^{\alpha}+O(1)$ as $|\imath|\to\infty$, and $W(x,y)=C_{1}(x+y)+O(1)$ as $x,y\to\infty$.
\end{assumption}

We are ready to present our main result regarding any finite-$N$ equilibrium amid dynamic legal risk and the associated $\gamma$-limiting equilibrium as $N\to\infty$ with $\gamma<1/2$ under stealth trading.

\begin{theorem}\label{theorem:p}
Let Assumption \ref{as:1} and Assumption \ref{as:2} hold, with $1\leq p<\infty$. Then, we have the following two statements: \smallskip\\
(i) For any fixed $N\geq1$, the necessary conditions governing any equilibrium $(\theta^{\star}_{N},P^{\star}_{N})$ in the sense of Definition \ref{eq.def1} are given by\footnote{We write $\varpi_{0}'(\cdot\pm)$, respectively, for the left derivative and right derivative of $\varpi_{0}$.}
\begin{align}\label{NC1}
  &(1-c\bar Y^{\th_N^{\star v}}_{t})(v-P_{N,t}^\star)+\frac{Z_t^{\th^{\star v}_N}}{\sqrt{N}\si_t}-\lambda^{(0,1)}(t,\theta^{\star v}_{N,t}) \nonumber\\
  &\qquad\times\bigg(W\bigg(\bigg(\int^{t}_{0}(\varpi_{0}(\theta^{\star v}_{N,s}))^{p}\dd s\bigg)^{1/p},\max\bigg\{0,c\int^{t}_{0}\theta^{\star v}_{N,s}(v-P^\star_{N,s})\dd s\bigg\}\bigg)+Y^{\theta^{\star v}_N}_{t}\bigg) \nonumber\\
  &\quad-(\varpi_{0}(\theta^{\star v}_{N,t}))^{p-1}{\big(\varpi_{0}'(\theta^{\star v}_{N,t}+)\1_{\{\D\th^v_t>0\}}+\varpi_{0}'(\theta^{\star v}_{N,t}-)\1_{\{\D\th^v_t<0\}}\big)}\h Y_t^{\th^{\star v}_N}+\frac{Z_t^{\th^{\star v}_N}}{\sqrt{N}\si_t}=0, \quad\D\th^{v}\in \cA,\;v\in\cV, \nonumber\\
  &P^{\star}_{N,t}=\frac{\int_{\mathcal{V}}vX^{\theta^{\star v}_N}_{t}\dd\PP[V\leq v]}{\int_{\mathcal{V}}X^{\theta^{\star v}_N}_{t}\dd\PP[V\leq v]},\quad t\in[0,T],
\end{align}
where $X^{\theta^{\star v}_N}$ is specified in (\ref{dp}), $(Y^{\th^{\star v}_{N}},Z^{\th^{\star v}_N})$ is understood to satisfy the BSDE (\ref{BSDE1}) with $\theta=\theta^{\star}_N$, and $(\bar Y^{\th^{\star v}_{N}},\bar Z^{\th^{\star v}_N})$ and $(\h Y^{\th^{\star v}_{N}},\h Z^{\th^{\star v}_N})$ satisfy the following BSDEs, respectively:
\begin{align}\label{BSDE2}
  \bar Y^{\th^{\star v}_{N}}_t&=\int_t^T\bigg[\l_s^{\th^{\star v}_{N}}W^{(0,1)}\bigg(\bigg(\int^{s}_{0}(\varpi_{0}(\theta^{\star v}_{N,r}))^{p}\dd r\bigg)^{1/p},\max\bigg\{0,c\int^{s}_{0}\theta^{\star v}_{N,r}(v-P^\star_{N,r})\dd r\bigg\}\bigg) \nonumber\\
  &\qquad\times{\big(\1_{\{c\int^{s}_{0}(v-P_{N,r}^\star)\th_{N,r}^{\star v}\dd r>0\}}+\1_{\{\int^{s}_{0}(v-P_{N,r}^\star)\th_{N,r}^{\star v}\dd r=0,\; \int^s_{0}(v-P_{N,r}^\star)\D\th_r^v\dd r>0\}}\big)}-\l_s^{\th^{\star v}_{N}}\bar Y_s^{\th^{\star v}_{N}}+\frac{\th^{\star v}_{N,s}\bar Z_s^{\th^{\star v}_{N}}}{\sqrt{N}\si_s}\bigg]\dd s \nonumber\\
  &\quad-\int_t^T\bar Z_s^{\th^{\star v}_{N}}\dd B_s, \nonumber\\
  \h Y^{\th^{\star v}_N}_t&=\int_t^T\bigg[\l_s^{\th^{\star v}_{N}}W^{(1,0)}\bigg(\bigg(\int^{s}_{0}(\varpi_{0}(\theta^{\star v}_{N,r}))^{p}\dd r\bigg)^{1/p},\max\bigg\{0,c\int^{s}_{0}\theta^{\star v}_{N,r}(v-P^\star_{N,r})\dd r\bigg\}\bigg) \nonumber\\
  &\qquad\times\bigg(\int^{s}_{0}(\varpi_{0}(\th^{\star v}_{N,r}))^{p}\dd r\bigg)^{1/p-1}-\l_s^{\th^{\star v}_{N}}\h Y_s^{\th^{\star v}_{N}}+\frac{\th^{\star v}_{N,s}\h Z_s^{\th^{\star v}_{N}}}{\sqrt{N}\si_s}\bigg]\dd s-\int_t^T\h Z_s^{\th^{\star v}_{N}}\dd B_s, \nonumber\\
  &\quad v\in\mathcal{V},\;t\in[0,T].
\end{align}
\smallskip
(ii) Suppose that $v\neq\E[V]$ and $2\beta\eta<\eta+\alpha-1$ under Assumption \ref{as:3}. Then, a necessary condition for a $\gamma$-limiting equilibrium $(\tilde{\theta}^{\star}_\g,\tilde{P}^{\star}_\g)$ to exist in the sense of Definition \ref{leq.def} is
\begin{equation}\label{og}
  \gamma=\frac{\beta\eta}{\eta+\alpha-1}.
\end{equation}
Moreover, whenever such $(\tilde{\theta}^{\star}_\g,\tilde{P}^{\star}_\g)$ exists, $\tilde{\theta}^{\star}_\g$ is deterministic.
\end{theorem}

\medskip

Statement (i) in Theorem \ref{theorem:p} has provided a partial characterization of the finite-$N$ equilibrium in Definition \ref{eq.def2}, through the necessary conditions in \eqref{BSDE2} underlying its existence, which are structurally similar to the outcomes of applying the stochastic maximum principle for optimal control problems -- thanks to Proposition \ref{pro:2} -- in spite of the game structure between the IT and the MM. Nevertheless, the BSDEs in (\ref{NC1}) and (\ref{BSDE2}) form a collection of convoluted fixed-point systems indexed by the auxiliary variable $\D\th^v\in\cA$, $v\in\cV$, for which neither existence nor uniqueness can be guaranteed in general, whereas this analytical challenge in turn accentuates the computational advantage of the associated $\gamma$-limiting equilibrium; this point will be elaborated upon in discussing statement (ii). In particular, the use of $\D\th^v\in\cA$ in \eqref{NC1} is clarified in the following remark. % Since the payoff depends on the path of the insider strategy, the first order condition needs to hold of all direction for $\D\th$. Notice that for each $\D\th$, a system \eqref{NC1} and \eqref{BSDE2} is defined.

\begin{remark}\label{rmk5}
Given the general dependence on the paths of $\tl\th^v$ (recalling Proposition \ref{pro:3}), the auxiliary variable $\D\th^v\in\cA$ appears due to the general non-differentiability of the composed aggregation function $W$ and the rate of penalty $\varpi_{0}$ at $0$. In particular, if $\varpi_0$ happens to be differentiable at $0$ and $W^{(0,1)}(\cdot,0)\equiv0$,\footnote{The latter property holds trivially if $c=0$ (no civil penalties).} then the first condition in (\ref{NC1}) can be stated without using $\D\th^v$, or
\begin{align}\label{NC1'}
  &(1-c\bar Y^{\th_N^{\star v}}_{t})(v-P_{N,t}^\star)+\frac{Z_t^{\th^{\star v}_N}}{\sqrt{N}\si_t}-\lambda^{(0,1)}(t,\theta^{\star v}_{N,t}) \nonumber\\
  &\qquad\times\bigg(W\bigg(\bigg(\int^{t}_{0}(\varpi_{0}(\theta^{\star v}_{N,s}))^{p}\dd s\bigg)^{1/p},\max\bigg\{0,c\int^{t}_{0}\theta^{\star v}_{N,s}(v-P^\star_{N,s})\dd s\bigg\}\bigg)+Y^{\theta^{\star v}_N}_{t}\bigg) \nonumber\\
  &\quad-(\varpi_{0}(\theta^{\star v}_{N,t}))^{p-1}\varpi_{0}'(\theta^{\star v}_{N,t})\h Y_t^{\th^{\star v}_N}+\frac{Z_t^{\th^{\star v}_N}}{\sqrt{N}\si_t}=0, \quad v\in\cV,\nonumber\\
  &P^{\star}_{N,t}=\frac{\int_{\mathcal{V}}vX^{\theta^{\star v}_N}_{t}\dd\PP[V\leq v]}{\int_{\mathcal{V}}X^{\theta^{\star v}_N}_{t}\dd\PP[V\leq v]},\quad t\in[0,T],
\end{align}
with (\ref{BSDE2}) simplifying into
\begin{align}\label{BSDE2'}
  \bar Y^{\th^{\star v}_{N}}_t&=\int_t^T\bigg[\l_s^{\th^{\star v}_{N}}W^{(0,1)}\bigg(\bigg(\int^{s}_{0}(\varpi_{0}(\theta^{\star v}_{N,r}))^{p}\dd r\bigg)^{1/p},\max\bigg\{0,c\int^{s}_{0}\theta^{\star v}_{N,r}(v-P^\star_{N,r})\dd r\bigg\}\bigg) \nonumber\\
  &\quad-\l_s^{\th^{\star v}_{N}}\bar Y_s^{\th^{\star v}_{N}}+\frac{\th^{\star v}_{N,s}\bar Z_s^{\th^{\star v}_{N}}}{\sqrt{N}\si_s}\bigg]\dd s-\int_t^T\bar Z_s^{\th^{\star v}_{N}}\dd B_s, \nonumber\\
  \h Y^{\th^{\star v}_{N}}_t&=\int_t^T\bigg[\l_s^{\th^{\star v}_{N}}W^{(1,0)}\bigg(\bigg(\int^{s}_{0}(\varpi_{0}(\theta^{\star v}_{N,r}))^{p}\dd r\bigg)^{1/p},\max\bigg\{0,c\int^{s}_{0}\theta^{\star v}_{N,r}(v-P^\star_{N,r})\dd r\bigg\}\bigg) \nonumber\\
  &\qquad\times\bigg(\int^{s}_{0}(\varpi_{0}(\th^{\star v}_{N,r}))^{p}\dd r\bigg)^{1/p-1}-\l_s^{\th^{\star v}_{N}}\h Y_s^{\th^{\star v}_{N}}+\frac{\th^{\star v}_{N,s}\h Z_s^{\th^{\star v}_{N}}}{\sqrt{N}\si_s}\bigg]\dd s-\int_t^T\h Z_s^{\th^{\star v}_{N}}\dd B_s, \nonumber\\
  &\quad v\in\mathcal{V},\;t\in[0,T].
\end{align}
In this case, an (equivalent) equilibrium can be found by solving a single fixed-point system formed by \eqref{NC1'} and \eqref{BSDE2'}.
\end{remark}

Statement (ii) in Theorem \ref{theorem:p} further establishes that, under the parametric constraint $\eta+\alpha-1>2\beta\eta$, any $\gamma$-limiting equilibrium is deterministic, which feature is attributable to the absence of randomness in the MM's pricing rule under diminished price impact (recalling \eqref{MM.lp}). Remarkably, such determinism is not overly idealized but arises naturally in the context of stealth trading, where the IT can execute trades without concern for his price impact, which in turn provides a clean shell for analyzing the effects of regulatory practices on the IT's trading behavior over time -- in Section \ref{S:4}. In addition, the above parametric constraint implies that, for a fixed penalty growth rate (given $\alpha$), a higher investigation intensity (small $\beta$) tends to incentivize stealth trading for the IT, who is able to achieve a suitable stealth level ($\gamma$), underlying a stable equilibrium. Conversely, for a fixed detection mechanism (given $\eta$ and $\beta$), the same holds true with enhanced penalties (large $\alpha$). This constraint is also considerably loose in that even when $\alpha=1$, signifying a mild penalty scheme, then it only necessitates that $\beta<1/2$, i.e., detection of the IT is always possible however large the noise trading population is, with $1/2$ corresponding to the growth rate of the order flow from the NTs. More precisely, in this general setting, we recover the same stealth index in Ma, Xia, and Zhang (\citeyear{MXZ25}) \text{Thm.} 3. Specifically, the stealth index is directly related to $\beta$ and inversely related to $\alpha$, showing the IT's internalization of legal risk even in the large-population setting -- jointly controlled by the RG's detection mechanism and impositions of penalties, regardless of which factor dominates.

% {\color{blue} Add interpretations of equilibrium structures. Hint at the RG's problem.}

Pertaining to the facilitation of equilibrium analysis, we also state the following technical remark to clarify the infeasibility to establish a general existence and uniqueness result for the $\gamma$-limiting equilibrium.

\begin{remark}\label{sol}
In light of Theorem \ref{theorem:p} statement (ii), solving a $\gamma$-limiting equilibrium comes down to solving a deterministic optimal control problem (more specifically, (\ref{DCP-1}), (\ref{DCP-2}), and (\ref{DCP-3}) in Appendix \ref{B}), which can be tackled with classical approaches such as deterministic maximization principles or the Hamilton--Jacobi--Bellman equation. Although the objective functions generally depend on the paths of the control $\tl\th^{v}$, a tractable method is to consider integrals over $\tl\th^{v}$ as state dynamics that depend on the control. That said, because the resulting admissibility set (or control space), namely $\mathbb{L}^{2}([0,T];\sgn(v-\E[V])\mathbb{R}_{+})$ (see (\ref{ladm})), is noncompact, neither existence nor the uniqueness of an optimal control can be guaranteed in general. %\footnote{Nevertheless, the optimal value always exists (nonzero) and is unique by our analysis and one can always find it by taking the limit of the payoff achieved by $\e$-optimal controls.} plus the payoff structure is generally not Lipschitz-continuous unless  special choices of $\a$, $\eta$ and $p$.
In Section \ref{S:4}, we shall examine scenarios in which the limiting equilibrium can be solved explicitly and which offer sufficient coverage to draw economic insights.
\end{remark}

Aside from the aforementioned empirical support, the functional realism of considering the $\gamma$-limiting equilibrium lies in another key result, i.e., it serves as a practically reasonable approximation for the finite-population equilibrium. Beforehand, we need to introduce the notion of $\epsilon$-equilibrium in the continuous-time setting.

\begin{definition}\label{eeq.def}
For any fixed $N\geq1$ and $\e\geq0$, an $\epsilon$-equilibrium is a pair $(\theta^{\star,\e}_{N},P^{\star,\e}_{N})$ of trading strategy and pricing rule such that: Based on \eqref{IT.o4} and \eqref{MM.p2},
\begin{equation}\label{IT.eop}
  N^{-\gamma}J(P^{\star,\e}_{N};\theta^{\star,\e\;v}_{N},v)\geq \sup_{\theta^{v}\in\cA}N^{-\g}J(P^{\star,\e}_{N};\theta^{v},v)-\e,\quad v\in\mathcal{V},
\end{equation}
and
\begin{equation}\label{MM.ep}
  \dbE[|P^{\star,\e}_t-P^{\th^{\star,\e}}_t|]\leq \e,\quad t\in[0,T].
\end{equation}
\end{definition}

In Definition \ref{eeq.def}, while \eqref{IT.eop} is a standard condition for $\epsilon$-optimality among (equilibrium) trading strategies, the unconditional expectation (with respect to $\PP$) is taken for the pricing rules in \eqref{MM.ep} mostly for convenience -- to ease the derivation of a weak, universal integrability condition for the asset value $V$ (see \eqref{V.ic}). In principle, one may define the $\epsilon$-equilibrium in a purely pointwise sense for each $\omega\in\Omega$ and proceed with the same type of arguments, albeit under a much more obscure condition for $V$. Theorem \ref{theorem:s} provides precise statements about the approximation power of the $\gamma$-limiting equilibrium.

\begin{theorem}\label{theorem:s}
Assume the setting of Theorem \ref{theorem:p}. Let $(\tilde{\theta}^{\star}_{\gamma},\tilde{P}^{\star}_{\gamma})$ be a $\gamma$-limiting equilibrium, with $\gamma$ given by (\ref{og}) and $\tilde{\theta}^{\star}_\g$ being deterministic, and suppose that
\begin{equation}\label{V.ic}
  \E\big[(V-\E[V])^{2}e^{3\int^{T}_{0}(\tilde{\theta}^{\star V}_{t}/\sigma_{t})^{2}\dd t}\big]<\infty.
\end{equation}
Then, $(N^{\gamma}\tilde{\theta}^{\star}_\g,\tilde{P}^{\star}_{\gamma})$ is also an $\epsilon_{N}$-equilibrium in the finite-$N$ setting for any $N\geq1$, where, for some constant $C>0$ independent of $N$,
\begin{equation}\label{eps.rate}
  \epsilon_{N}\leq
  \begin{cases}
    CN^{\gamma-1/2},&\;\text{if }\beta=0, \\
    \begin{cases}
      CN^{\max\{\gamma-1/2,-\gamma\}},&\;\text{if }\alpha=1, \\
      CN^{\max\{\gamma-1/2,-\gamma(\alpha-1)\}},&\;\text{if }\alpha>1,
    \end{cases}
    &\;\text{if }\beta>0.
  \end{cases}
\end{equation}
\end{theorem}

Theorem \ref{theorem:s} postulates that any $\gamma$-limiting equilibrium with $\gamma<1/2$, upon existence, is $\epsilon$-optimal in the finite-population setting, for which closeness is inversely related to the population size $N$ and also accelerated by the stealth index $\gamma$, subject to mutual influence from $\alpha$ (severity of criminal penalty) and $\beta$ (investigation intensity). This result reinforces a key conclusion in Ma, Xia, and Zhang (\citeyear{MXZ25}), supporting the practical validity of using the $\gamma$-limiting equilibrium for tractable analysis -- especially in light of its empirical grounding for stealth trading (\text{cf.} Meulbroek \citeyear{M92}, Chakravarty \citeyear{C01}, Frino \text{et al.} \citeyear{FSWZ13}, and Kacperczyk and Pagnotta \citeyear{KP24}). The integrability condition \eqref{V.ic}, which naturally restricts the asset value $V$ from being excessively heavy-tailed, is easily verified on a case-by-case basis, provided that the limiting equilibrium strategy $\tl\th^{\star v}$ is sufficiently slowly varying in $v$.\footnote{This condition can hold even for distributions that lack globally existing moment-generating functions, including many generalized Gaussian distributions -- as is the case with the scenarios in Section \ref{S:4}.}

The closing result in this section is the following theorem, which addresses the maximum trade-focused criminal penalty scheme in \eqref{crpe.sl}, corresponding to the limit of the degree of (criminal) penalty concentration $p\to\infty$ in \eqref{crpe.s}. Again, consideration of this scheme is justified by its practical relevance in that regulatory scrutiny is readily triggered by a single trade representing a large fraction of daily volume, and so penalizing the most violent trade(s) targets the most detectable and reckless behavior. Technically, what the theorem establishes is an analogous property that the resulting equilibrium can be approximated by one under penalties in the general form \eqref{crpe.s} with $1\leq p<\infty$, in much the same sense as a $\gamma$-limiting equilibrium approximates any finite-$N$ equilibrium (Theorem \eqref{theorem:s}).

\begin{theorem}\label{theorem:sup}
Let Assumption \ref{as:1} and Assumption \ref{as:2} hold and let $N\geq1$ be fixed. Suppose we restrict $\cA$ to the subset $\cA'$ containing all such $\theta^v$ (for any $v\in\cV$) that are uniformly bounded and Lipschitz-continuous in time with a common Lipschitz constant $L>0$, i.e., $|\th^v_t|\leq K$ and $|\th^v_t-\th^v_s|\leq L|t-s|$ (a.s.) for any $t,s\in[0,T]$. Then the following two statements hold. \smallskip\\
(i) If there is an (equivalent) equilibrium $(\th^{\star}_N, P^{\star}_N)$ under (\ref{crpe.sl}), then $(\th^{\star}_N,P^{\star}_N)$ is also an $\epsilon_p$-equilibrium under (\ref{crpe.s}), with $\lim_{p\to\infty}\epsilon_p=0$. \\
(ii) For any $\e>0$, there exists $p_0\geq1$ such that for all $p\geq p_0$, an equilibrium $(\th^{\star}_N,P^{\star}_N)$ under (\ref{crpe.s}), upon existence, is an $\e$-equilibrium under (\ref{crpe.sl}).
\end{theorem}

\medskip

\section{Dynamic legal risk and regulatory impact}\label{S:4}

The goal of this section is to study the effect of dynamic legal risk on the IT's temporal trading strategies through various explicit solutions, via a specification analysis for Theorem \ref{theorem:p}. In view of stealth trading, we concentrate on the $\gamma$-limiting equilibrium, where the stealth index $\gamma$ is jointly determined by the three coefficients $\alpha,\beta,\eta$ defining regulatory practices, as explained in statement (ii) of Theorem \ref{theorem:p}. Subject to a certain stealth level (fixed $\gamma$), the trading strategies bear fundamental relations to the parameters $p,b,c,C_{1}$ as well as $C_{2}=\kappa bC_{1}$, or more generally, the functional forms of $\Pi_{0}$ and $W$, which are indicative of penalty structures. As already demonstrated, the limiting equilibrium filters out impact-induced noise from the MM's pricing rule, thereby revealing a much clearer pattern of the IT's trading behavior within an evolving legal context. In this regard, analyzing the limiting equilibrium provides equally valuable insights for our key conclusions, apart from a direct comparison with existing literature.

To keep the analysis clean, we shall focus on the two scenarios, representing two sufficiently differentiated contexts and corresponding to Subcase 2.3 and Subcase 2.2, respectively, in the proof of Theorem \ref{theorem:p} in Appendix \ref{B} -- one with severe criminal penalties upon successful prosecution and the other with little to none. The first scenario involves $C_{2}$ but not $C_{1}$, and we start by considering $\eta$ fixed while varying $p$ and $C_{2}$ in a reasonable range. Also, since $\gamma$ is fixed, we shall omit it in the subscript of the (limiting) equilibrium strategy for notational simplicity throughout this section.

\begin{proposition}\label{pro:4}
Let $\eta=1$, $\alpha>1$, and $2\beta<\alpha$. Then, with $\gamma=\beta/\alpha\in[0,1/2)$, the unique $\gamma$-limiting equilibrium is given by $(\tilde\theta^{\star}\equiv\tilde\theta^{\star}_{\gamma},\tilde{P}^{\star}_{\gamma}\equiv\E[V])$, where, for $v\in\mathcal{V}\setminus\{\E[V]\}$ and $t\in[0,T)$,
\begin{equation}\label{P1:th}
  \tilde\theta^{\star v}_{t}=\sgn(v-\E[V])\bigg(\frac{g_{v}(0)}{T}\bigg)^{1/(p\alpha)} \bigg(|v-\E[V]|-C_{2}\bigg(g^{-1}_{v}\bigg(\frac{T-t}{T}g_{v}(0)\bigg)\bigg)^{1/p}\bigg)^{-1},
\end{equation}
in which, in terms of the incomplete Beta function $\mathrm{B}_{\cdot}(\cdot,\cdot)$ (DLMF \S8.17),
\begin{equation}\label{P1:gv}
  g_{v}(x):=\frac{p|v-\E[V]|^{p(\alpha+1)}}{C^p_{2}} \big(\mathrm{B}_{1}(p,p\alpha+1)-\mathrm{B}_{C_{2}x^{1/p}/|v-\E[V]|}(p,p\alpha+1)\big),\quad x\in\bigg[0,\bigg(\frac{|v-\E[V]|}{C_{2}}\bigg)^{p}\bigg],
\end{equation}
and $g^{-1}_{v}(\cdot)$ denotes its inverse.
\end{proposition}

Proposition \ref{pro:4} reveals an interesting phenomenon: Judging from (\ref{P1:th}), the equilibrium strategy $\tl\th^{\star v}_{t}$ as a function of time is strictly increasing and, somewhat surprisingly, still unbounded as $t\nearrow T$, despite the presence of legal risk -- in particular, $\lim_{t\nearrow T}\tilde\theta^{\star v}_{t}=\sgn(v-\E[V])\infty$. On a closer look, this explosive behavior can be attributed to the specific structure of detection mechanism and legal penalties. In this context, knowing that the RG exercises a fixed level of detection diligence ($\eta=1$), the IT is incentivized to exploit a perceived vulnerability as he holds out until the public announcement, thus maintaining a motive to pursue a bridge-type strategy -- in parallel to the original discovery of Back (\citeyear{B92}) \text{Thm.} 1 (see also Caldentey and Stacchetti \citeyear{CS10} \text{Thm.} 1 for the case of a random time of announcement). Nevertheless, here the blowup rate is significantly lower, with
\begin{equation}\label{th_asym}
  |\tl\theta^{\star v}_{t}|\sim K(T-t)^{-1/(p\alpha+1)},\quad\text{as }t\nearrow T,
\end{equation}
for some positive constant $K$, noted that $p\alpha>1$ (details in Appendix \ref{C}). This indicates that, with the prevailing legal risk, the strategy has gained a desirable degree of temporal integrability despite unboundedness, particularly in agreement with the square-integrability of strategies imposed in (\ref{wf.adm.t}), or more precisely, $\tl\theta^{\star v}\in\mathbb{L}^{2}([0,T];\sgn(v-\E[V])\mathbb{R}_{+})$ based on \eqref{ladm}. Noteworthily, as detection of insider trading is conducted dynamically, in realization, the IT's strategy can only blow up if he is not identified before the announcement, namely on $\{\tau>T\}$. To technical extent, this integrability property also speaks to the naturalness of incorporating the alternative (weak) formulation (Qiao and Zhang \citeyear{QZ25a}) in solving insider trading models with legal risk.

We give a numerical illustration of (\ref{P1:th}) with the following parametric choices: $T=1$, $p\in[1,6]$, $\alpha=2$, $C_{2}\in[1,3]$; we also assume that $\E[V]=\sqrt{e}$.\footnote{One can show that the integrability condition \eqref{V.ic} is satisfied for all three propositions in this section if the distribution of $V$ has exponential or faster tail decay, trivially including the Gaussian distribution. Detailed verification is omitted here to keep conciseness, as the specific form of the distribution is not essential beyond the generality conveyed by this condition.} Due to symmetry, we focus on the case $v-\E[V]>0$ by taking $v=3$ so that the illustration is for (equilibrium) buy strategies exclusively. In Figure \ref{fig:1}, the top left panel shows the paths of $\tilde{\theta}^{\star v}$ by fixing $C_{2}=1$ while varying $p$, the top right panel does so fixing $p=2$ and varying $C_{2}$, and the bottom panel displays a 3D plot for the value of $\theta^{\star v}_{t}$ at $t=T/2=0.5$ with respect to both $p$ and $C_{2}$.

%\textcolor{orange}{[Figure 1 about here]}

\begin{figure}[H]
  \centering
  \includegraphics[scale=0.35]{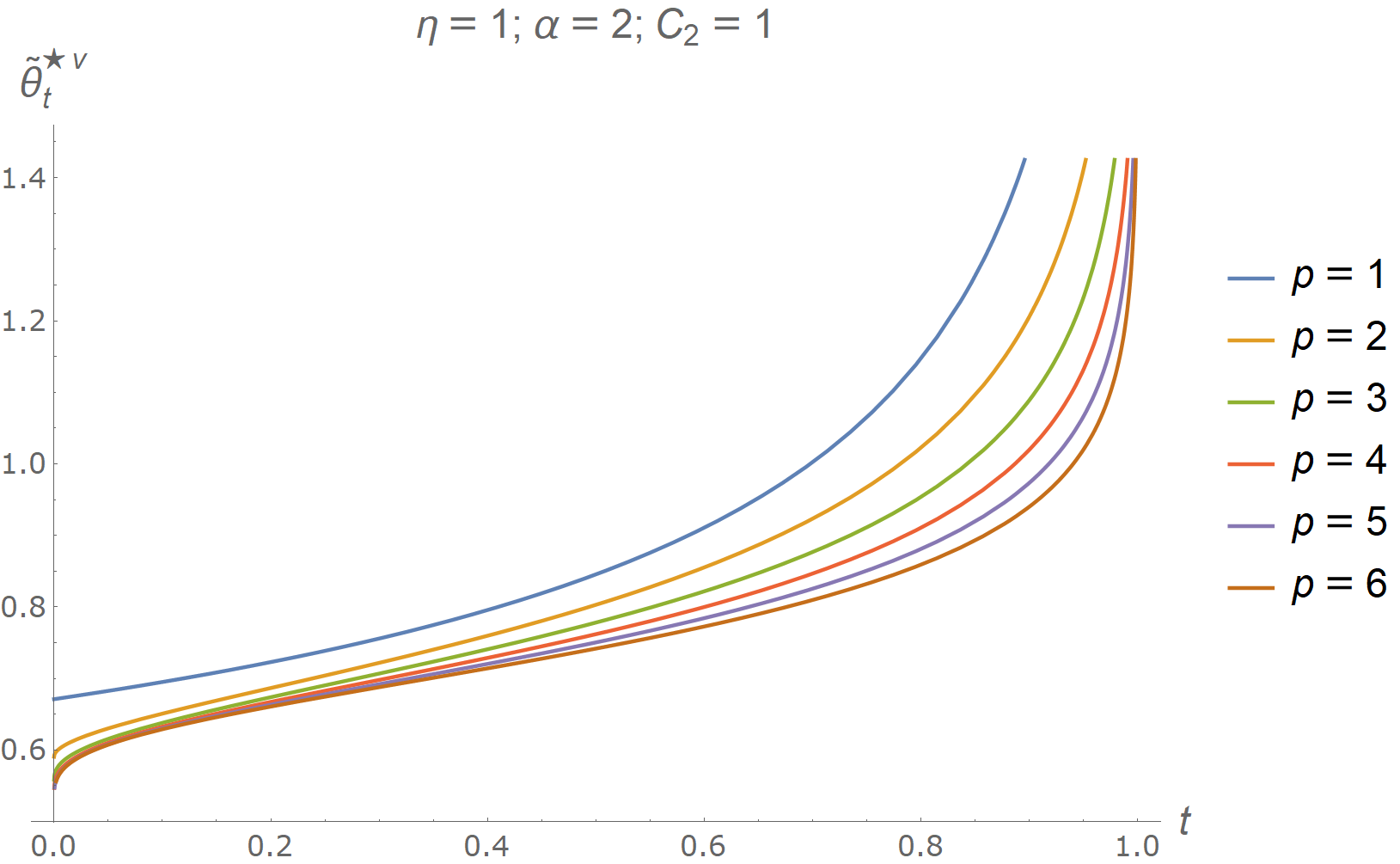}
  \includegraphics[scale=0.35]{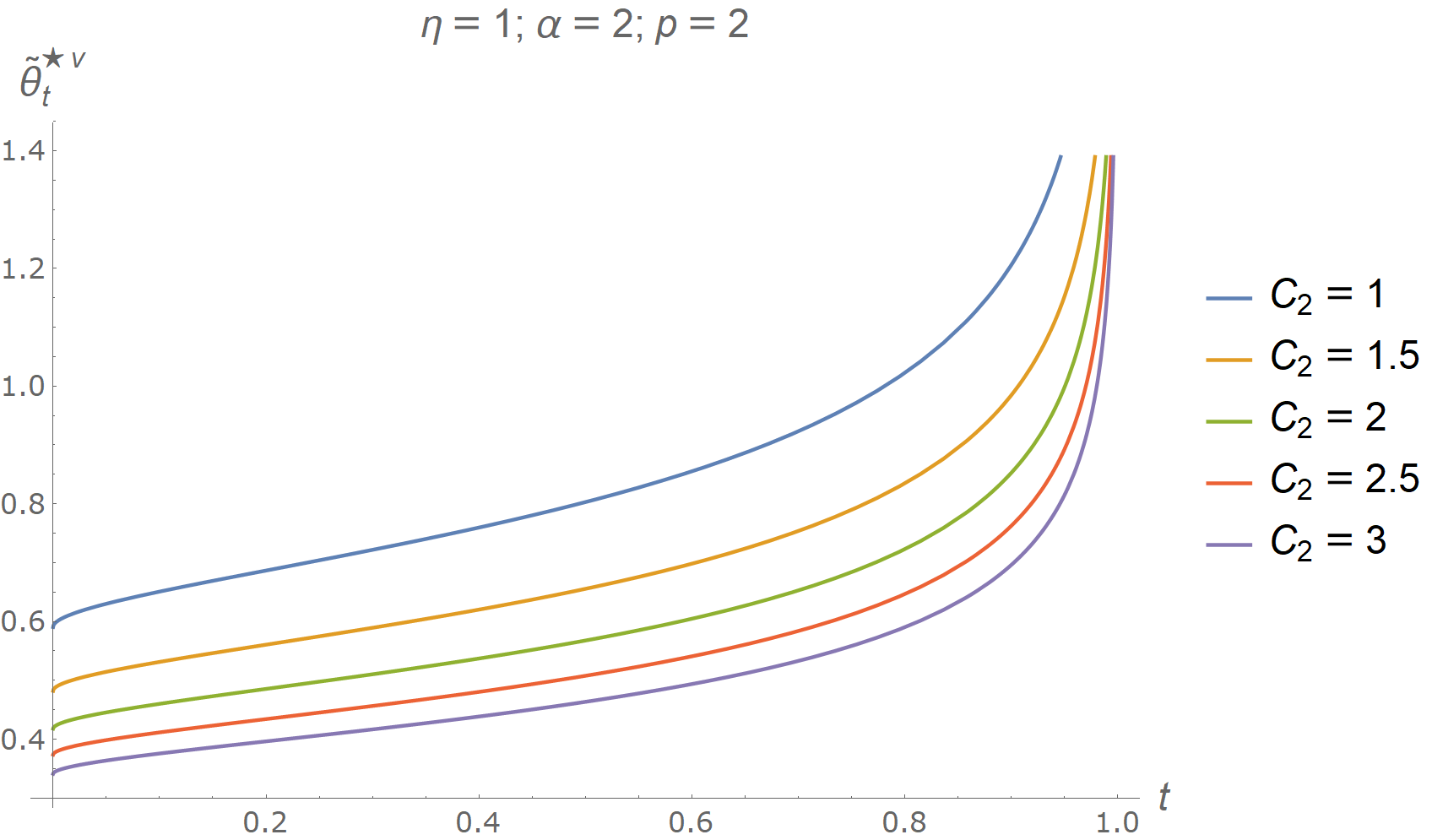}
  \includegraphics[scale=0.35]{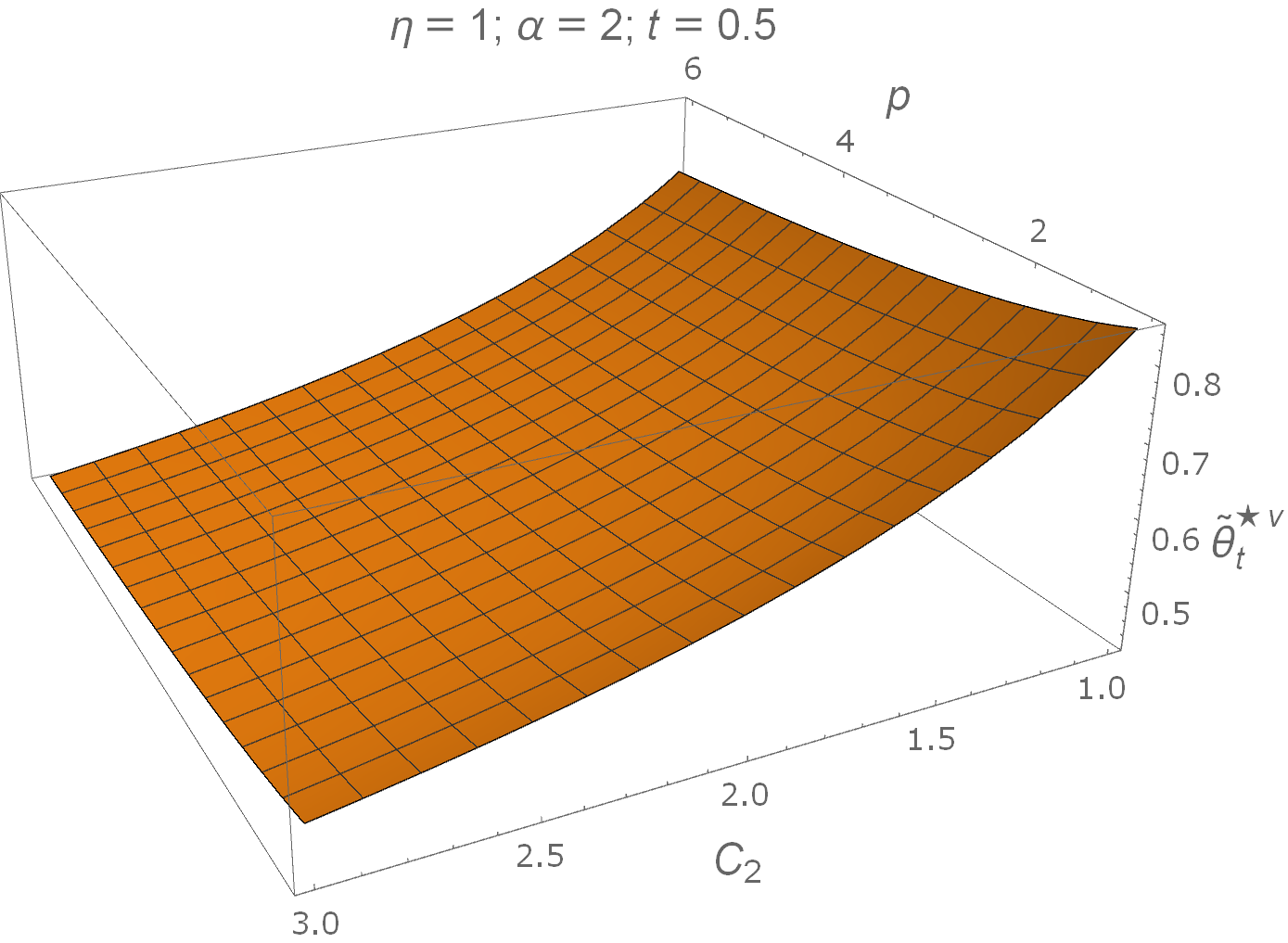}\\
  \caption{Limiting equilibrium strategies and legal risk impact (I)}\label{fig:1}
\end{figure}

From the two top panels, it is visible that the whole paths of the equilibrium strategies are decreasing in $p$ as well as in $C_{2}$, holding other things constant. Since $p$ measures the degree of (criminal) penalty concentration and $C_{2}$ serves indirectly as a (criminal) penalty multiplier, such a decrease comes as a natural indication of the deterrent effect of legal enforcement. In comparison, increasing $C_{2}$ is associated with a much sharper decrease compared to increasing $p$, which highlights that adjusting penalty levels, rather than their temporal distribution, tends to be a more effective mechanism for deterring insider trading -- or, put differently, for fostering markets less prone to adverse selection. Besides, subject to dependence on $p$ and $C_{2}$, the temporal increase of $\tl\theta^{\star v}$ can exhibit notable local curvatures.

Another remarkable observation is that, with dynamic legal risk, the IT's trading behavior can change drastically from the static case where prosecution can only occur at time $T$; compare \c{C}etin (\citeyear{C25}) \text{Prop.} 5.2. More precisely, the IT no longer trades a constant multiple of the discrepancy between the asset's fundamental value and its market price, or $v-\E[V]$, while the trade size gradually increases over time. Again, this increase is economically justified in that as time passes and the IT manages to continuously evade detection, the diminishment of legal risk encourages more aggressive trading to capture greater profit. In other words, the later the point in time, the lower the risk of prosecution, which dynamic is uniquely revealable when detection probabilities are time-varying.

% Another remarkable feature is that the trading strategy is indeed square-integrable over time, while it exhibits a bridge-type structure as in Back; the integrability issue is resolved by the presence of legal risk, even though the IT can still increase his trade size to infinity close to the public announcement date, given that he has not been caught by then.

\smallskip

Next, we consider the pinned parametric relation $\eta=p\alpha$. Under this relation, the RG may trade off diligence in detecting minor informed trades (as $\eta$ can freely increase) while mandating higher and more concentrated penalties (by increasing $p$ and $C_{2}$) to compensate for this reduced oversight. As a consequence, advantageous selection costs rise both in magnitude and persistence, as the IT is incentivized to fragment his activity into smaller, less detectable trades. This case is thus meaningful for understanding whether increased penalties can effectively counterbalance detection lapses -- or whether such lapses create a calculable space where insiders can strategically outweigh legal risks.

\begin{proposition}\label{pro:5}
Let $\alpha>1$, $\eta=p\alpha>1$, and $2\beta<\alpha$. Then, with $\gamma=\beta\eta/(\eta+\alpha-1)\in[0,1/2)$, the unique $\gamma$-limiting equilibrium is given by $(\tilde\theta^{\star},\tilde{P}^{\star}_{\gamma}\equiv\E[V])$, where, for $v\in\mathcal{V}\setminus\{\E[V]\}$ and $t\in[0,T]$,
\begin{equation}\label{P2:th}
  \theta^{\star v}_{t}=\sgn(v-\E[V])\bigg(\frac{|v-\E[V]|}{p\alpha C_{2}}\bigg)^{1/((p+1)\alpha-1)}T^{1/(p(1-\alpha-p\alpha))},
\end{equation}
which is time-invariant.
\end{proposition}

Interestingly, when holding $\eta$ equal to $p\alpha$, the equilibrium strategy becomes time-invariant -- a noticeable contrast with the previous case in Proposition \ref{pro:4}. While absolute time-invariance is not the central issue here, in the wake of diminished price impact by stealth trading, a more fundamental implication is that the IT executes trades proportional to a fixed power ($1/((p+1)\alpha-1)\in(0,1)$) of the fundamental price discrepancy ($V-\E[V]$), yet becoming consistent with \c{C}etin (\citeyear{C25}). This constancy arises from the IT's perception of legal enforcement as stable and balanced over time, especially sensitive to large orders, not considering it creating profit-maximizing escalation as the announcement ($T$) nears. Indeed, even when $t\approx T$, since $\eta>1$, the RG gains agility in detecting substantial trades, notwithstanding a growing tolerance for minor ones, and by factoring in this enforcement behavior, the IT rationally judges it suboptimal to risk increasing trade size before the announcement, opting instead for a stable, security-focused strategy. It is important to note, however, that the constant (equilibrium) strategy still depends on the trading horizon length ($T$), which is to be expected since, based on (\ref{P2:th}), with $1-\alpha-p\alpha<0$, $\tl\theta^{\star v}$ is inversely related to $T$; that is, a shorter horizon ($T$ smaller) allows less time for potential detection, hence giving the IT greater incentive to trade aggressively.

In addition, in this case, the IT's trade size does not universally decrease as the concentration degree $p$ increases. Contrarily, $\tl\theta^{\star v}$ can even increase with $p$, especially when $\alpha$ is large. This implies the IT's tendency to exploit regulatory oversight in detecting minor trades, which dominates the otherwise enhanced concentration of penalty based on his trade history. Differently, the parameter $C_{2}$ (indirect criminal penalty multiplier) remains effective in bringing down the amount traded by the IT. This observation underlines that regulatory investigation remains a crucial deterrent to illicit insider trading -- one that cannot simply be offset by na\"{\i}vely increasing penalties, regardless of how narrowly those penalties are focused on a single or several large trades within the entire history.

%\textcolor{orange}{[Figure 2 about here]}

\begin{figure}[H]
  \centering
  \includegraphics[scale=0.35]{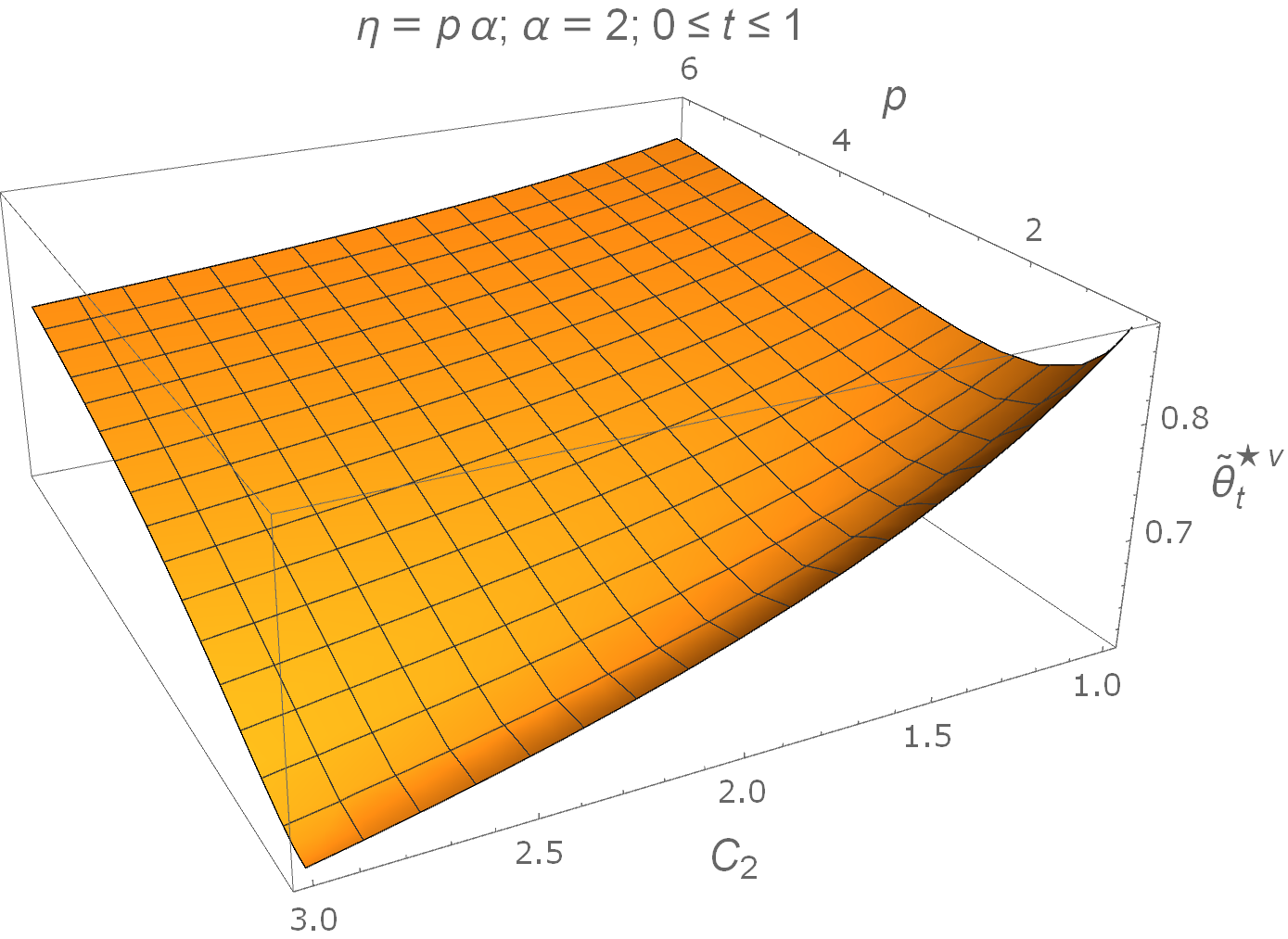}\\
  \caption{Limiting equilibrium strategies and legal risk impact (II)}\label{fig:2}
\end{figure}

For a numerical illustration of (\ref{P2:th}), we stick to $T=1$, $p\in[1,6]$, $\alpha=2$, $C_{2}\in[1,3]$, and $\E[V]=\sqrt{e}$, taking $v=3$ to only illustrate buy strategies. Figure \ref{fig:2} gives a 3D plot for the value of $\theta^{\star v}_{t}$ at arbitrary $t\in[0,1]$ (due to time invariance) with respect to $p$ and $C_{2}$.

From Figure \ref{fig:2}, we see that $C_{2}$ is indeed associated with a much sharper reduction in $\tl\th^{\star v}_{t}$ compared to the effect of $p$. In particular, as $p$ increases, $\tl\th^{\star v}_{t}$ ultimately rises --  because the accompanying increase in $\eta$ simultaneously makes smaller, more conservative trades progressively harder to detect.

\smallskip

The second scenario is about criminal penalties that are ineffective on large insider trades, which takes us to the case $\alpha=1$, involving the parameters $C_{1}$ in place of $C_{2}$. To facilitate analysis, we shall assume a linear functional form for the hazard rate, which also benefits comparability to the case considered in Proposition \ref{pro:4} (with $\eta=1$). Within this scenario, we again look into two cases, which reveal a parametric singularity whereby the uniqueness of equilibrium strategies collapses when criminal penalties are excessively averaged over the trading horizon or removed altogether.

\begin{proposition}\label{pro:6}
Let $\alpha=1$ and $\beta<1/2$, and let $\lambda(t,z)=\lambda(z)=\kappa|z|$ with $\kappa>0$. Then, with $\gamma=\beta\in[0,1/2)$, if either $p=1$ or $b=0$, there are infinitely many $\gamma$-limiting equilibria given by $(\tilde\theta^{\star},\tilde{P}^{\star}_{\gamma}\equiv\E[V])$, where, for $v\in\mathcal{V}\setminus\{\E[V]\}$,
\begin{equation}\label{P3:th}
  \tilde\theta^{\star v}\in\bigg\{\vartheta\in\mathbb{L}^{2}([0,T];\sgn(v-\E[V])\mathbb{R}_{+}):\;\int^{T}_{0}|\vartheta(t)|\dd t=\frac{|v-\E[V]|}{\kappa C_{1}(b+c|v-\E[V]|)}\bigg\}.
\end{equation}
If $p>1$ and $b>0$, there exists a unique $\gamma$-limiting equilibrium, $(\tilde\theta^{\star},\tilde{P}^{\star}_{\gamma}\equiv\E[V])$, with
\begin{equation}\label{P3:th2}
  \tilde\theta^{\star v}_{t}=(h'(H^{-1}(t)))^{-1/(p-1)},\quad t\in[0,T),
\end{equation}
where $h$ and $H$ are jointly determined by the following ODE system:
\begin{equation}\label{P3:DS}
  \begin{cases}
    \displaystyle h''(x)=\frac{\kappa bC_{1}}{\mu}\bigg(\frac{p-1}{p}\bigg)^{2}e^{-\kappa x}(h(x))^{1/p-1}(h'(x))^{(2p-1)/(p-1)}, \\
    H'(x)=(h'(x))^{-1/(p-1)}, \quad x\in(0,\bar x), \\
    h(0)=0,\;h'(0)=\chi, \\
    H(0)=0,\;H(\bar x)=T, \\
    e^{-\kappa \bar x}(|v-\E[V]| - \kappa C_{1}(b(h(\bar x))^{1/p} + c|v-\E[V]|\bar{x})) = \mu(h'(\bar x))^{-1/(p-1)}.
  \end{cases}
\end{equation}
\end{proposition}

In comparison with the first scenario ($\alpha>1$), the present scenario with $\alpha=1$ uncovers dissimilar trading behaviors when criminal penalties are mild. First, if $p=1$ or $b=0$, meaning that the penalty structure is fairly even over time and no criminal charges are pursued, then, according to (\ref{P3:th}), the IT's focus shifts entirely to his cumulative order
\begin{equation}\label{Th_tar}
  \tl\Theta^{\star v}_{T}:=\sgn(v-\E[V])\frac{|v-\E[V]|}{C_{1}(b+c|v-\E[V]|)}
\end{equation}
at the time of announcing the asset value. This strategic freedom is understandably a direct result of the weak penalty structure, with the shallow temporal concentration to allow the IT to consider a wide range of timing strategies and the absence of criminal risk to motivate locally explosive strategies, for which the blowup can be timed anywhere up until the announcement. From (\ref{P3:th}), the target cumulative order $\tl\Theta^{\star v}_{T}$ is obviously inversely related to all three parameters $b,c,C_{1}$ tied to the strength of criminal or civil penalization, particularly for large trades. We also remark that in this edge case, the square-integrability condition is strictly binding, which is not strictly necessary to meet (\ref{Th_tar}), while (first-order) integrability on $[0,T]$ is enough.

Second, when $p>1$ and $b>0$, indicating that concentrated criminal penalties are within consideration, the unique equilibrium strategy in (\ref{P3:th2}), strictly increasing over time, inherits the explosive behavior from the previous case with $\alpha>1$, which is only observed in approaching the announcement date. In particular, it can be shown (see the proof in Appendix \ref{C}) that (\ref{P3:th}) satisfies
\begin{equation}\label{th_asym2}
  |\tl\th^{\star v}_{t}|\sim K(T-t)^{-1/(p+1)},\quad\text{as }t\nearrow T,
\end{equation}
for some positive constant $K$. Notably, this behavior can be regarded as a special case of (\ref{th_asym}) as $\alpha\searrow1$ (keeping $p>1$ though), which implies that when penalties remain modest ($\alpha=1$), the IT retains a strong incentive to pursue arbitrarily large trades, irrespective of the concentration degree $p>1$. The degree $p$ plays only a secondary role -- to influence the rate at which the IT's position escalates toward the announcement date, and any such escalation depends entirely on him having avoided detection up to that point.

We provide separate illustrations of (\ref{P3:th}) and (\ref{P3:th2}), as the former represents infinitely many equilibrium strategies. The following parametric choices are made: $T=1$, $p\in[3/2,2]$, $\kappa=C_{1}=1$, $b\in[1,3]$, and $c\in[1,3]$, along with the same assumption $\E[V]=\sqrt{e}$ and the choice $v=3$ to emphasize buy strategies. Regarding (\ref{P3:th}), in Figure \ref{fig:3} we give a 3D plot for the corresponding (unique) target cumulative order (\ref{Th_tar}) with respect to $b$ and $c$, simply to highlight the aforementioned inverse relations. On the other hand, implementing (\ref{P3:th2}) is a computationally intense task, as it requires solving the (forward-backward) ODE system (\ref{P3:DS}) numerically, which can be accomplished by using a two-parameter shooting algorithm. %\textcolor{blue}{(optional detail in Appendix)}.
For this reason, in Figure \ref{fig:3}, we only visualize the paths of $\tilde{\theta}^{\star v}$ by varying $p\in\{3/2,7/4,2\}$, while fixing $b=2$ and $c=1$.

% \textcolor{orange}{[Figure 3 and Figure 4 about here.]}

\begin{figure}[H]
  \centering
  \includegraphics[scale=0.35]{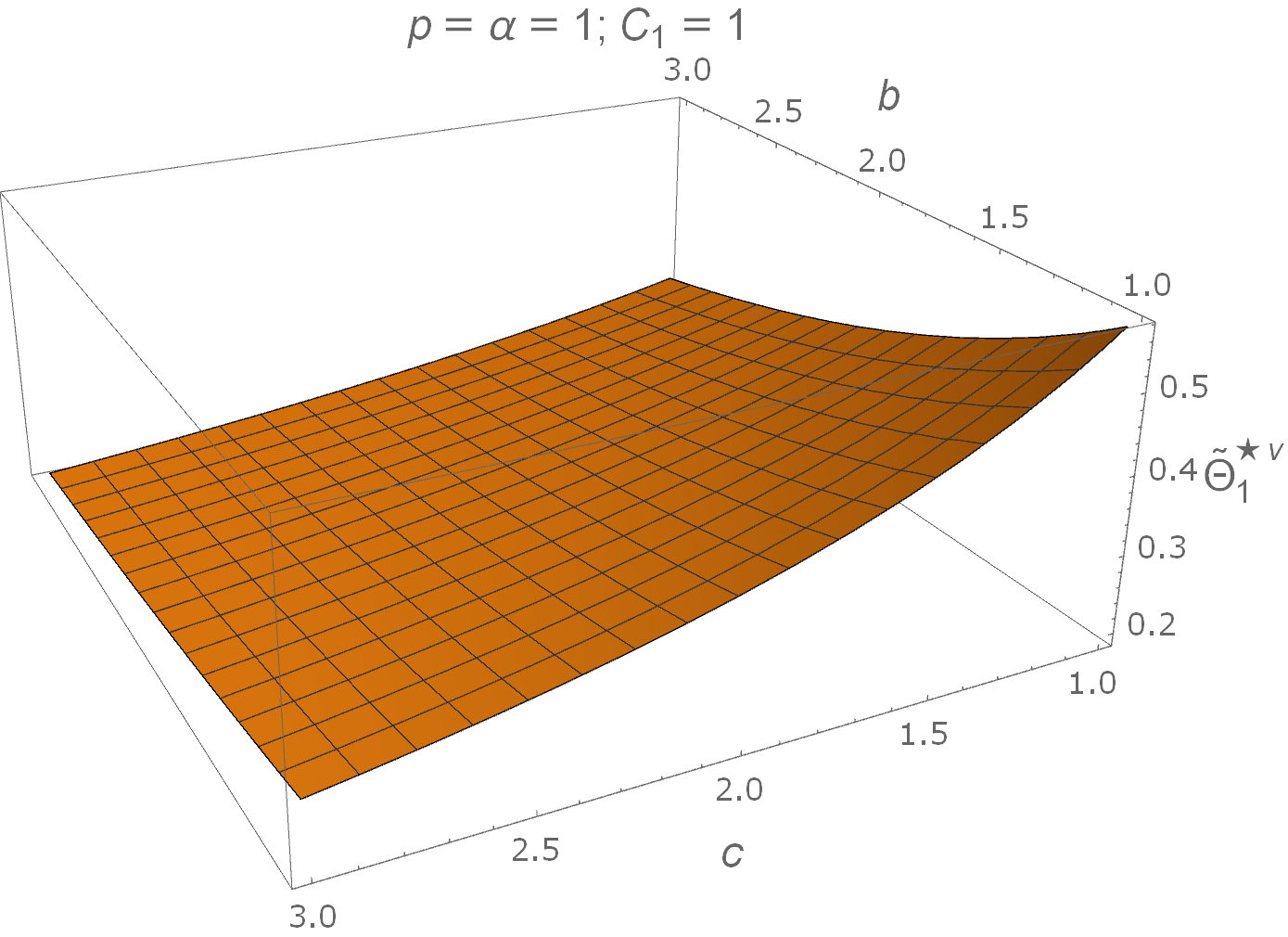}
  \includegraphics[scale=0.35]{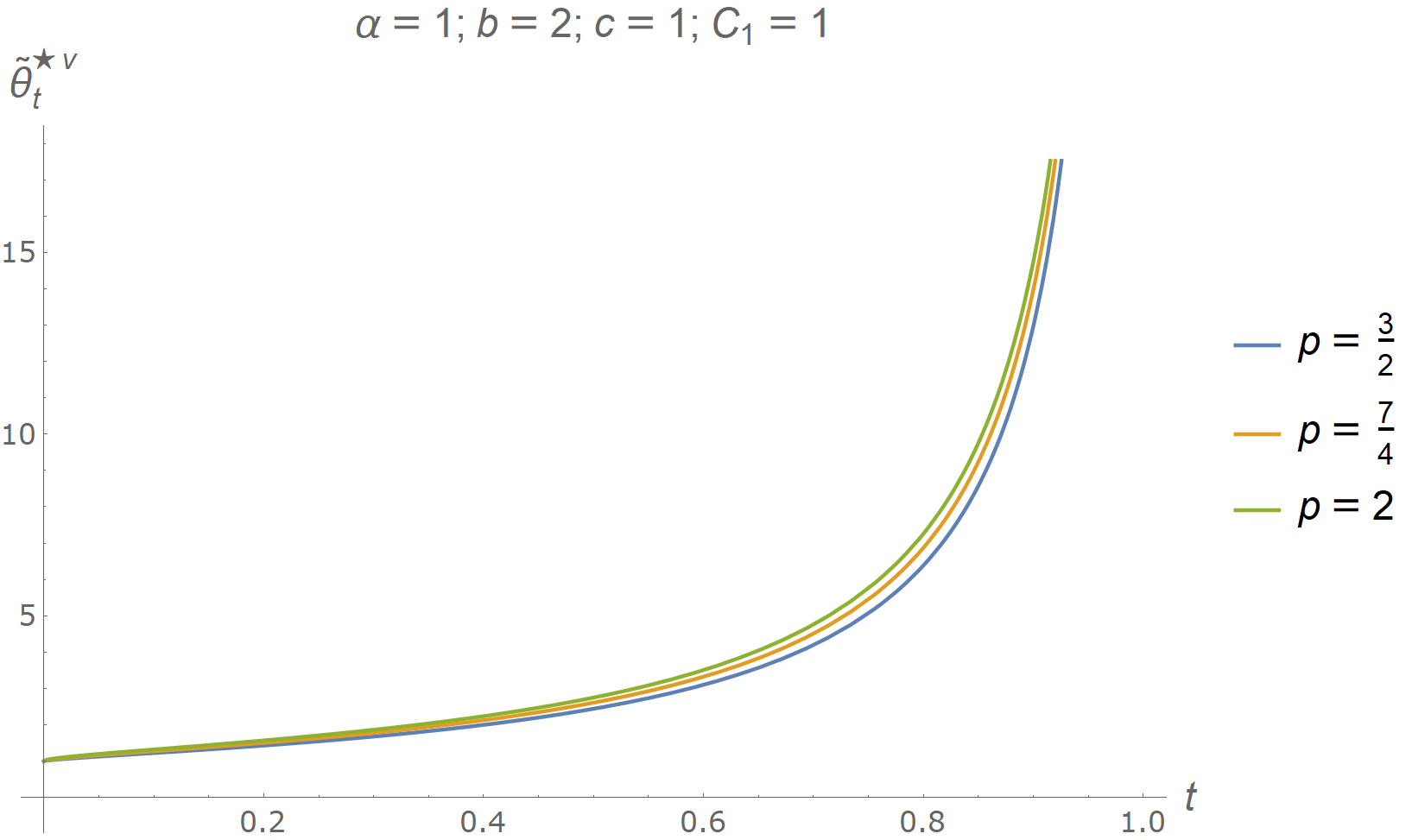}\\
  \caption{Limiting equilibrium strategies and legal risk impact (III)}\label{fig:3}
\end{figure}

Beyond confirming the monotonicity of the equilibrium strategies, Figure \ref{fig:3} further points to an ambiguous effect of the concentration degree $p$ under weak criminal penalties, as the paths appear glued to each other. On a closer look, we see that, at least locally, an increase in $p$ can even elevate the IT's trading intensity -- somewhat paradoxically -- particularly in the initial stages. This suggests that when large insider trades face insufficient penalties, concentrating enforcement on periods of heightened activity fails as an effective deterrent, and this is especially true at the onset of trading, where the IT perceives legal risk as still latent. Meanwhile, as evidenced by Figure \ref{fig:3}, civil penalties alone lack the definitive capacity to influence the IT's temporal allocation of trade sizes. We should also note that Figure \ref{fig:3} itself does not adequately display the strategic blowup behavior as $t\nearrow T$ in (\ref{th_asym2}), which is numerically verified as well.\footnote{For example, for the given set of parameter values, at $t=1-10^{-5}$, we compute that $\tilde{\theta}^{\star v}_{t}\approx 135497, 127866, 105855$ for $p=3/2,7/4,2$, respectively.}

\smallskip

We seal this section with three additional remarks.

\begin{remark}\label{rem:6}
As mentioned since Section \ref{S:1}, the three coefficients $\beta,\eta,\alpha$ primarily determine the level of stealth that the IT exercises in trading, as directly measured by the stealth index $\gamma$ in (\ref{og}), as they are all associated -- in one way or another -- to the population size $N$ of the NTs. More specifically, $\beta$ reflects the obscuring effect of $N$ on the hazard rate, hence the probability of detection and prosecution, $\eta$ measures the RG's difficulty in detecting relatively small insider trades and is indirectly affected by the obscuring effect, and $\alpha$ is linked to the severity of criminal penalties applied to large trades, conditional on achieving a suitable stealth level. The additional parameters $p,b,c,C_{1},C_{2}$ then operate atop these coefficients and are functional only once stealth is established (with $\gamma$ determined). Among them, $p$ and $b$ govern strategy-focused criminal penalties, measuring their temporal concentration and acting as an indirect multiplier, respectively, and $c$ is the (direct) multiplier for profit-dependent civil penalties, while $C_{1},C_{2}$ stem from the aggregation function $W$, reflecting its relative weighting of the two penalty types, particularly in large-trade scenarios.
\end{remark}

\begin{remark}\label{rem:7}
The proofs of the three propositions in this section, as shown in Appendix \ref{C}, rely heavily on Pontryagin's maximum principle (see, e.g., Zhou \citeyear{Z90}) for (path-dependent) deterministic control problems. Alternatively, they can be approached using Hamilton--Jacobi--Bellman equations. However, guessing an ansatz for the associated value function from the outset is a rather intricate task, especially for the cases considered in Proposition \ref{pro:5} and Proposition \ref{pro:6}, despite the possibility of a-posteriori verification.
\end{remark}

\begin{remark}\label{rem:8}
Under stealth trading, explicitly modeling the RG's objective, such as minimizing average losses to normal traders or maximizing post-trade price efficiency, becomes much less critical. More precisely, in the finite-$N$ equilibrium, the total expected wealth of all the NTs combined at the conclusion of trading is given by
\begin{align*}
  L_{N}&=\E\bigg[\sqrt{N}\int^{\tau\wedge T}_{0}\sigma_{t}\dd B_{t}\;(V-P^{\star}_{N,\tau\wedge T})+\int^{\tau\wedge T}_{0}\sqrt{N}\int^{t}_{0}\sigma_{s}\dd B_{s}\dd P^{\star}_{N,t}\bigg] \\
  &=\sqrt{N}\E\bigg[\int^{\tau\wedge T}_{0}\sigma_{t}(V-P^{\star}_{N,t})\dd B_{t}\bigg]-N\E\bigg[\int^{\tau\wedge T}_{0}\sigma^{2}_{t}P^{\star(0,1)}_{N,t}\dd t\bigg],
\end{align*}
where $P^{\star}_{N}$ is the equilibrium pricing rule (Definition \ref{eq.def1}), and $P^{\star(0,1)}_{N}$ is its (presumably existent) derivative process based on (\ref{MM.p1}). Then, the scaled total wealth $\tilde{L}:=\lim_{N\to\infty}N^{-1/2}L_{N}=\E\big[\int^{\tau}_{0}\sigma_{s}(V-\E[V])\dd B_{s}\big]=0$, and similarly, in the same limit, the post-trade variance of the asset value is $\E[\Var(V)]=\Var(V)$, both completely deterministic. This is intuitive: If the IT is sufficiently sophisticated (as guaranteed under optimality) to select a stealth level that asymptotically eliminates price impact and minimizes harm to the NTs, then pursuing such regulatory goals becomes largely quixotic in equilibrium. Technically, this property prevails as long as $\beta<1/2$, meaning that the RG is capable of detecting (illegal) insider trades at all, however large the scale of concurrent trading is, or whenever $\alpha\geq1$, meaning that significant criminal penalties are actively sought. Nevertheless, this section has demonstrated that the RG can indirectly shape the IT's behavior, including its temporal flexibility and explosive tendencies, by designing specific detection mechanisms and penalty structures, even when the IT maintains sophisticated stealth. On the other hand, since the limiting equilibrium is only an approximation of the true finite-population one (in light of Theorem \ref{theorem:s}), a further regulatory insight is that adjusting the three coefficients $\beta,\eta,\alpha$ can slow the convergence rate -- if mitigating price inefficiency from illicit trading must be a primary goal. In connection with this, the present analysis intentionally departs from explicit regulatory impact on price efficiency in a legal-risk setting (in contrast to the optimal regulation considerations in Carr\'{e}, Collin-Dufresne, and Gabriel \citeyear{CC-DG22} \text{Sect.} 4 and \c{C}etin \citeyear{C25} \text{Sect.} 6) (see Footnote 1), offering instead a novel perspective, alongside empirical grounding.

% In this case, even if the RG imposes very weak penalty, $\alpha=0$ (even bounded), we have $\gamma=\beta$, which is still possible if $\beta<1/2$. RG's goal is to dually increase likelihood of detection and collecting penalty. Under deterministic insider strategy, the expected penalty is $\int^{T}_{0}\lambda(t,\theta_{t})e^{-\int^{t}_{0}\lambda(s,\theta_{s})\dd s}W\big(\int^{t}_{0}\varpi_{0}(\theta_{s})\dd s,\int^{t}_{0}\theta_{s}(V-\E[V])\dd s\big)\dd t$ if $\gamma=\beta=0$, and so we can comment on the effect of different coefficients on the functional forms of resulting expected penalty. For instance, if $\gamma=\beta>0$, then the functional form of $W$ does not make any difference, subject to the same asymptotic behavior; if $\gamma<\beta$, then imposing civil penalty is futile, etc.
\end{remark}

%Wealth of the NTs:
%Price efficiency:
%\begin{equation*}
%  \E[\Var[V|\upsigma(Q_{[0,\tau]})]]
%\end{equation*}

% Note that if $N\to\infty$, $\tilde{L}=\E\big[\int^{\tau}_{0}\sigma_{s}(V-\E[V])\dd B_{s}\big]=0$. This means that under stealth trading, RG can do nothing to improve the NTs' expected loss, which is already minimal at 0. This appears to suggest that in the limit setting, the RG's goal is not to minimize noise traders' loss or maximize price efficiency, but to choose $\beta,\eta,\alpha$ to the possible extent for a lower convergence rate (?).}

%\textcolor{purple}{Some key points}

\medskip

\section{Conclusions}\label{S:5}

Dynamic legal risk plays a pivotal role in shaping illegal insiders' trading behavior in a continuous-time trading horizon, as it encompasses a wide spectrum of practical detection mechanisms and penalty structures. By using the alternative formulation, a technique new to the economics literature, the insider trading problem can be cast in a form in which the market maker's pricing rule becomes essentially exogenous. The validity of doing so is rooted in a fundamental change of measure, which transforms the physical measure $\PP$ into the impact-neutral measure $\PP^{\th^v}$, absorbing the insider's price impact into market expectations. Such a transformation greatly facilitates equilibrium analysis with the pricing rule liberated from the endogenous order flow process $Q$ (Definition \ref{eq.def2}), enabling a BSDE-type characterization of the equilibrium (Theorem \ref{theorem:p} statement (i)). Critically, by allowing the key model components to depend reasonably on the population size of liquidity traders, we have also shown that the insider can adapt his stealth level to mitigate price impact, consistent with established empirical evidence. This consideration has led to the limiting equilibrium in Definition \ref{leq.def}, in which the market maker's pricing rule becomes constant (under mild parametric conditions), with the stealth level (gauged by the index $\gamma$) jointly determined by the effect of noise-trading volume on legal detection effort and the specific penalty severity applied to large insider trades (Theorem \ref{theorem:p} statement (ii)). Under the same conditions, in the limiting equilibrium, the insider adopts a deterministic strategy that is ``risk-free'' in terms of price impact, yet adapts to persistent legal risk. The practical relevance of analyzing limiting equilibria -- rather than their finite-population counterparts -- is further reinforced by its strong approximation power (as an $\epsilon$-equilibrium) in (empirically supported) large-population settings (Definition \ref{eeq.def} and Theorem \ref{theorem:s}), which provides a convenient foundation for examining a broad class of regulatory measures on insider trading.

Through the examination of limiting equilibria in different illustrative scenarios, we have shown that regulatory measures exert a significant influence on an insider's trading behavior, in spite of the insider's ability to properly hide his trade information. The three propositions in Section \ref{S:4} have clarified that: First, the insider does not necessarily trade a constant multiple of the gap between his valuation and the market price of the asset, and the trading intensity can sharply increase towards the end of the horizon; second, legal sanctions are by no means an adequate substitute for incompetence in regulatory diligence; third, strategy-focused criminal penalties are particularly effective in deterring aggressive insider trading compared with criminal penalties relying on illicit profit, exhibiting unique temporal constraints on the insider's trading intensity.

This paper also delineates various avenues for future research. For instance, the framework can be extended beyond a single representative insider to model collaboration or competition within a group, thereby examining how legal enforcement influences collective trading behavior. Besides, embedding the regulator's optimization problem -- where investigation intensity and penalty functions are chosen strategically -- into the model would allow the insider's optimization to treat these features as endogenous; on surface, this approach would yield forward-looking insights into optimal regulatory design and shift the perspective from reactive enforcement to preemptive policy. Additionally, though less important, explicitly solving the fixed-point system governed by the BSDEs in (\ref{NC1}) and (\ref{BSDE2}) for characterizing the finite-population equilibrium under special circumstances may be of interest for the purpose of discovering explicit functional forms of insider strategies. We hope that the present work has provided a comprehensive framework for analyzing continuous-time insider trading under legal risk, considering stealth trading and imposing practically minimal restrictions on the structure of investigations and penalties. The proposed model framework reveals new facets of strategic insider trading behavior in response to legal configurations -- insights that would be otherwise unattainable under prior frameworks -- together with important implications for regulating sophisticated, stealth-achieving insiders.

\bigskip

\begin{appendices}

\section{Proofs in Section \ref{S:2}}\label{A}

\renewcommand{\theequation}{A.\arabic{equation}}

\noindent\textbf{Proof of Proposition \ref{pro:1}.}
\begin{proof}
By (\ref{pe.s1}), we have
\begin{align*}
  J(P;\theta^{v},v)&=\E\bigg[\int^{T}_{0}\mathds{1}_{\{\tau>T\}}\theta^{v}_{t}(v-P(t,Q^{\theta^{v}}_{[0,t]}))\dd t-\mathds{1}_{\{\tau\leq T\}}\Pi_{\rm a}\bigg(\theta^{v}_{[0,\tau]\!]},\int^{\tau}_{0}\theta^{v}_{t}(v-P(t,Q^{\theta^{v}}_{[0,t]}))\dd t\bigg)\bigg] \\
  &=:\mathfrak{E}_{1}-\mathfrak{E}_{2},
\end{align*}
noting $\Pi_{\rm a}\big(\theta^{v}_{[0,\tau]\!]},\int^{\tau}_{0}\theta^{v}_{t}(v-P(t,Q^{\theta^{v}}_{[0,t]}))\dd t\big)=\Pi\big(\theta^{v}_{[0,\tau]\!]},\int^{\tau}_{0}\theta^{v}_{t}(v-P(t,Q^{\theta^{v}}_{[0,t]}))\dd t\big)-\int^{\tau}_{0}\theta^{v}_{t}(v-P(t,Q^{\theta^{v}}_{[0,t]}))\dd t$.
From (\ref{th}), (\ref{de.s}), and (\ref{Q.e1}), it is clear that conditional on $\upsigma(B_{[0,T]})$, $M_{\Lambda}=M_{\Lambda^{\theta^{v}}}$ is independent from $\theta^{v}$ and $Q^{\theta^{v}}$, with $v\in\mathcal{V}$ given. Since $\{\tau>T\}=\{M_{\Lambda_{T}}=0\}$, by applying the law of iterated expectations it follows that
\begin{align*}
  \mathfrak{E}_{1}&=\E\bigg[\E\bigg[\mathds{1}_{\{\tau>T\}}\int^{T}_{0}\theta^{v}_{t}(v-P(t,Q^{\theta^{v}}_{[0,t]}))\dd t\bigg|\upsigma(B_{[0,T]})\bigg]\bigg] \\
  &=\E\bigg[\PP[M_{\Lambda^{\theta^{v}}_{T}}=0|\upsigma(B_{[0,T]})]\int^{T}_{0}\theta^{v}_{t}(v-P(t,Q^{\theta^{v}}_{[0,t]}))\dd t\bigg] \\
  &=\E\bigg[e^{-\Lambda^{\theta^{v}}_{T}}\int^{T}_{0}\theta^{v}_{t}(v-P(t,Q^{\theta^{v}}_{[0,t]}))\dd t\bigg],
\end{align*}
and similarly,
\begin{align*}
  \mathfrak{E}_{2}&=\E\bigg[\E\bigg[\mathds{1}_{\{\tau\leq T\}}\Pi_{\rm a}\bigg(\theta^{v}_{[0,\tau]\!]},\int^{\tau}_{0}\theta^{v}_{t}(v-P(t,Q^{\theta^{v}}_{[0,t]}))\dd t\bigg)\bigg|\upsigma(B_{[0,T]})\bigg]\bigg] \\
  &=\E\bigg[\int^{T}_{0}\Pi_{\rm a}\bigg(\theta^{v}_{[0,t]},\int^{t}_{0}\theta^{v}_{s}(v-P(s,Q^{\theta^{v}}_{[0,s]}))\bigg) \dd\PP[M_{\Lambda^{\theta^{v}}_{t}}\geq1|\upsigma(B_{[0,T]})]\bigg],
\end{align*}
for which
\begin{equation*}
  \frac{\dd\PP[M_{\Lambda^{\theta^{v}}_{t}}\geq1|\upsigma(B_{[0,T]})]}{\dd t}=\frac{\dd(1-e^{-\Lambda^{\theta^{v}}_{t}})}{\dd t}=\lambda^{\theta^{v}}_{t}e^{-\Lambda^{\theta^{v}}_{t}},\quad\PP\text{-a.s.}
\end{equation*}
Combining $\mathfrak{E}_{1}$ and $\mathfrak{E}_{2}$, we have
\begin{equation*}
  J(P;\theta^{v},v)=\E\bigg[e^{-\Lambda^{\theta^{v}}_{T}}\int^{T}_{0}\theta^{v}_{t}(v-P(t,Q^{\theta^{v}}_{[0,t]}))\dd t-\int^{T}_{0}\lambda^{\theta^{v}}_{t}e^{-\Lambda^{\theta^{v}}_{t}} \Pi_{\rm a}\bigg(\theta^{v}_{[0,t]},\int^{t}_{0}\theta^{v}_{s}(v-P(s,Q^{\theta^{v}}_{[0,s]}))\dd s\bigg)\dd t\bigg],
\end{equation*}
which after applying integration-by-parts to the first integral gives (\ref{IT.o2}).
\end{proof}

\medskip

\noindent\textbf{Proof of Proposition \ref{pro:2}.}
\begin{proof}
For the IT, given $P$ and $V=v$, (\ref{Q.e2}) immediately allows to rewrite (\ref{IT.o2}) into (\ref{IT.o4}) by the measure change in (\ref{mc}).

For the MM, take $\theta$ as given. By (\ref{MM.p1}), we know by the definition of conditional expectations that for any $t\in[0,T]$ and any random variable $\varphi(t)\in\mathbb{L}^{\infty}_{\upsigma(Q_{[0,t]})}(\Omega;\mathbb{R})$,
\begin{equation*}
  \E[V\varphi(t)]-\E[P^{\theta}(t,Q_{[0,t]})\varphi(t)]=\E[(V-\E[V|\upsigma(Q_{[0,t]})])\varphi(t)]=0,\quad t\in[0,T].
\end{equation*}
Then, by the total expectation law and the measure change again,
\begin{align*}
  0&=\E[V\varphi(t)]-\E[P^{\theta}(t,Q_{[0,t]})\varphi(t)] \\
  &=\E[\E[(V-P^{\theta}(t,Q_{[0,t]}))\varphi(t)|V]] \\
  &=\int_{\mathcal{V}}\E^{\theta^{v}}\bigg[\bigg(v-P^{\theta}\bigg(t,\bigg(\sqrt{N}\int^{\cdot}_{0}\sigma_{s}\dd B_{s}\bigg)_{[0,t]}\bigg)\bigg)\varphi(t)\bigg]\dd\PP[V\leq v] \\
  &=\int_{\mathcal{V}}\E[X^{\theta^{v}}_{t}(v-P^{\theta}_{t})\varphi(t)]\dd\PP[V\leq v] \\
  &=\E\bigg[\varphi(t)\int_{\mathcal{V}}X^{\theta^{v}}_{t}(v-P^{\theta}_{t})\dd\PP[V\leq v]\bigg],
\end{align*}
where the last equality follows from the Fubini--Tonelli theorem, which is applicable because
\begin{align*}
  \int_{\mathcal{V}}\E[|X^{\theta^{v}}_{t}(v-P^{\theta}_{t})\varphi(t)|]\dd\PP[V\leq v] &\leq\|\varphi\|_{\dbL^\infty}\int_{\mathcal{V}}\E[|v-P(t,Q_{[0,t]}^{\th^v}\big)|]\dd\PP[V\leq v]\\
  &\leq \|\varphi\|_{\dbL^\infty}(\E [|V|]+\E[|P(t,Q_{[0,t]}^{\th^V})|]),
\end{align*}
finite by condition \eqref{PL2} in Definition \ref{as.def}. Therefore, the arbitrariness of $\varphi(t)$ ensures that $\PP$-a.s.,
\begin{equation*}
  \int_{\mathcal{V}}X^{\theta^{v}}_{t}(v-P^{\theta}_{t})\dd\PP[V\leq v]=0,
\end{equation*}
or (\ref{MM.p2}).
\end{proof}

\medskip

\noindent\textbf{Proof of Proposition \ref{pro:3}.}
\begin{proof}
For fixed $v\in\cV$ and $\theta^{v}\in\cA$, let us define
\begin{equation*}
  Y_t^{\th^v}:=(\G_t^{\th^v})^{-1}\E\bigg[\int_t^T\G_s^{\th^v}\bigg(\th_s^v(v-P_s^{\th}) -\l_s^{\th^v}\Pi_{\rm a}\bigg(\th_{[0,s]}^v,\int_0^s\th_r^v(v-P_r^{\th})\dd r\bigg)\bigg)\dd s\bigg|\mathscr{F}_t\bigg],\;\;t\in[0,T],
\end{equation*}
where
\begin{equation}\label{Gamma}
  \G_t^{\th^v}:=\mathcal{E}\bigg(\int_0^\cdot\frac{\th^v_s}{\sqrt{N}\si_s}\dd B_s-\int_0^\cdot\l_s^{\th^v}\dd s\bigg)_{t}=X_t^{\th_v}e^{-\L_t^{\th^v}},
\end{equation}
with $X^{\th^v}$ given by (\ref{dp}). Via the standard theory of BSDEs (e.g., Zhang \citeyear{Z17} \text{Prop.} 4.1.2), it is familiar that $Y^{\th^v}$ solves the BSDE in the form (\ref{BSDE1}). Then, taking $t=0$, it follows that
\begin{align*}
  Y_0^{\th^v}&=\E\bigg[\int_0^T\G_s^{\th^v}\bigg[\th_s^v(v-P_s^{\th^v})-\l_s^{\th^v}\Pi_{\rm a}\bigg(\th_{[0,s]}^v,\int_0^s\th_r^v(v-P_r^{\th^v})\dd r\bigg)\bigg]\dd s\bigg] \\
  &=\E^{\theta^{v}}\bigg[\int^{T}_{0}\bigg[e^{-\Lambda^{\theta^{v}}_{t}}\theta^{v}_{t}(v-P_t^{\th^v})\dd t-\lambda^{\theta^{v}}_{t}e^{-\Lambda^{\theta^{v}}_{t}} \Pi_{\rm a}\bigg(\theta^{v}_{[0,t]},\int^{t}_{0}\theta_s^{\th^v}(v-P_s^{\th^v})\dd s\bigg)\bigg]\dd t\bigg] \\
  &=J(P;\theta^v,v),
\end{align*}
where the second equality uses the measure change in (\ref{mc}).
\end{proof}

\medskip

\section{Proofs in Section \ref{S:3}}\label{B}

\renewcommand{\theequation}{B.\arabic{equation}}

Before proving Theorem \ref{theorem:p}, we establish two technical lemmas.

\begin{lemma}\label{XUI}
Suppose $\gamma<1/2$. If $\tl\th^v\in\cA$ (Definition \ref{as.def}) for $v\in\cV$ given, then the collection of random variables $\{X^{N^{\gamma}\tilde{\theta}^{v}}_{T}:\;N\geq1\}$ is uniformly integrable.
\end{lemma}

\begin{proof}
Let $v\in\cV$ be given. By applying the Cauchy--Schwarz inequality and the Markov inequality, we have that for any constant $K>0$,
\begin{equation*}
  \dbE\big[X^{N^{\gamma}\tilde{\theta}^{v}}_{T}\1_{\{X^{N^{\gamma}\tilde{\theta}^{v}}_{T}>K\}}\big]\leq \big(\dbE\big[(X^{N^{\gamma}\tilde{\theta}^{v}}_{T})^2\big]\dbE[\1_{\{X^{N^{\gamma}\tilde{\theta}^{v}}_{T}>K\}}]\big)^{1/2}\leq \frac{\dbE\big[(X^{N^{\gamma}\tilde{\theta}^{v}}_{T})^2\big]}{K},
\end{equation*}
and so it suffices to show that $\dbE\big[(X^{N^{\gamma}\tilde{\theta}^{v}}_{T})^2\big]$ is bounded in $N$. By \eqref{dp} and the Cauchy--Schwarz inequality again,
\begin{align*}
  \dbE\big[(X^{N^{\gamma}\tilde{\theta}^{v}}_{T})^2\big]&=\dbE\big[e^{2\int_0^T N^{\g-1/2}\tl\th_t^v\dd B_t-\int_0^T(N^{\g-1/2}\tl\th_t^v)^2\dd t}\big] \\
  &\leq\Big(\dbE\big[e^{4\int_0^T N^{\g-1/2}\tl\th_t^v\dd B_t-8\int_0^T (N^{\g-1/2}\tl\th_t^v)^2\dd t}\big]\dbE\big[ e^{6N^{2\g-1}\int_0^T(\tl\th_t^v)^2\dd t}\big]\Big)^{1/2} \\
  &=\dbE\big[e^{6N^{2\g-1}\int_0^T(\tl\th_t^v)^2\dd t}\big]^{1/2},
\end{align*}
where the last equality uses that $e^{4\int_0^\cdot N^{\g-1/2}\tl\th_s^v\dd B_s-8\int_0^\cdot (N^{\g-1/2}\tl\th_s^v)^2\dd s}$ is an $(\mathbb{F},\PP)$-martingale with expected value $1$ (given Novikov's condition in \eqref{wf.adm}). Then, by consulting Zhang (\citeyear{Z17}) Theorem 7.2.3, there exist constants $\e>0$ and $C_\e>0$ depending only on $\|\tl\th^v\|_{\mathbb{L}^2}$ such that $\dbE[e^{\e\int_0^T(\tl\th^v_t)^2\dd t}]\leq C_\e$, with $\lim_{\e\to0}C_{\e}<\infty$. Since $\gamma-1/2<0$, then $M:=\sup_{N\geq1}\dbE\big[(X^{N^{\gamma}\tilde{\theta}^{v}}_{T})^2\big]\leq \sup_{N\geq1}C^{1/2}_{6N^{2\gamma-1}}<\infty$.
    % \begin{equation*}
    %     \begin{aligned}
    %         \dbE (X^{N^{\gamma}\tilde{\theta}^{v}}_{T})^2=&\dbE e^{2N^{\g-1/2}\int_0^T\tl\th^v_s\dd s-N^{2\g-1}\int_0^T(\tl\th^v_s)^2\dd B_s}\\
    %         =&\dbE e^{2N^{\g-1/2}\int_0^T\tl\th^v_s\dd B_s-2N^{2\g-1}\int_0^T(\tl\th^v_s)^2\dd s}e^{N^{2\g-1}\int_0^T(\tl\th^v_s)^2\dd s}\\
    %         \leq&(\dbE(X^{2N^{\gamma}\tilde{\theta}^{v}}_{T})^2)^\frac{1}{2}(\dbE)^\frac{1}{2}
    %     \end{aligned}
    % \end{equation*}
\end{proof}

\medskip

\begin{lemma}\label{XYUI}
In the setting of Lemma \ref{XUI}, if $\{Y^N:\;N\geq1\}$ is a sequence converging to $Y$ in $\dbL^2_{\mathcal{F}}(\Omega;\mathbb{R})$, then $\{X^{N^{\gamma}\tilde{\theta}^{v}}_{T}Y^N:\;N\geq1\}$ is also uniformly integrable.
\end{lemma}

\begin{proof} Observe that for all $N\geq1$ and any constants $K,L>0$,
\begin{align*}
  \dbE\big[X^{N^{\gamma}\tilde{\theta}^{v}}_{T}Y^N\1_{\{|X^{N^{\gamma}\tilde{\theta}^{v}}_{T}Y^N|>K\}}\big]&=\dbE \big[X^{N^{\gamma}\tilde{\theta}^{v}}_{T}Y^N\1_{\{|X^{N^{\gamma}\tilde{\theta}^{v}}_{T}Y^N|>K,|Y^N|>L\}}\big]+\dbE \big[X^{N^{\gamma}\tilde{\theta}^{v}}_{T}Y^N\1_{\{|X^{N^{\gamma}\tilde{\theta}^{v}}_{T}Y^N|>K,|Y^N|\le L\}}\big] \\
  &\leq\big(\dbE\big[(X^{N^{\gamma}\tilde{\theta}^{v}}_{T})^2\big]\dbE[(Y^N)^2\1_{\{|Y^N|>L\}}]\big)^{1/2}+\dbE \big[ LX^{N^{\gamma}\tilde{\theta}^{v}}_{T}\1_{\{|X^{N^{\gamma}\tilde{\theta}^{v}}_{T}|>K/L,|Y^N|\le L\}}\big] \\
  &\leq M\dbE[(Y^N)^2\1_{\{|Y^N|>L\}}]^{1/2}+L\dbE\big[X^{N^{\gamma}\tilde{\theta}^{v}}_{T}\1_{\{|X^{N^{\gamma}\tilde{\theta}^{v}}_{T}|>K/L\}}\big],
\end{align*}
where $M$ comes from the proof of Lemma \ref{XUI}.
Clearly, by the assumed $\mathbb{L}^{2}$-convergence, $\{Y^N:\;N\geq1\}$ is also uniformly integrable, while Lemma \ref{XUI} governs the uniform integrability of $\{X^{N^{\gamma}\tilde{\theta}^{v}}_{T}:\;N\ge1\}$. Based on the last inequality, for any $\e>0$, we can choose $L$ and $K$ large enough sequentially such that $\dbE\big[X^{N^{\gamma}\tilde{\theta}^{v}}_{T}Y^N\1_{\{|X^{N^{\gamma}\tilde{\theta}^{v}}_{T}Y^N|>K\}}\big]\leq\e$, hence the desired uniform integrability.
\end{proof}

\medskip

\noindent\textbf{Proof of Theorem \ref{theorem:p}.}

\begin{proof} For statement (i), let $N\geq1$ be fixed. First, by Definition \ref{eq.def2}, given any equilibrium trading strategy $\theta^\star_N$, the equilibrium pricing rule $P^\star_{N}$ is obtained from (\ref{MM.p2}) in Proposition \ref{pro:2} directly. On the other hand, given any (equilibrium pricing rule) $P^\star_N$, we adopt the stochastic maximal principle to derive necessary conditions that $\theta^\star_{N}$ satisfies.
% By Definition \ref{eq.def2}, IT wants to solve the maximization problem $\sup_{\th}J(P^\star_N;\th,v)$ for all $v\in\cV$. Fix $\th\in\cA_v$ $v\in\cV$,

Given $v\in\mathcal{V}$, for any $\th^v,\D\th^v\in\cA$, let us define $DY_{t}^{\th^v,\D\th^v}:=\lim_{\e\to0}(Y_t^{\th^v+\e\D\th^v}-Y_t^{\th^v})/\e$, which limit exists by Assumption \ref{as:1} and Assumption \ref{as:2} as well as the integrability conditions in Definition \ref{as.def}. Then, employing a perturbation argument, we see that for $\theta^{v}\in\cA$, $DY^{\th^{v},\D\th^{v}}$ satisfies the following BSDE:\footnote{In this case, by consulting Cvitani\'{c} and Zhang (\citeyear{CZ13}) Lemma 10.1.8, the BSDE \eqref{DY:BSDE} is well-posed.}
\begin{align}\label{DY:BSDE}
  DY_{t}^{\th^v,\D\th^v}=&\int_t^{T}\bigg[\D\th_s^v(v-P_{N,s}^\star)-N^{-\b}\D\th_s^v\lambda^{(0,1)}(s,\theta^{v}_{s}) \nonumber\\
  &\qquad\times\bigg(W\bigg(\bigg(\int^{s}_{0}(\varpi_{0}(\theta^{v}_{r}))^{p}\dd r\bigg)^{1/p},\max\bigg\{0,c\int^{s}_{0}\theta^{v}_{r}(v-P^\star_{N,r})\dd r\bigg\}\bigg)+Y^{\theta^{v}}_{s}\bigg) \nonumber\\
  &\quad-\l^{\theta^{v}}_sW^{(1,0)}_{s}\bigg(\int^{s}_{0}(\varpi_{0}(\theta^{v}_{r}))^{p}\dd r\bigg)^{1/p-1} \nonumber\\
  &\qquad\times\int^{s}_{0}(\varpi_{0}(\theta^{v}_{r}))^{p-1} {\big(\varpi_{0}'(\theta^{v}_{r}+)\1_{\{\D\th^v_r>0\}}+\varpi_{0}'(\theta^{v}_{r}-)\1_{\{\D\th^v_r<0\}}\big)}\D\th_r^v\dd r \nonumber\\
  &\quad-\l^{\theta^{v}}_s\bigg(W^{(0,1)}_{s}{\big(\1_{\{c\int^{s}_{0}(v-P_{N,r}^\star)\th_r^v\dd r>0\}}+\1_{\{\int^{s}_{0}(v-P_{N,r}^\star)\th_r^v\dd r=0,\; \int^{s}_{0}(v-P_{N,r}^\star)\D\th_r^v\dd r>0\}}\big)}\bigg) \nonumber\\
  &\qquad\times\bigg(c\int^{s}_{0}(v-P_{N,r}^\star)\D\th_r^v\dd r\bigg)-\l^{\theta^{v}}_sDY_{s}^{\th^v,\D\th^v} \nonumber\\
  &\quad+\frac{\D\th_s^vZ_s^{\th^v}+\th_s^vDZ_{s}^{\th^v,\D\th^v}}{\sqrt{N}\si_s}\bigg]\dd s-\int_t^TDZ_{s}^{\th^v,\D\th^v}\dd B_s,
\end{align}
where $W^{(1,0)}_{s}$ is a shorthand for $W^{(1,0)}\big(\big(\int^{s}_{0}(\varpi_{0}(\theta^{v}_{r}))^{p}\dd r\big)^{1/p},\max\big\{0,c\int^{s}_{0}\theta^{v}_{r}(v-P^\star_{N,r})\dd r\big\}\big)$, and similarly for $W^{(0,1)}_{s}$.

Thus, with $\G^{\th^v}$ defined in \eqref{Gamma}, by consulting Zhang (\citeyear{Z17}) \text{Prop.} 4.1.2, we have that when $t=0$, (\ref{DY:BSDE}) implies that
\begin{align}\label{DY}
  DY_{0}^{\th^v,\D\th^v}=&\E\bigg[\int_0^{T}\G_s^{\th^v}\bigg[\D\th_s^v(v-P^\star_{N,s})-N^{-\b}\D\th_s^v\lambda^{(0,1)}(s,\theta^{v}_{s}) \nonumber\\
  &\qquad\times\bigg(W\bigg(\bigg(\int^{s}_{0}(\varpi_{0}(\theta^{v}_{r}))^{p}\dd r\bigg)^{1/p},\max\bigg\{0,c\int^{s}_{0}\theta^{v}_{r}(v-P^\star_{N,r})\dd r\bigg\}\bigg)+Y^{\theta^{v}}_{s}\bigg) \nonumber\\
  &\quad-\l^{\th^v}_sW^{(1,0)}_{s}\bigg(\int^{s}_{0}(\varpi_{0}(\theta^{v}_{r}))^{p}\dd r\bigg)^{1/p-1} \nonumber\\
  &\qquad\times\int^{s}_{0}(\varpi_{0}(\theta^{v}_{r}))^{p-1} {\big(\varpi_{0}'(\theta^{v}_{r}+)\1_{\{\D\th^v_r>0\}}+\varpi_{0}'(\theta^{v}_{r}-)\1_{\{\D\th^v_r<0\}}\big)}\D\th_r^v\dd r \nonumber\\
  &\quad-\l^{\th^v}_sW^{(0,1)}_{s}\bigg(c\int^{s}_{0}(v-P_{N,r}^\star)\D\th_r^v\dd r\bigg) \nonumber\\
  &\qquad\times\big(\1_{\{c\int^{s}_{0}(v-P_{N,r}^\star)\th_r^v\dd r>0\}}+\1_{\{\int^{s}_{0}(v-P_{N,r}^\star)\th_r^v\dd r=0,\; \int^{s}_{0}(v-P_{N,r}^\star)\D\th_r^v\dd r>0\}}\big)+{\frac{\D\th_s^vZ_s^{\th^v}}{\sqrt{N}\si_s}}\bigg]\dd s\bigg] \nonumber\\
  =&\E\bigg[\int_0^{T}\D\th_s^v\G_s^{\th^v}\bigg[(v-P_{N,s}^\star)+{\frac{Z_s^{\th^v}}{\sqrt{N}\si_s}}-N^{-\b}\lambda^{(0,1)}(s,\theta^{v}_{s}) \nonumber\\
  &\qquad\times\bigg(W\bigg(\bigg(\int^{s}_{0}(\varpi_{0}(\theta^{v}_{r}))^{p}\dd r\bigg)^{1/p},\max\bigg\{0,c\int^{s}_{0}\theta^{v}_{r}(v-P^\star_{N,r})\dd r\bigg\}\bigg)+Y^{\theta^{v}}_{s}\bigg) \nonumber\\
  &-c(v-P_{N,s}^\star)(\G_s^{\th^v})^{-1}\E\bigg[\int_s^T\G_r^{\th^v}\l^{\th^v}_rW_r^{(0,1)} \nonumber\\
  &\qquad\times{\big(\1_{\{c\int^{r}_{0}(v-P_{N,u}^\star)\th_u^v\dd u>0\}}+\1_{\{\int^{r}_{0}(v-P_{N,u}^\star)\th_u^v\dd u=0,\;\int^{r}_{0}(v-P_{N,u}^\star)\D\th_u^v\dd u>0\}}\big)}\dd r\bigg|\mathscr{F}_s\bigg] \nonumber\\
  &-(\varpi_{0}(\theta^{v}_{s}))^{p-1}(\G_s^{\th^v})^{-1} {\big(\varpi_{0}'(\theta^{v}_{s}+)\1_{\{\D\th^v_s>0\}}+\varpi_{0}'(\theta^{v}_{s}-)\1_{\{\D\th^v_s<0\}}\big)}\nonumber\\
  &\qquad\times\E\bigg[\int_s^T\G_r^{\th^v}\l^{\theta^v}_rW_r^{(1,0)}\bigg(\int^{r}_{0}(\varpi_{0}(\theta^{v}_{u}))^{p}\dd u\bigg)^{1/p-1}\dd r\bigg|\mathscr{F}_s\bigg]\bigg]\dd s\bigg],
\end{align}
where the second equality uses the law of iterated expectations and integration-by-parts to take the term $\Delta\theta^{v}_{r}$ out of the $r$-integrals. By setting in (\ref{DY})
\begin{align*}
  \bar Y^{\th^v}_{s}&:=(\G_s^{\th^v})^{-1}\E\bigg[\int_s^T\G_r^{\th^v}\l^{\th^v}_rW_r^{(0,1)}{\big(\1_{\{c\int^{r}_{0}(v-P_{N,u}^\star)\th_u^v\dd u>0\}}+\1_{\{\int^{r}_{0}(v-P_{N,u}^\star)\th_u^v\dd u=0,\; \int^{r}_{0}(v-P_{N,u}^\star)\D\th_u^v\dd u>0\}}\big)}\dd r\bigg|\mathscr{F}_s\bigg], \\
  \h Y^{\th^v}_{s}&:=(\G_s^{\th^v})^{-1}\E\bigg[\int_s^T\G_r^{\th^v}\l^{\theta^v}_rW_r^{(1,0)}\bigg(\int^{r}_{0}(\varpi_{0}(\theta^{v}_{u}))^{p}\dd u\bigg)^{1/p-1}\dd r\bigg|\mathscr{F}_s\bigg],\quad s\in[0,T],
\end{align*}
it follows that the couples $(\bar Y^{\th^v},\bar Z^{\th^v})$ and $(\h Y^{\th^v},\h Z^{\th^v})$ solve the BSDEs
\begin{align*}
  \bar Y^{\th^v}_{s}&=\int_s^T\bigg(\l_r^{\th^v}W_r^{(0,1)}{\big(\1_{\{c\int^{r}_{0}(v-P_{N,u}^\star)\th_u^v\dd u>0\}}+\1_{\{\int^{r}_{0}(v-P_{N,u}^\star)\th_u^v\dd u=0,\; \int^{r}_{0}(v-P_{N,u}^\star)\D\th_u^v\dd u>0\}}\big)}-\l_r^{\th^v}\bar Y_r^{\th^v}+\frac{\th^v_r\bar Z_r^{\th^v}}{\sqrt{N}\si_r}\bigg)\dd s \nonumber\\
  &\quad-\int_s^T\bar Z_r^{\th^v}\dd B_r, \nonumber\\
  \h Y^{\th^v}_{s}&=\int_s^T\bigg(\l_r^{\th^v}W_r^{(1,0)}\bigg(\int^{r}_{0}(\varpi_{0}(\theta^{v}_{u}))^{p}\dd u\bigg)^{1/p-1}-\l_r^{\th^v}\h Y_r^{\th^v}+\frac{\th^v_r\h Z_r^{\th^v}}{\sqrt{N}\si_r}\bigg)\dd r-\int_s^T\h Z_r^{\th^v}\dd B_r,
\end{align*}
respectively, which lead to (\ref{BSDE2}). We conclude that if $\th^\star$ is an equilibrium in Definition \ref{eq.def2}, then it must satisfy the first condition in (\ref{NC1}). This proves statement (i). \medskip

For statement (ii), we first argue the finiteness of the scaled objective function in (\ref{IT.lo}). According to Definition \ref{leq.def}, provided $\gamma<1/2$, we can fix $\tilde{P}^\star_\gamma=\E[V]$. Also, fix an arbitrary $v\in\mathcal{V}\setminus\{\E[V]\}$. For any $\tilde\th^v\in\cA$, note that
\begin{equation}\label{IT.lo2}
  \tilde{J}_{\gamma}(\tilde\th^v,v)=\lim_{N\to\infty}N^{-\gamma}J(\tilde{P}^\star_{\gamma};N^{\gamma}\tilde{\theta}^v,v)) =\lim_{N\to\infty}\E^{\theta^{v}}\bigg[\int^{T}_{0}e^{-\Lambda^{N^{\gamma}\tl\theta^{v}}_{t}}\tl\theta^{v}_{t}(v-\tilde P^\star_\g)\dd t-\int^{T}_{0}e^{-\Lambda^{N^{\gamma}\tl\theta^{v}}_{t}}f^{N,\tl\th^v}_t\dd t\bigg],
\end{equation}
where we denote
\begin{equation}\label{fN}
  f^{N,\tl\th^v}_t:=N^{-\g}\lambda(t,N^{\g-\b}\tl\theta^{v}_{t})W\bigg(\bigg(\int^{t}_{0}(\varpi_{0}(N^{\g}\tl\theta^{v}_{s}))^{p}\dd s\bigg)^{1/p},\max\bigg\{0,c\int^{t}_{0}N^{\g}\tl\theta^{v}_{s}(v-\tl P^\star_{\g})\dd s\bigg\}\bigg),\quad t\in[0,T].
\end{equation}
To show (\ref{og}), we shall consider several cases. \smallskip

\noindent\underline{Case 1}\quad Suppose $\beta=0$. We consider two subcases.

\noindent\underline{Subcase 1.1}\quad If $\gamma>0$ further, then by (\ref{de.s}) and Assumption \ref{as:2}, for any $t\in[0,T]$, $\lim_{N\to\infty}\lambda^{N^{\gamma}\tl\th^v}_{t}=\lim_{N\to\infty}\lambda(t,N^{\g}\tl\theta^{v}_{t})=\infty$ ($\PP$-a.s.), so that $e^{-\L_t^{\th^v}}\overset{\rm a.s.}{\to}0$, and the first integral in (\ref{IT.lo2}) goes to $0$ (a.s.) as $N\to\infty$. For the second integral, from Assumption \ref{as:3} we know that for any $t\in[0,T]$, $\varpi_{0}(N^{\g}\tl\theta^{v}_{t})=b N^{\g\a}|\tl\th_t^v|^{\a}+O(1)$ as $N\to\infty$. Then, with $\alpha\geq1$, it follows from the proof of Proposition \ref{pro:1} that for every $t\in[0,T]$, as $N\to\infty$,
\begin{equation*}
  \int^{T}_{0}e^{-\Lambda^{N^{\gamma}\tl\theta^{v}}_{t}}f^{N,\tl\th^v}_t\dd t\overset{\rm a.s.}{\to}
  \begin{cases}
    \displaystyle b\int^{T}_{0}|\tilde{\theta}^v_t|\dd t,&\quad\text{if }\alpha=1, \\
    \infty,&\quad\text{if }\alpha>1,
  \end{cases}
\end{equation*}
and so $\tilde{J}_{\gamma}(\tilde\th^v,v)\in\big\{-b\E\big[\int^{T}_{0}|\tilde{\theta}^v_t|\dd t\big],-\infty\big\}$ in (\ref{IT.lo2}). As a result, if $\tilde{J}_{\gamma}(\tilde\th^v,v)\neq0$ (with $b>0$), either $\tl\th^{\star v}=\mathbf{0}$ trivially or it does not exist. If $\tilde{J}_{\gamma}(\tilde\th^v,v)=0$ (with $b=0$), it violates the nonzero-value condition in (\ref{IT.lop}). In any case, the requirements in Definition \ref{leq.def} are violated.\footnote{In this case, any $\tilde{\theta}^\star$ with $\tilde\theta^{\star v}\neq\mathbf{0}$, $v\in\cV$, would be a limiting equilibrium strategy.} Thus, we must have $\g=\b=0$.

\noindent\underline{Subcase 1.2}\quad If $\gamma=0$, then $N$ only takes effect in the density process $X^{N^\g\tl\th}$ in \eqref{dp}, and we simply have
\begin{align}\label{J-tl.1}
  \tilde{J}_{0}(\tilde\th^v,v)=&\lim_{N\to\infty}J(\tl P^\star_{0};N^\g\tilde{\theta}^v,v)=\lim_{N\to\infty}\E\bigg[X_ T^{N^\g\tl\th^v}\int^{T}_{0}e^{-\Lambda^{\tl\theta^{v}}_{t}}\tl\theta^{v}_{t}(v-\tilde P^\star_0)\dd t-\int^{T}_{0}e^{-\Lambda^{\tl\theta^{v}}_{t}}f^{1,\tl\th^v}_t\dd t\bigg] \nonumber\\
  =&\E\bigg[\int_0^Te^{-\Lambda^{\tl\theta^{v}}_{t}}\tl\theta^{v}_{t}(v-\tilde P^\star_0)\dd t-\int^{T}_{0}e^{-\Lambda^{\tl\theta^{v}}_{t}}f^{1,\tl\th^v}_t\dd t\bigg],
\end{align}
with $f^{1,\tl\th^v}_t=\lambda(t,\tl\theta^{v}_{t})W\big(\big(\int^{t}_{0}(\varpi_{0}(\tl\theta^{v}_{s}))^{p}\dd s\big)^{1/p},c\int^{t}_{0}\tl\theta^{v}_{s}(v-\tl P^\star_{0})\dd s\big)$ from (\ref{fN}) and $\tl P^\star_{0}=\E[V]$. %{\color{blue}By the admissibility of $\th$, the integral part of \eqref{J-tl.1} is $\dbL^2$, thus the second equality hold by a variation of Lemma \ref{XYUI}. (???)}
By \eqref{fL2} and \eqref{PL2}, we have
\begin{equation*}
  \E\bigg[\bigg(\int^{T}_{0}e^{-\Lambda^{\tl\theta^{v}}_{t}}\Big(\tl\theta^{v}_{t}(v-\tilde P^\star_0)-f^{1,\tl\th^v}_t\Big)\dd t\bigg)^2\bigg]<\infty,
\end{equation*}
and thus by Lemma \ref{XYUI}, the limit in \eqref{J-tl.1} exists (due to uniform integrability), while the pointwise convergence in the second equality uses \eqref{dp.l}. Obviously, \eqref{J-tl.1} meets the requirements in Definition \ref{leq.def}.

In Subcase 1.2, note that there is no source of randomness within the expectation in (\ref{J-tl.1}), so it is sufficient to consider the trading strategy $\theta$ deterministic and treat the problem to maximize (\ref{J-tl.1}) as a deterministic control problem. In particular, under determinism, with Novikov's condition in \eqref{wf.adm} gone, the limiting admissibility set is $\mathbb{L}^{2}([0,T];\mathbb{R})$, but since $v-\E[V]\neq0$ is constant, the objective function in \eqref{J-tl.1} is positive only if $\tl\theta^{v}$ takes values in $\sgn(v-\E[V])\mathbb{R}_{+}$, and so we can effectively restrict to the set
\begin{equation}\label{ladm}
  \tl\cA:=\mathbb{L}^{2}([0,T];\sgn(v-\E[V])\mathbb{R}_{+}),
\end{equation}
and the control problem is
\begin{equation}\label{DCP-1}
  \sup_{\tl\th^v\in\tl\cA\setminus\{\mathbf{0}\}}\tilde{J}_{0}(\tilde\th^v,v) =\sup_{\tl\th^v\in\tl\cA\setminus\{\mathbf{0}\}}\bigg[\int_0^Te^{-\Lambda^{\tl\theta^{v}}_{t}}\tl\theta^{v}_{t}(v-\tilde P^\star_0)\dd t-\int^{T}_{0}e^{-\Lambda^{\tl\theta^{v}}_{t}}f^{1,\tl\th^v}_t\dd t\bigg],
\end{equation}
where \eqref{ladm} also allows the maximum operator in $f^{1,\tl\th^v}_{t}$ to be dropped, i.e., equivalently, $f^{1,\tl\th^v}_t=\lambda(t,\tl\theta^{v}_{t})W\big(\big(\int^{t}_{0}(\varpi_{0}(\tl\theta^{v}_{s}))^{p}\dd s\big)^{1/p},c\int^{t}_{0}\tl\theta^{v}_{s}(v-\tl P^\star_{\g})\dd s\big)$.

% \begin{equation}\label{DCP-1}
%   \sup_{\tl\th^v\in\tl\cA_{v}}\tilde{J}_{0}(\tilde\th,v)=\sup_{\tl\th^v\in\tl\cA_{v}}U(0,0,0,\tl\th^v),
% \end{equation}
% where for $x_1,x_3\geq0$ and $x_{2}\in\mathbb{R}$,
% \begin{align}\label{DCP-1a}
%   U(x_1,x_2,x_3,\tl\th^v)&=\int_0^T e^{-\xi_{3,t}^{\tl\th^v}}\bigg(\tilde{\theta}_t^v(v-\tl P^\star_{0})-\lambda(t,\tl\th^v_t)W\bigg((\xi_{1,t}^{\tl\th^v})^{1/p},c(v-\tl P^\star_{0})\xi_{2,t}^{\tl\th^v}\bigg)\bigg)\dd t, \nonumber\\
%   \xi_{1,t}^{\tl\th^v}&=x_1+\int_0^t |\tl\th_s^v|^{p}\dd s,\quad\xi_{2,t}^{\tl\th^v}=x_2+\int_0^t \tl\th_s^v\dd s,\quad\xi_{3,t}^{\tl\th^v}=x_3+\int_0^t \l(s,\tl\th_s^v)\dd s,\quad t\in[0,T].
% \end{align}
\smallskip

\noindent\underline{Case 2}\quad Suppose $\beta>0$. We consider three subcases.

\noindent\underline{Subcase 2.1} If $\gamma>\beta$, then for the same reason as in Subcase 1.1, $\tilde{J}_{\gamma}(\tilde\th,v)=0$ in (\ref{IT.lo2}), violating the nonzero value requirement, hence invalid.

\noindent\underline{Subcase 2.2}\quad Suppose $\g=\b$. Then, the first integral in (\ref{IT.lo2}) goes to $\int^{T}_{0}e^{-\Lambda^{\tl\th^v}_t}\tl\theta^{v}_{t}(v-\tilde P^\star_\g)\dd t$ (a.s.) as $N\to\infty$, so we only need to consider the second integral. As $N\to\infty$, $\PP$-\text{a.s.} for any $t\in[0,T]$, by Assumption \ref{as:3} we have that
\begin{equation}\label{fNorder1}
  f^{N,\tl\th^v}_t\sim C_{1}\lambda(t,\tl\th^v_t)\bigg(bN^{\g(\a-1)}\bigg(\int_0^t|\tl\th^v_s|^{p\a}\dd s\bigg)^{1/p}+N^{-\g}O(1)+\max\bigg\{0,c\int_0^t\tl\theta^{v}_{s}(v-\tl P^\star_{\gamma})\dd s\bigg\}\bigg).
\end{equation}
If $\alpha>1$, then from (\ref{fNorder1}), $f^{N,\tl\th^v}_t\overset{\rm a.s.}{\to}\infty$ as $N\to\infty$, and we have $\tl J_{\g}(\tilde\th,v)=\infty$ for all $\tl\th^v\neq\mathbf{0}$. Thus, let $\alpha=1$, in which case (\ref{fNorder1}) becomes
\begin{equation}\label{f1}
  f^{N,\tl\th^v}_t\overset{\rm a.s.}{\to}f_t^{\tl\th^v}:=C_{1}\lambda(t,\tl\th^v_t)\bigg(b\bigg(\int_0^t|\tl\th^v_s|^{p}\dd s\bigg)^{1/p}+\max\bigg\{0,c(v-\tl P^\star_{\g})\int_0^t\tl\theta^{v}_{s}\dd s\bigg\}\bigg),
\end{equation}
where the limit is well-defined and satisfies the requirements in Definition \ref{leq.def}. This case is compatible with (\ref{og}) for $\alpha=1$. Furthermore, by (\ref{f1}), as $N\to\infty$, $|f^{N,\tl\th^v}_t-f_t^{\tl\th^v}|^2=O(N^{-2\g})$ $\dbP$-\text{a.s.} for any $t\in[0,T]$, and thus actually $\lim_{N\to\infty}\E\big[\int_0^T|f^{N,\tl\th^v}_t-f_t^{\tl\th^v}|^2\dd t\big]=0$.

In Subcase 2.2, applying Lemma \ref{XYUI} again gives
\begin{equation}\label{J-tl.2}
\begin{aligned}
  \tilde{J}_{\g}(\tilde\th^v,v)=&\lim_{N\to\infty}\E\bigg[X_T^{N^{\g}\tl\th^v} \int^{T}_{0}e^{-\Lambda^{\tl\theta^{v}}_{t}}\bigg(\tl\theta^{v}_{t}(v-\tilde P^\star_\g)-f^{N,\tl\th^v}_t\bigg)\dd t\bigg]\\
  =&\E\bigg[\int^{T}_{0}e^{-\Lambda^{\tl\theta^{v}}_{t}}\bigg(\tl\theta^{v}_{t}(v-\tilde P^\star_\g)-C_{1}\lambda(t,\tl\th^v_t)\bigg(\bigg(\int_0^t|\tl\th^v_s|^{p}\dd s\bigg)^{1/p}+\max\bigg\{0,c(v-\tl P^\star_{\g})\int_0^t\tl\theta^{v}_{s}\dd s\bigg\}\bigg)\bigg)\dd t\bigg],
\end{aligned}
\end{equation}
whose maximization can be similarly considered a deterministic control problem in $\tl\cA$ from \eqref{ladm}, which, with the maximum operator in \eqref{J-tl.2} dropped, is specified as
\begin{equation}\label{DCP-2}
  \sup_{\tl\th^v\in\tl\cA\setminus\{\mathbf{0}\}}\tilde{J}_{\g}(\tilde\th^v,v) =\sup_{\tl\th^v\in\tl\cA\setminus\{\mathbf{0}\}}\bigg[\int^{T}_{0}e^{-\Lambda^{\tl\theta^{v}}_{t}}\bigg(\tl\theta^{v}_{t}(v-\tilde P^\star_\g)-C_{1}\lambda(t,\tl\th^v_t)\bigg(\bigg(\int_0^t|\tl\th^v_s|^{p}\dd s\bigg)^{1/p}+c(v-\tl P^\star_{\g})\int_0^t\tl\theta^{v}_{s}\dd s\bigg)\bigg)\dd t\bigg].
\end{equation}
% namely
% \begin{align}\label{DCP-2}
%   \sup_{\tl\th^v}\tilde{J}_{\g}(\tilde\th^v,v)&=\sup_{\tl\th^v}U(0,0,0,\tl\th^v), \nonumber\\
%   U(x_1,x_2,x_3,\tl\th^v)&=\int_0^T e^{-\xi_{3,t}^{\tl\th^v}}\bigg(\tilde{\theta}_t^v(v-P^\star_{\g})-\lambda(t,\tl\th^v_t)\bigg(\1_{\{\alpha=1\}} (\xi_{1,t}^{\tl\th^v})^{1/p}+c(v-\tl P^\star_{0})\xi_{2,t}^{\tl\th^v}\bigg)\bigg)\dd t, \nonumber\\
%   \xi_{1,t}^{\tl\th^v}&=x_1+\int_0^t |\tl\th_s^v|^{p}\dd s,\quad\xi_{2,t}^{\tl\th^v}=x_2+\int_0^t \tl\th_s^v\dd s,\quad\xi_{3,t}^{\tl\th^v}=x_3+\int_0^t \l(s,\tl\th_s^v)\dd s,
% \end{align}
% for $x_{1},x_{3}\geq0$, $x_2\in\mathbb{R}$, and $t\in[0,T]$, from where using the same argument as in Subcase 1.2 yields the existence of $\tl\th^{\star v}$ such that $\tilde{J}_{\g}(\tilde\th^{\star v},v)=\sup_{\tl\th^v\in\tl\cA_v}\tilde{J}_{\g}(\tilde\th^v,v)$.

\smallskip

% $f^{N,\tl\th^v}_t\overset{\rm a.s.}{\to}0$. In this case, $\tl J_{\g}(\tilde\th,v)=\E[\int^{T}_{0}\tl\theta^{v}_{t}(v-P^{\star}_\gamma)\dd t]$ for all $\tl\th$, and as a result $\sup_{\tilde{\theta}}\tilde{J}_{\gamma}(\tilde{\theta},v)=\infty$, also invalid. \smallskip

\noindent\underline{Subcase 2.3}\quad Suppose $\g<\b$. Then, since for any $t\in[0,T]$, $\lim_{N\to\infty}\lambda^{N^{\gamma}\tl\th^v}_{t}=\lim_{N\to\infty}\lambda(t,N^{\g-\beta}\tl\theta^{v}_{t})=0$ ($\PP$-a.s.) and thus $e^{-\L_t^{\th^v}}\overset{\rm a.s.}{\to}1$, the first integral in (\ref{IT.lo2}) tends to $\int^{T}_{0}\tl\theta^{v}_{t}(v-\tilde P^\star_\g)\dd t$ as $N\to\infty$, and so we consider the second integral as well. By the conditions on $\lambda$ and $W$ in Assumption \ref{as:3} again, we have that as $N\to\infty$ ($\PP$-\text{a.s.} for any $t\in[0,T]$),
\begin{align}\label{fNorder2}
  f^{N,\tl\th^v}_t&\sim \kappa C_{1}\bigg(bN^{\g(\eta+\a-1)-\b\eta}|\tl\th_t^v|^{\eta}\bigg(\int_0^t|\tl\th^v_s|^{p\a}\dd s\bigg)^{1/p}+N^{(\g-\b)\eta-\g}O(1) \nonumber\\
  &\quad+\max\bigg\{0,cN^{(\g-\b)\eta}|\tl\th_t^v|^{\eta}\int_0^t\tl\theta^{v}_{s}(v-\tl P^\star_{\gamma})\dd s\bigg\}\bigg).
\end{align}
By comparing the powers of $N$ in (\ref{fNorder2}), since $\g<\b$, we must have $\a>1$, or else the last term dominates and $\g=\b$. Assuming $\a>1$, if $\g(\eta+\a-1)>\b\eta$, then from (\ref{fNorder2}), clearly $f^{N,\tl\th^v}_t\overset{\rm a.s.}{\to}\infty$ as $N\to\infty$, and since $\g<\b$, it follows that $\tl J_{\g}(\tilde\th^v,v)=\infty$ for all $\tl\th^v\neq\mathbf{0}$, hence invalid. If $\g(\eta+\a-1)<\b\eta$, then $f^{N,\tl\th^v}_t\overset{\rm a.s.}{\to}0$. In this case, $\tl J_{\g}(\tilde\th^v,v)=\E\big[\int^{T}_{0}\tl\theta^{v}_{t}(v-\tl P^{\star}_\gamma)\dd t\big]$ for all $\tl\th$, and as $v\neq\E[V]=\tl P^{\star}_\gamma$, $\tl\th^{\star v}=\pm\infty$, also invalid. Thus, it can only be that $\g(\eta+\a-1)=\b\eta$, i.e., (\ref{og}) holds for $\a>1$, and we have from (\ref{fNorder2})
\begin{equation}\label{f2}
  f^{N,\tl\th^v}_t\overset{\rm a.s.}{\to}C_{2}|\tl\th_t^v|^{\eta}\bigg(\int_0^t |\tl\th^v_s|^{p\a}\dd s\bigg)^{1/p},
\end{equation}
where $C_{2}=\kappa bC_{1}$. The last limit is also well-defined and satisfies the requirements in Definition \ref{leq.def}.

In Subcase 2.3, with $\g<\b$ and $\a>1$, again, $e^{-\L_t^{\th^v}}\overset{\rm a.s.}{\to}1$ as $N\to\infty$, making the corresponding limit of \eqref{IT.lo2} even simpler. Using the same argument as in Subcase 1.2 and Subcase 2.2, Lemma \ref{XYUI} gives
\begin{equation}\label{J-tl.3}
  \tilde{J}_{\g}(\tilde\th^v,v)=\E\bigg[\int^{T}_{0}\bigg(\tl\theta^{v}_{t}(v-\tilde P^\star_\g)-C_{2}|\tl\th_t^v|^{\eta}\bigg(\int_0^t |\tl\th^v_s|^{p\a}\dd s\bigg)^{1/p}\bigg)\dd t\bigg],
\end{equation}
the maximization of which is equivalent to the deterministic control problem
\begin{equation}\label{DCP-3}
  \sup_{\tl\th^v\in\tl\cA\setminus\{\mathbf{0}\}}\tilde{J}_{\g}(\tilde\th^v,v) =\sup_{\tl\th^v\in\tl\cA\setminus\{\mathbf{0}\}}\int^{T}_{0}\bigg(\tl\theta^{v}_{t}(v-\tilde P^\star_\g)-C_{2}|\tl\th_t^v|^{\eta}\bigg(\int_0^t |\tl\th^v_s|^{p\a}\dd s\bigg)^{1/p}\bigg)\dd t.
\end{equation}
% \begin{align}\label{DCP-3}
%   \sup_{\tl\th^v\in\tl\cA_v}\tilde{J}_{\g}(\tilde\th^v,v)&=\sup_{\tl\th^v\in\tl\cA_v}U(0,\tl\th^v), \nonumber\\
%   U(x,\tl\th^v)&=\int_0^T\big(\tilde{\theta}_t^v(v-P^\star_{\g})-C_2|\tl\th_t^v|^{\eta}(\xi_t^{\tl\th^v})^{1/p}\big)\dd t, \nonumber\\
%   \xi_t^{\tl\th^v}&=x+\int_0^t |\tl\th_s^v|^{p\a}\dd s,\quad x\geq0,\;t\in[0,T],
% \end{align}
% and similarly as before, there exists $\tl\th^{\star v}\in\tl\cA_v$ such that $\tilde{J}_{\g}(\tilde\th^\star,v)=\sup_{\tl\th\in\tl\cA_v}\tilde{J}_{\g}(\tilde\th^v,v)$.

\smallskip

Combining Subcase 1.2, Subcase 2.2, and Subcase 2.3, we have obtained that for all $\beta\geq0,\eta\geq1,\alpha\geq1$,
\begin{equation*}
  \gamma=
  \begin{cases}
    0,\quad&\text{if }\beta=0, \\
    \begin{cases}
      \beta,\quad&\text{if }\alpha=1, \\
      \displaystyle \frac{\beta\eta}{\eta+\alpha-1},\quad&\text{if }\alpha>1, \\
    \end{cases}
    \quad&\text{if }\beta>0,
  \end{cases}
\end{equation*}
which simplifies to (\ref{og}), with $\gamma<1/2$. This completes the proof.
\end{proof}

\medskip

\noindent\textbf{Proof of Theorem \ref{theorem:s}.}

\begin{proof}
First, for the pricing rule, since $\gamma<1/2$, from (\ref{MM.lp}) we know $\tilde{P}^\star_{\gamma}=\E[V]$. Also, by (\ref{dp}), we have that for any $\tl\theta^{v}\in\tl\cA$ deterministic, $v\in\mathcal{V}$,
\begin{equation}\label{dp.b}
  |X^{N^{\gamma}\tilde\theta^{v}}_{t}-1|\leq K^{v}_{t}N^{\gamma-1/2},\quad\PP\text{-a.s.},\quad t\in[0,T],\
\end{equation}
where $K^{v}_{t}\in\mathbb{L}^{0}_{\mathscr{F}_{t}}(\Omega;\mathbb{R}_{++})$ is a positive random variable depending on $v$; for example, we take
\begin{equation}\label{Kc}
  K^{v}_{t}=\frac{\exp\big|\log\mathcal{E}\big(\int^{\cdot}_{0}\tl\th^v_{s}/\sigma_{s}\dd B_{s}\big)_{t}\big|-1}{\big|\log\mathcal{E}\big(\int^{\cdot}_{0}\tl\th^v_{s}/\sigma_{s}\dd B_{s}\big)_{t}\big|}.
\end{equation}
% Then, since $\tilde{\theta}^{v}\in\tilde{\cA}_{v}$, we see that $K^{v}_{t}$ is also bounded.
Then, (\ref{MM.p2}) implies that for any $t\in[0,T]$, $\PP$-a.s.,
\begin{align}\label{MM.pb}
  |P^{\star}_{N,t}-\tilde{P}^\star_{\gamma}|=|P^{\star}_{N,t}-\E[V]| &=\frac{\big|\int_{\mathcal{V}}(v-\E[V])X^{N^{\gamma}\tilde\theta^{v}}_{t}\dd\PP[V\leq v]\big|}{\int_{\mathcal{V}}X^{N^{\gamma}\tilde\theta^{v}}_{t}\dd\PP[V\leq v]} \nonumber\\
  &=\frac{\big|\int_{\mathcal{V}}(v-\E[V])(X^{N^{\gamma}\tilde\theta^{v}}_{t}-1)\dd\PP[V\leq v]\big|}{\int_{\mathcal{V}}X^{N^{\gamma}\tilde\theta^{v}}_{t}\dd\PP[V\leq v]} \nonumber\\
  &\leq\frac{\int_{\mathcal{V}}|v-\E[V]|K^{v}_{t}N^{\gamma-1/2}\dd\PP[V\leq v]}{\int_{\mathcal{V}}X^{N^{\gamma}\tilde\theta^{v}}_{t}\dd\PP[V\leq v]} \nonumber\\
  &\leq\frac{N^{\gamma-1/2}\int_{\mathcal{V}}|v-\E[V]|K^{v}_{t}\dd\PP[V\leq v]}{\min\big\{\int_{\mathcal{V}}X^{\tilde\theta^{v}}_{t}\dd\PP[V\leq v],1\big\}},
  % &\leq\frac{\int_{\mathcal{V}}|v-\E[V]|K^{v}_{t}\dd\PP[V\leq v]}{\int_{\mathcal{V}}\min\{X^{\tilde\theta^{v}}_{t},1\}\dd\PP[V\leq v]},
\end{align}
where the second last inequality uses (\ref{dp.b}) and the last follows from the monotonicity of the sequence $\{X^{N^{\gamma}\tilde\theta^{v}}_{t}:\;N\in\mathbb{Z}_{++}\}$ of positive random variables (for every $v\in\mathcal{V}$); regarding the last bound, the Cauchy--Schwarz inequality gives that
\begin{align}\label{MM.pb2}
  \E\bigg[\frac{\int_{\mathcal{V}}|v-\E[V]|K^{v}_{t}\dd\PP[V\leq v]}{\min\{\int_{\mathcal{V}}X^{\tilde\theta^{v}}_{t}\dd\PP[V\leq v],1\big\}}\bigg]&\leq\bigg(\E\bigg[\int_{\mathcal{V}}(|v-\E[V]|K^{v}_{t})^{2}\dd\PP[V\leq v]\bigg] \nonumber\\
  &\quad\times\E\bigg[\max\bigg\{\bigg(\int_{\mathcal{V}}X^{\tilde\theta^{v}}_{t}\dd\PP[V\leq v]\bigg)^{-1},1\bigg\}^{2}\bigg]\bigg)^{1/2}.
\end{align}
Since $\tl\th$ is deterministic, the random variable $\log X^{\tl\th^{v}}_{t}=\log\mathcal{E}\big(\int^{\cdot}_{0}\tl\th^v_{s}/\sigma_{s}\dd B_{s}\big)_{t}$ (given $t\in[0,T]$) is normally distributed with mean $-1/2\int^{t}_{0}(\tl\th^v_{s}/\sigma_{s})^{2}\dd s$ and variance $\int^{t}_{0}(\tl\th^v_{s}/\sigma_{s})^{2}\dd s$. Thus, with \eqref{Kc}, by applying the Fubini--Tonelli theorem,
%as the function $(e^{x}-1)/x$, $x\geq0$, is strictly increasing,
the first expectation on the right-hand side of \eqref{MM.pb2} is finite if
\begin{align*}
  &\quad\int_{\cV}\E\big[(v-\E[V])^{2}e^{2|\log\mathcal{E}(\int^{\cdot}_{0}\tl\th^v_{s}/\sigma_{s}\dd B_{s})_{t}|}\big]\dd\PP[V\leq v] \\
  &\leq\int_{\cV}\E\big[(v-\E[V])^{2}e^{2|\int^{t}_{0}\tl\th^v_{s}/\sigma_{s}\dd B_{s}|+\int^{t}_{0}(\tl\th^v_{s}/\sigma_{s})^{2}\dd s}\big]\dd\PP[V\leq v] \\
  &\leq\int_{\cV}(v-\E[V])^{2}e^{3\int^{t}_{0}(\tl\th^v_{s}/\sigma_{s})^{2}\dd s}\dd\PP[V\leq v]<\infty.
\end{align*}
On the other hand, the second expectation is finite if
\begin{equation*}
  \E\bigg[\bigg(\int_{\mathcal{V}}X^{\tilde\theta^{v}}_{t}\dd\PP[V\leq v]\bigg)^{-2}\bigg]\leq \int_{\mathcal{V}}\E[(X^{\tilde\theta^{v}}_{t})^{-2}]\dd\PP[V\leq v]=\int_{\mathcal{V}}e^{3\int^{t}_{0}(\tl\th^v_{s}/\sigma_{s})^{2}\dd s}\dd\PP[V\leq v]<\infty.
\end{equation*}
Therefore, if \eqref{V.ic} holds, then \eqref{MM.pb2} is finite (denoted by $C>0$), and it follows from \eqref{MM.pb} that $\E[|P^{\star}_{N,t}-\tilde{P}^\star_{\gamma}|]\leq CN^{\gamma-1/2}$, which fits into \eqref{MM.ep}.

% Denote $f_1(x)=\cL(\int_0^t\si_sdB_s)$ and $f_v\sim\cL(V)$
% \textcolor{red}{
% \begin{equation}
%     \begin{aligned}
%         P^\star_t=&\E\bigg[V|\int_0^tN^\g\tl\th_s(V)\dd s+\sqrt{N}\int_0^t\si_sdB_s\bigg]\\
%         =&\int vf_{v|q}(v)dv \\
%         =&\int v \frac{f_{q|v(q)}(q)\times f_v(v)}{f_q(q)}dv\\
%         =&\int v \frac{f_1(N^{-1/2}(q-\int_0^tN^\g\th_s(v)\dd s))\times f_V(v)}{\int_vf_1(N^{-1/2}(q-\int_0^tN^\g\th_s(v)\dd s))f_V(v)\dd v}dv\\
%         =&\int v \frac{f_1(0)\times f_V(v)}{\int_vf_1(0)f_V(v)\dd v}dv=\E V
%     \end{aligned}
% \end{equation}
% }
\smallskip

Next, we consider the strategy $N^{\gamma}\tilde{\theta}^{\star}$ taking the pricing rule $\tilde{P}^{\star}_{\gamma}=\E[V]$ as given. % Based on (\textcolor{blue}{Definition of $\epsilon$-equilibria}), we aim to show that for arbitrary $\e>0$ there exists $N_\e\in\dbN$ such that for all $N>N_\e$,
%\begin{equation*}
%  |N^{-\gamma}J(\tilde{P}^\star_{\gamma};N^{\gamma}\tilde{\theta}^v,v))-\tilde{J}_{\g}(\tilde\th^v,v)|\leq\frac{\e}{2},\quad\tl\th^v\in\tl\cA_v,\; v\in\{0,1\}.
%\end{equation*}
%\textcolor{blue}{Then apply triangle equality we get $N^{\gamma}\tilde{\theta}^{\star}$ is a $\e$-equilibrium.}
Based on \eqref{IT.lo2} and \eqref{fN}, by consulting Zhang (\citeyear{Z17}) \text{Prop.} 4.1.2, we know that $N^{-\gamma}J(\tilde{P}^\star_{\gamma};N^{\gamma}\tilde{\theta}^v,v)=\breve Y_0^{N,\g,\tl\th^v}$ for any $\tl\theta^{v}\in\cA$, $v\in\cV$, where $(\breve Y^{N,\g,\tl\th^v},\breve Z^{N,\g,\tl\th^v})$ solves the BSDE
\begin{equation*}
  \breve Y_t^{N,\g,\tl\th^v}=\int_t^T\bigg[\tl\th_s^v(v-\tl P^\star_\g)-f_s^{N,\tl\th^v}-\l_s^{N^\g\tl\th^v}\breve Y_s^{N,\g,\tl\th^v}+\frac{N^{\g}\tl\th^v_s\breve Z_s^{N,\g,\tl\th^v}}{\sqrt{N}\si_s}\bigg]\dd s-\int_t^T\breve Z_s^{N,\g,\tl\th^v}\dd B_s,\quad t\in[0,T],
\end{equation*}
while for $\tilde{J}_{\g}(\tilde\th^v,v)=\tl Y_0^{\tl\th^{v}}$, a similar BSDE-type representation based on $(\tl Y^{\tl\th^v},\tl Z^{\tl\th^v})$ can be considered; following the proof of Theorem \ref{theorem:p} statement (ii), we only need to focus on Subcase 1.2, Subcase 2.2, and Subcase 3.2. In the following, $C>0$ denotes a generic constant which depends on the $\tl\th^v$ but is independent of $N$ and whose value can vary from line to line. \smallskip

\noindent\underline{Subcase 1.2}\quad With $\g=0$, we have \eqref{J-tl.1}. In this case, we have $\tilde{J}_{\g}(\tilde\th^v,v)=\tl Y_0^{\tl\th^{v}}$, with $\tl Y^{\tl\th^v}$ solving the BSDE (with no diffusion term)
\begin{equation*}
  \tl Y_t^{\tl\th^v}=\int_t^T\big(\tl\th_s^v(v-\tl P^\star_\g)-f_s^{1,\tl\th^v}-\l_s^{\tl\th^v}\tl Y_s^{\tl\th^v}\big)\dd s,\quad t\in[0,T],
\end{equation*}
% \begin{equation*}
%   \tl Y_t^{\tl\th^v}=\int_t^T\big[\tl\th_s^v(v-\tl P^\star_\g)-f_s^{1,\tl\th^v}-\l_s^{\tl\th^v}\tl Y_s^{\tl\th^v}\big]\dd s-\int_t^T\tl Z_s^{\tl\th^v}\dd B_s,\quad t\in[0,T],
% \end{equation*}
where recall (\ref{fN}).
%note that since $\g=0$, $f_s^{N,\tl\th^v}=f_s^{1,\tl\th^v}$.
Thus, by an a-priori estimate of BSDEs (e.g., Zhang \citeyear{Z17} Theorem 4.2.1), we obtain
\begin{equation}\label{BSDE.est1}
  \big|\breve Y_0^{N,\g,\tl\th^v} - \tl Y_0^{\tl\th^{v}}\big|^2\leq\dbE\bigg[\int_0^T\bigg|\frac{\tl\th^v_t\breve Z_t^{N,\g,\tl\th^v}}{\sqrt{N}\si_t}\bigg|\dd t\bigg]^2\leq\bigg(\int_0^T|\tl\th^v_t |^2 \dd r\bigg)\dbE\bigg[\int_0^T\bigg|\frac{\breve Z_t^{N,\g,\tl\th^v}}{\sqrt{N}\si_t}\bigg|^2\dd t\bigg]\leq\frac{C}{N},
\end{equation}
where the last inequality is due to Novikov's condition in (\ref{wf.adm.t}), and by the a-priori estimate of BSDEs and Novikov's condition,
\begin{equation}\label{apriori}
  \dbE\bigg[\sup_{t\in[0,T]}\breve Y_t^{N,\g,\tl\th^v}+\int_0^T |\breve Z_t^{N,\g,\tl\th^v}|^2\dd t\bigg]\leq\dbE\bigg[\int_0^T\big|\tilde{\theta}_t^v(v-\tl P^\star_\g)-f_t^{1,\tl\th^v}\big|^2\dd t\bigg]\leq C.
\end{equation}
Based on \eqref{IT.eop} ($\epsilon$-optimality) in Definition \ref{eeq.def}, we have $\epsilon_N\leq CN^{-1/2}$ in this case. \smallskip

\noindent\underline{Subcase 2.2}\quad With $\g=\b$ and $\a=1$, we have \eqref{J-tl.2}. In this case, similarly, $\tilde{J}_{\g}(\tilde\th^v,v)=\tl Y_0^{\tl\th^{v}}$, and $\tl Y^{\tl\th^v}$ solves the BSDE
\begin{equation*}
  \tl Y_t^{\tl\th^v}=\int_t^T\big(\tl\th_s^v(v-\tl P^\star_\g)- f_s^{\tl\th^v}-\l_s^{\tl\th^v}\tl Y_s^{\tl\th^v}\big)\dd s,\quad t\in[0,T],
\end{equation*}
where $ f^{\tl\th^v}_s:=C_{1}\lambda(s,\tl\th^v_s)\big(b\big(\int_0^s|\tl\th^v_r|^{p}\dd r\big)^{1/p}+c(v-\tl P^\star_{\g})\int_0^s\tl\theta^{v}_{r}\dd r\big)$ based on \eqref{f1}. Then, a similar a-priori estimate implies
\begin{equation*}
  \big|\breve Y_0^{N,\g,\tl\th^v}-\tl Y_0^{\tl\th^{v}}\big|^2 \leq\dbE\bigg[\int_0^T\bigg(\big|f_t^{N,\tl\th^v}-f_t^{\tl\th^v}\big|^2+\bigg|\frac{N^{\gamma}\tl\th^v_t\breve Z_t^{N,\g,\tl \th^v}}{\sqrt{N}\si_t}\bigg|^2 \bigg)\dd t\bigg].
\end{equation*}
By \eqref{fNorder1}, it is clear that $|f_t^{N,\tl\th^v}-f_t^{\tl\th^v}|\leq CN^{-\g}$, while similar to \eqref{apriori},
\begin{equation}\label{EZ.est}
  \E\bigg[\int_0^T\bigg|\frac{N^{\gamma}\tl\th^v_t\breve Z_t^{N,\g,\tl \th^v}}{\sqrt{N}\si_t}\bigg|^2 \dd t\bigg]\leq C N^{2\g-1}.
\end{equation}
Since $\a=1$, by \eqref{og} we have $(\g-\b)\eta=\g(1-\a)=0$. Thus, combining the last two bounds yields
\begin{equation}\label{BSDE.est2}
  \big|\breve Y_0^{N,\g,\tl\th^v}-\tl Y_0^{\tl\th^{v}}\big|^2\leq C (N^{-2\g}+N^{2\g-1})=CN^{\max\{-2\gamma,2\gamma-1\}},
\end{equation}
showing that $\epsilon\leq CN^{\max\{\gamma-1/2,-\gamma\}}$ in Definition \ref{eeq.def}. \smallskip

\noindent\underline{Subcase 2.3}\quad Now, $\g<\b$, and $\a>1$, and we have \eqref{J-tl.3}. As before, $\tilde{J}_{\g}(\tilde\th^v,v)=\tl Y_0^{\tl\th^{v}}$, where $(\tl Y^{\tl\th^v},\tl Z^{\tl\th^v})$ solves the following BSDE:
\begin{equation*}
  \tl Y_t^{\tl\th^v}=\int_t^T\big(\tl\th_s^v(v-\tl P^\star_\g)-\tl f_s^{\tl\th^v}\big)\dd s,\quad t\in[0,T],
\end{equation*}
where $\tl f^{\tl\th^v}_s:=C_{2}|\tl\th_s^v|^{\eta}\big(\int_0^s |\tl\th^v_r|^{p\a}\dd r\big)^{1/p}$ by \eqref{f2}. Thus, we have
\begin{equation*}
  \big|\breve Y_0^{N,\g,\tl\th^v}-\tl Y_0^{\tl\th^{v}}\big|^2 \leq\dbE\bigg[\int_0^T\bigg(\big|f_t^{N,\tl\th^v}- f^{\tl\th^v}_t\big|^2+\big|\lambda(t,N^{\g-\b}\tl\theta^{v}_{t})\breve Y_t^{N,\g,\tl\th^v}\big|^2+\bigg|\frac{N^{\gamma}\tl\th^v_t\breve Z_t^{N,\g,\tl \th^v}}{\sqrt{N}\si_t}\bigg|^2 \bigg)\dd t\bigg].
\end{equation*}
Using \eqref{fNorder2} and (\ref{og}), we see that
\begin{equation*}
  \big|f_s^{N,\tl\th^v}-f^{\tl\th^v}_t\big|\leq C(N^{(\g-\b)\eta-\g}+N^{(\g-\b)\eta})=C(N^{-\g(\a-1)}+N^{-\a\g})=CN^{-\g(\a-1)}.
\end{equation*}
Also, since $\g<\b$, $\lambda(s,N^{\g-\b}\tl\theta^{v}_{s})=O(N^{(\g-\b)\eta})=O(N^{-\g(\a-1)})$ as $N\to\infty$. Together with the estimate in \eqref{EZ.est}, these imply
\begin{equation}\label{BSDE.est3}
  \big|\breve Y_0^{N,\g,\tl\th^v}-\tl Y_0^{\tl\th^{v}}\big|^2\leq C N^{\max\{2\gamma-1,-2\gamma(\a-1)\}},
\end{equation}
and so $\epsilon_N\leq C N^{\max\{\gamma-1/2,-\gamma(\a-1)\}}$ in this case.

Combining (\ref{BSDE.est1}), (\ref{BSDE.est2}), and (\ref{BSDE.est3}), along with (\ref{MM.pb}), we arrive at (\ref{eps.rate}) as desired.
\end{proof}

\medskip

\noindent\textbf{Proof of Theorem \ref{theorem:sup}.}

\begin{proof}
First, let us note that by \eqref{MM.p2} in Proposition \ref{pro:2} -- under the alternative formulation -- changing the form of $\varpi_{0}$ has no effect on the pricing rule $P$ for any given strategy $\theta$. This means that $\epsilon=0$ in (\ref{MM.ep}) in Definition \ref{eeq.def} once an equilibrium strategy $\th^\star_N$ is fixed, and so we only need to consider the $\epsilon$-optimality in (\ref{IT.eop}).
% Notice that when apply $(\th^{\star}, P^{\star})$ as the equilibrium of the model with sup norm to the model with the $p$-norm or as the $\e_N$-equilibrium of the model with $p$-norm to the model with the sup norm, since the $\th^{\star}$ didn't change, optimal pricing rule remains to be $P^\star$ and thus the \eqref{MM.ep} in Definition \ref{eeq.def} is automatically satisfied, thus suffice to consider the $\e$-optimality of IT \eqref{IT.eop}.

% In proving both statements, we shall establish an estimate.
Based on Proposition \ref{pro:3}, for any $v\in\cV$, $\th^v\in\cA$, and $P=P^{\th^v}$ (recalling Definition \ref{as.def}), let us write $(Y^{\th^v},Z^{\th^v})$ from \eqref{BSDE1} as $(Y^{p},Z^{p})$ and $(Y^\infty,Z^\infty)$, respectively, with $p\geq1$, i.e., $Y^p_0$ and $Y^{\infty}_0$ are the IT's objective function \eqref{IT.o4} under \eqref{crpe.s} and \eqref{crpe.sl}, respectively. Then, $(Y^{p},Z^{p})$ and $(Y^\infty,Z^\infty)$ respectively solve the following BSDEs:\footnote{In the remainder of the proof, with $N\geq1$ fixed, we suppress the subscript $N$ in all notation for conciseness.}
\begin{align}\label{BSDEp}
  Y_t^{p}&=\int_t^T\bigg(\th_s^v(v-P_{s})-\l_s^{\th^v}W\bigg(\bigg(\int^{s}_{0}(\varpi_{0}(\theta^{v}_{r}))^{p}\dd r\bigg)^{1/p},\max\bigg\{0,c\int^{s}_{0}\theta^{v}_{r}(v-P_{r})\dd r\bigg\}\bigg)  \nonumber\\
  &\quad-\l_s^{\th^v}Y_s^{p}+\frac{\th^v_sZ_s^{p}}{\sqrt{N}\si_s}\bigg)\dd s-\int_t^TZ_s^{p}\dd B_s,\quad t\in[0,T],
\end{align}
and
\begin{align}\label{BSDEinfty}
  Y_t^{\infty}&=\int_t^T\bigg(\th_s^v(v-P_s)-\l_s^{\th^v}W\bigg(\sup_{r\in[0,s]}\varpi_{0}(\theta_{r}), \max\bigg\{0,c\int_0^s\th_r^v(v-P_r) \dd r\bigg\}\bigg)  \nonumber\\
  &\quad-\l_s^{\th^v}Y_s^{\infty}+\frac{\th^v_sZ_s^{\infty}}{\sqrt{N}\si_s}\bigg)\dd s-\int_t^TZ_s^{\infty}\dd B_s,\quad t\in[0,T].
\end{align}
Similar to the proof of Theorem \ref{theorem:p} statement (i), by the a-priori estimate of BSDEs and the fact that function $W$ is Lipschitz-continuous in the first argument (Assumption \ref{as:1}), we have that for some positive constant $C$,
\begin{equation}\label{eps.p}
  |Y_0^{\infty}-Y_0^{p}|^2\leq C\E\bigg[\int_0^T\bigg|\sup_{r\in[0,s]}\varpi_{0}(\theta_{r})-\bigg(\int_0^s(\varpi_{0}(\theta_{r}))^p\dd r\bigg)^{1/p}\bigg|^2\dd s\bigg].
\end{equation}
% Suppose we can find a set of $S$ such that $|\int_0^T|f(s)|^pds|^\frac{1}{p}$ converge to $\sup_{s\in[0,T]}|f(s)|$ uniformly for all $f\in S$, and $\varpi_{0}(\th(\o,\cdot))\in S$ $\O-$almost surely then we are done. We claim that it is suffice to let $S$ be the set of Lipchitz continuous with a common Lipchitz constant $L$ and uniformly bounded functions.

%It remains to argue that the set $\cA'$ is compact. Since for all $\th^{v}\in\cA$, $\|\varpi_0(\th^v)\|_p\nearrow\|\varpi_0(\th^v)\|_\infty$ and the norms $\|\cdot\|_{p}$ and $\|\cdot\|_{\infty}$ are all continuous functionals on $\cA'$, Dini's Theorem implies that \textcolor{blue}{$\|\cdot\|_p$ converge to $\|\cdot\|_\infty$ uniformly, but since $\cA'$ ensures Lipschitz continuity in $t$ with a common Lipschitz constant $L$ and uniformly bounded, then by Arzela--Ascoli, $\cA'$ is compact. By assumption, all $\th^v\in\cA_{N,v}$ has a common Lipscthiz constant and bounded, $\varpi_{0}$ is also Lipchitz, thus we are done. \emph{Cite the previous two BSDEs in the arguments.}}

Assume that the set $\cA'$ is compact. Since for any $\th^{v}\in\cA'$,
\begin{equation}\label{pconv}
  \int_0^T\bigg(\int^t_0(\varpi_0(\th_t^v))^{p}\dd s\bigg)^{1/p}\dd t\nearrow\int_0^T\sup_{s\in[0,t]}\varpi_0(\th_s^v)\dd t,\quad \PP\text{-a.s.},
\end{equation}
and by Assumption \ref{as:2}, the two (outer) integrals %$\th^v(\o)\mapsto\int_0^T\|\varpi_0(\th_t^v(\o))\|_{p}\dd t$ and $\th^v(\o)\mapsto\int_0^T\|\varpi_0(\th_t^v(\o))\|_{\infty}\dd t$
on both sides of the last convergence are continuous functionals (of $\theta^{v}$) on $\cA'$, an application of Dini's Theorem implies that the convergence in \eqref{pconv} is also uniform, i.e., for any $\e>0$, there exists $p_0\geq1$ such that for all $p\geq p_0$ and any $\th^{v}\in\cA$,
\begin{equation*}
  \int_0^T\sup_{s\in[0,t]}\varpi_0(\th_s^v)\dd t-\int_0^T\bigg(\int^t_0(\varpi_0(\th_t^v))^{p}\dd s\bigg)^{1/p}\dd t\leq\e,\quad \dbP\text{-a.s.}
\end{equation*}
On the other hand, by the imposed condition of uniform boundedness and Lipschitz continuity (in time) with a common Lipschitz constant $L$ in $\cA'$, the compactness of $\cA'$ results from the Arzela--Ascoli theorem.

Therefore, (\ref{eps.p}) implies that $|Y_0^{\infty}-Y_0^{p}|\leq\e_p$, where $\e_p$ does not depend on the choice of $\th^{v}$ and $P$, and $\lim_{p\to\infty}\e_p=0$, and we can set $\th$ to be the (presumably existent) equilibrium strategy $\th^{\star}_N$ (either under \eqref{crpe.s} or under \eqref{crpe.sl}) to establish both statements in the theorem. %Thus by triangle inequality, (i) and (ii) follows trivially.
Indeed, for statement (i), by including $\th^{\star}_N$ under \eqref{crpe.sl} in the superscripts, %denote $(Y_t^{p,\th^v},Z_t^{p,\th^v})$ and $(Y_t^{\infty,\th^v},Z_t^{\infty,\th^v})$ as the solution of BSDE \eqref{BSDEp} and \eqref{BSDEinfty} with strategy $\th^v\in\cA'$,
we have from \eqref{eps.p} that
\begin{equation*}
  \big|Y_0^{\infty,\th_N^{\star v}}-Y_0^{p,\th_N^{\star v}}\big|+Y_0^{p,\th_N^{\star v}}\geq Y_0^{\infty,\th_N^{\star v}}\geq Y_0^{\infty}\geq Y_0^{p}-|Y_0^{p}- Y_0^{\infty}|,
\end{equation*}
where $Y^{p}_0\equiv Y^{p,\th^v}_0$ and $Y^{\infty}_0\equiv Y^{\infty,\th^v}_0$ and which holds for any $\th^v\in\cA'$, which shows that $Y_0^{p,\th_N^\star}\geq Y_0^{p,\th^v}-2\e_p$ for any $\th^v\in\cA'$. The same argument goes for statement (ii) with $\th^{\star}_N$ taken under \eqref{crpe.s}.
\end{proof}

\medskip

\section{Proofs in Section \ref{S:4}}\label{C}

\renewcommand{\theequation}{C.\arabic{equation}}

\noindent\textbf{Proof of Proposition \ref{pro:4}.}
\begin{proof}
With $\eta=1$ and $\alpha>1$, by (\ref{og}) we immediately have $\gamma=\beta/\alpha$, and we consider the problem in (\ref{DCP-3}) from Subcase 2.3 in the proof of Theorem \ref{theorem:p} statement (ii), where, for $v\neq\E[V]$ given, the objective function is essentially
\begin{equation}\label{J-tl.3a}
  \tilde{J}_{\g}(\tilde\th^v,v)=\int^{T}_{0}\bigg(\tl\theta^{v}_{t}(v-\E[V])-C_{2}|\tl\th_t^v|^{\eta}\bigg(\int_0^t |\tl\th^v_s|^{p\a}\dd s\bigg)^{1/p}\bigg)\dd t.
\end{equation}
Since (\ref{DCP-3}) is deterministic optimization, we employ Pontryagin's maximum principle (see, e.g., Zhou \citeyear{Z90}) to solve it. The idea is to obtain a first-order condition that must hold along an optimal trajectory for the state variable $x=\xi(t)=\int^{t}_{0}|\tl\theta^{v}_{s}|^{p\alpha}\dd s\geq0$, for $t\in[0,T]$.

Let us denote $q\equiv q(p):=p\alpha>1$. In the following, based on \eqref{wf.adm.t}, we restrict to the case $v-\E[V]>0$, with $\tl\theta^{v}_{t}\geq0$ for $t\in[0,T]$, as if $v-\E[V]<0$, the same arguments apply to $\tl\theta^{v}_{t}\leq0$. Further, we assume without loss of generality that $\tl\theta^{v}_{t}>0$ for \text{a.e.} $t\in[0,T]$, so that $x$ is strictly increasing on $[0,T]$.\footnote{Indeed, according to the below arguments, any interval on which
$\tl\theta^{v}\equiv0$ strictly reduces the objective function and so any such $\tl\theta^{v}$ cannot be optimal.} Then, we can re-parameterize the original objective function in (\ref{J-tl.3a}) with the state variable $x$, i.e.,
\begin{equation}\label{o_rep}
  \tl J_{\gamma}(\vartheta,v)=\int^{x(T)}_{0}(v-\E[V]-C_{2}y^{1/p})\vartheta(y)^{1-q}\dd y.
\end{equation}
Since $x(0)=\xi(0)=0$, the state equation $\dd x/\dd t=\vartheta^{q}$ gives the relation
\begin{equation}\label{o_rep_t}
  t=\int^{x(t)}_{0}(\vartheta(y))^{-q}\dd y,\quad t\leq T.
\end{equation}
Also, with the increase of $x$ on $[0,T]$, the transversality condition ensures that under optimality, $x(T)=\bar{x}$ with $(v-\E[V]-C_{2}\bar x^{1/p})=0$, or
\begin{equation}\label{x_bar}
  \bar x:=\bigg(\frac{v-\E[V]}{C_{2}}\bigg)^{p},
\end{equation}
namely the maximal attainable state value. Thus, we also have the terminal condition coupled to (\ref{o_rep}),
\begin{equation}\label{o_rep_T}
  T=\int^{\bar{x}}_{0}(\vartheta(y))^{-q}\dd y.
\end{equation}

By forming the Lagrangian for \eqref{o_rep} with \eqref{o_rep_T}, namely
\begin{equation*}
  \mathcal{L}(\vartheta,x):=\int^{\bar x}_{0}\big((v-\E[V]-C_{2}y^{1/p})(\vartheta(y))^{1-q}-\mu(\vartheta(y))^{-q}\big)\dd y,
\end{equation*}
with multiplier $\mu\geq0$, we observe that the integrand is local in $y$, and so perturbing $\vartheta$ directly yields the pointwise first-order condition
\begin{equation*}
  (v-\E[V]-C_{2}x^{1/p})(1-q)(\vartheta(x))^{-q}+\mu q(\vartheta(x))^{-q-1}=0,
\end{equation*}
which, upon simplification, gives the optimal control in terms of $x$,
\begin{equation}\label{varth_x}
  \vartheta(x)=K(v-\E[V]-C_{2}x^{1/p})^{-1},
\end{equation}
where $K>0$ is a constant involving $\mu$ and which is clearly strictly positive for $x<\bar x$.
%The maximum state $\bar x$ is determined by the condition $v-\E[V]-C_{2}\bar x^{1/p}=0$, or $\bar x=((v-\E[V])/C_{2})^{p}$.

Now, if we define
\begin{equation}\label{gv:int}
  g_{v}(x):=\int^{\bar x}_{x}(v-\E[V]-C_{2}y^{1/p})^{q}\dd y,\quad x\in[0,\bar x],
\end{equation}
which is strictly decreasing in $x$, with $g_{v}(\bar x)=0$, then by plugging \eqref{varth_x} into the terminal condition \eqref{o_rep_T} we have
\begin{equation}\label{K}
  K=\bigg(\frac{g_{v}(0)}{T}\bigg)^{1/q}.
\end{equation}
Using (\ref{o_rep_T}) and rearranging terms gives
\begin{equation}\label{x_t}
  x(t)=g^{-1}_{v}\bigg(\frac{T-t}{T}g_{v}(0)\bigg),
\end{equation}
noting that the inverse function $g^{-1}_{v}$ is well-defined and positive on $[0,g_{v}(0))$.
%We can verify that under (\ref{P1:th}), the value function
%\begin{equation*}
%  \cJ(t,x)=\int^{T}_{t}\big((v-\E[V])\vartheta(x(s))-C_{2}(\vartheta(x(s)))^{q}(x(s))^{1/p}\big)\dd t
%\end{equation*}
%is indeed a classical solution of the PDE (?)
Putting \eqref{K} and \eqref{x_t} into \eqref{varth_x} and simplifying and setting $\tl\theta^{\star v}_{t}=\vartheta$, also recalling $q=p\alpha$, we arrive at the expression in (\ref{P1:th}).

The integral in (\ref{gv:int}) fits into the integral representation for the incomplete Beta function. In particular, using the substitution $C_{2}y^{1/p}/(v-\E[V])\mapsto u$,
\begin{align*}
  \int^{\bar x}_{x}(v-\E[V]-C_{2}y^{1/p})^{q}\dd y&=\bigg(\int^{\bar x}_{0}-\int^{x}_{0}\bigg)(v-\E[V])^{q}\bigg(1-\frac{C_{2}y^{1/p}}{v-\E[V]}\bigg)^{q}\dd y \\
  &=\frac{p(v-\E[V])^{q+p}}{C^{p}_{2}}\bigg(\int^{C_{2}\bar{x}^{1/p}/(v-\E[V])}_{0}-\int^{C_{2}x^{1/p}/(v-\E[V])}_{0}\bigg)u^{p-1}(1-u)^{q}\dd u \\
  &=\frac{p(v-\E[V])^{q+p}}{C^{p}_{2}}\big(\mathrm{B}_{C_{2}\bar{x}^{1/p}/(v-\E[V])}(p,q+1)-\mathrm{B}_{C_{2}x^{1/p}/(v-\E[V])}(p,q+1)\big),
\end{align*}
which with (\ref{x_bar}) and $q=p\alpha$ yields the explicit formula in (\ref{P1:gv}).

Moreover, based on (\ref{varth_x}), it is clear that $\tl\theta^{\star v}_{t}=\vartheta$ is continuous and increasing in $t\in(0,T)$. With $\kappa:=C^{p}_{2}(v-\E[V])^{1-p}/p$, note that $(v-\E[V]-C_{2}x^{1/p})=\kappa(\bar x-x)$ as $x\nearrow\bar x$, and thus,
\begin{equation*}
  g_{v}(x)=\int^{\bar x}_{x}(v-\E[V]-C_{2}y^{1/p})^{q}\dd y\sim\kappa^{q}\int^{\bar x}_{x}(\bar x-y)^{q}\dd y=\frac{\kappa^{q}}{q+1}(\bar x-x)^{q+1},\quad\text{as }x\nearrow\bar{x},
\end{equation*}
which with (\ref{K}) and (\ref{x_t}) implies that, as $t\nearrow T$,
\begin{equation*}
  \bar{x}-x(t)=\bigg(\frac{K^{q}(q+1)}{\kappa^{q}}\bigg)^{1/(q+1)}(T-t)^{1/(q+1)},
\end{equation*}
and therefore,
\begin{equation*}
  \tl\theta^{\star v}_{t}=K(v-\E[V]-C_{2}(x(t))^{1/p})^{-1}\sim\frac{K}{\kappa}\bigg(\frac{K^{q}(q+1)}{\kappa^{q}}\bigg)^{-1/(q+1)}(T-t)^{-1/(q+1)}.
\end{equation*}
Since $q>1$, this verifies that $\tl\theta^{\star v}\in\tl\cA$ indeed (square-integrable), hence being a valid equilibrium strategy according to Definition \ref{leq.def}, completing the proof.
\end{proof}

\medskip

\noindent\textbf{Proof of Proposition \ref{pro:5}.}
\begin{proof}
Again, we restrict to the case $v-\E[V]>0$ and, by the maximum principle, re-parameterize the original objective function in (\ref{J-tl.3a}) with state $x$; that is, with $\vartheta\equiv\vartheta(x)$,
\begin{equation}\label{o_rep2}
  \tl J_{\gamma}(\vartheta,v)=\int^{\bar x}_{0}\big((v-\E[V])(\vartheta(y))^{1-\eta}-C_{2}y^{1/p}\big)\dd y\quad\text{subject to}~~T=\int^{\bar x}_{0}(\vartheta(y))^{-\eta}\dd y,
\end{equation}
where $\bar{x}>0$ corresponds to the maximal attainable value of $x$, which equals $x(T)$ as $x$ is increasing. We form the Lagrangian
\begin{equation*}
  \mathcal{L}(\vartheta,x):=\int^{\bar x}_{0}\big((v-\E[V])(\vartheta(y))^{1-\eta}-C_{2}y^{1/p}-\mu(\vartheta(y))^{-\eta}\big)\dd y,
\end{equation*}
with multiplier $\mu\geq0$, which shows that the $\vartheta$-terms are separate from $y$ in the integrand, and thus perturbation yields constant optimal $\vartheta$, namely
\begin{equation}\label{mu}
  \mu=(v-\E[V])\bigg(1-\frac{1}{\eta}\bigg)\vartheta.
\end{equation}
Under the transversality condition, $\bar{x}$ is such that $(v-\E[V])\vartheta^{1-\eta}-C_{2}\bar{x}^{1/p}-\mu\vartheta^{-\eta}=0$, which along with \eqref{mu} further gives
\begin{equation}\label{varth_xb2a}
  \frac{v-\E[V]}{\eta}\vartheta^{1-\eta}=C_{2}\bar x^{1/p}.
\end{equation}
On the other hand, from the terminal condition in (\ref{o_rep2}) we have
\begin{equation}\label{varth_xb2b}
  \vartheta=\bigg(\frac{\bar x}{T}\bigg)^{1/\eta}.
\end{equation}
Solving (\ref{varth_xb2a}) and (\ref{varth_xb2b}) we obtain
\begin{equation*}
  \tl\theta^{\star v}\equiv\vartheta=\bigg(\frac{v-\E[V]}{\eta C_{2}}\bigg)^{p/(p(\eta-1)+\eta)}T^{p(\eta-1)/(\eta(p(\eta-1)+\eta))},
\end{equation*}
the same as (\ref{P2:th}) with $\eta=p\alpha$, recalling $v-\E[V]>0$ here. %\footnote{One can also verify, optionally, that the value function $\cJ$ under the constant optimal control $\vartheta$ is indeed the (classical) solution to the foregoing PDE.}
\end{proof}

\medskip

\noindent\textbf{Proof of Proposition \ref{pro:6}.}
\begin{proof}
Since $\alpha=1$ and $\eta=1$, we have $\gamma=\beta$ from (\ref{og}). Here, we consider the problem in (\ref{DCP-2}) from Subcase 2.2 in the proof of Theorem \ref{theorem:p} statement (ii). With $v\neq\E[V]$ given, the objective function is now
\begin{equation}\label{J-tl.2a}
  \tilde{J}_{\g}(\tilde\th^v,v)=\int^{T}_{0}e^{-\kappa \int^{t}_{0}|\tl\theta^v_{s}|\dd s}\bigg(\tl\theta^{v}_{t}(v-\E[V]) -\kappa C_{1}|\tl\th^v_t|\bigg(b\bigg(\int_0^t|\tl\th^v_s|^{p}\dd s\bigg)^{1/p}+c(v-\E[V])\int_0^t\tl\theta^{v}_{s}\dd s\bigg)\bigg)\dd t.
\end{equation}
Focusing again on $v-\E[V]>0$, with $\tl\th^{v}_{t}\geq0$ for all $t\in[0,T]$, we apply the maximum principle, considering two separate cases. \smallskip

\noindent\underline{Case 1}\quad $p=1$ or $b=0$. In this case, \eqref{J-tl.2a} can be significantly simplified with the state variable $x=\int^{t}_{0}\tilde{\theta}^{v}_{s}\dd s$. Re-parameterizing \eqref{J-tl.2a} with $x$, we have
\begin{equation}\label{o_rep3}
  J(\vartheta,v)=\int^{x_{1}(T)}_{0}e^{-\kappa y}\big(v-\E[V]-\kappa C_{1}(b+c(v-\E[V]))y\big)\dd y,
\end{equation}
which happens to be independent of $\vartheta$. This means that (\ref{o_rep3}) is valid even if $\vartheta=0$ on some intervals within $[0,x_{1}(T)]$. Thus, the original problem becomes to maximize $J(\vartheta,v)$ with respect to the only controllable variable $x_{1}(T)\equiv\bar x$. However, since the integrand in (\ref{o_rep3}) clearly has to be positive under optimality, $\bar x$ is such that $v-\E[V]-\kappa C_{1}(b+c(v-\E[V]))\bar x=0$, or
\begin{equation}\label{x_bar2}
  \bar{x}=\frac{v-\E[V]}{\kappa C_{1}(b+c(v-\E[V]))},
\end{equation}
and $\tl\theta^{\star v}\geq0$ only needs to satisfy
\begin{equation*}
  \bar{x}=\int^{T}_{0}\tl\theta^{\star v}_{t}\dd t,
\end{equation*}
verifying (\ref{P3:th}).

\smallskip

\noindent\underline{Case 2}\quad $p>1$ and $b>0$.\quad Assuming that $\tl\theta^{v}_{t}>0$ for \text{a.e.} $t\in[0,T]$ and re-parameterizing the objective function in \eqref{J-tl.2a} with the state variable $x=\int^{t}_{0}\tilde{\theta}^{v}_{s}\dd s$, we now have
\begin{equation}\label{o_rep3'}
  J(\vartheta,v)=\int^{\bar x}_{0}e^{-\kappa y}\big(v-\E[V]-\kappa C_{1}(b(h(y))^{1/p}+c(v-\E[V])y)\big)\dd y,
\end{equation}
where $\bar x=x(T)>0$ is, again, the maximal attainable value of $x$, and
\begin{equation}\label{h}
  h(x):=\int^{x}_{0}(\vartheta(y))^{p-1}\dd y.
\end{equation}
Also, from the state equation $\dd x/\dd t=\vartheta$ we have that
\begin{equation}\label{o_rep_t'}
  t=\int^{x}_{0}(\vartheta(y))^{-1}\dd y,
\end{equation}
which yields the terminal condition coupled to (\ref{o_rep3'}),
\begin{equation}\label{o_rep_T'}
  T=\int^{\bar x}_{0}(\vartheta(y))^{-1}\dd y.
\end{equation}

Forming the Lagrangian for (\ref{o_rep3'}) and (\ref{o_rep_T'}), we have
\begin{equation}\label{Lgn}
  \mathcal{L}(\vartheta,x)=\int^{x}_{0}\big(e^{-\kappa y}(v-\E[V]-\kappa C_{1}(b(h(y))^{1/p}+c(v-\E[V])y))-\mu(\vartheta(y))^{-1}\big)\dd y,
\end{equation}
with multiplier $\mu\geq0$. Since $h$ involves integration of the control $\vartheta$, by perturbation with respect to $\vartheta$ we have that for any fixed $\Delta\vartheta$ supported on $[0,\bar x]$,
\begin{align*}
  D(h(x))^{1/p}&=\frac{1}{p}(h(x))^{1/p-1}Dh(x) \\ &=\frac{1}{p}(h(x))^{1/p-1}\bigg[\frac{\dd}{\dd\e}\int^{x}_{0}(\vartheta(y)+\e\Delta\vartheta(y))^{p-1}\dd y\bigg]\bigg|_{\e=0} \\ &=\frac{p-1}{p}(h(x))^{1/p-1}\int^{x}_{0} (\vartheta(y))^{p-2}\Delta\vartheta(y)\dd y,
\end{align*}
and thus, for (\ref{Lgn}),
\begin{align*}
  D\mathcal{L}(\vartheta,x)&=-\kappa bC_{1}\int^{\bar x}_{0}e^{-\kappa y}D(h(y))^{1/p}\dd y+\mu\int^{\bar x}_{0}(\vartheta(y))^{-2}\Delta\vartheta(y)\dd y \\
  &=-\kappa bC_{1}\int^{\bar x}_{0}e^{-\kappa y}\bigg(\frac{p-1}{p}(h(y))^{1/p-1}\int^{y}_{0}(\vartheta(z))^{p-2}\Delta\vartheta(z)\dd z\bigg)\dd y +\mu\int^{\bar x}_{0}(\vartheta(y))^{-2}\Delta\vartheta(y)\dd y \\
  &=\int^{\bar x}_{0}\Delta\vartheta(z)\bigg(\mu(\vartheta(z))^{-2}-\kappa bC_{1}\frac{p-1}{p}(\vartheta(z))^{p-2}\int^{\bar x}_{z}e^{-\kappa y}(h(y))^{1/p-1}\dd y\bigg)\dd z.
\end{align*}
Since $\Delta\vartheta$ is arbitrary, the last integrand must vanish, and this gives us the following first-order condition:
\begin{equation*}
  \mu=\kappa bC_{1}\frac{p-1}{p}(\vartheta(x))^{p}\int^{\bar x}_{x}e^{-\kappa y}(h(y))^{1/p-1}\dd y.
\end{equation*}
Noting that $\vartheta=(h')^{1/(p-1)}$ from (\ref{h}), by writing $\varsigma=p\mu/(\kappa bC_{1}(p-1))$, the above becomes
\begin{equation}\label{o_foc}
  (h'(x))^{p/(p-1)}=\varsigma\bigg(\int^{\bar x}_{x}e^{-\kappa y}(h(y))^{1/p-1}\dd y\bigg)^{-1}.
\end{equation}
By differentiating the left-hand side of \eqref{o_foc}, we have
\begin{equation}\label{o_foc_L}
  \frac{\dd}{\dd x}((h'(x))^{p/(p-1)})=\frac{p}{p-1}(h'(x))^{p/(p-1)-1}h''(x),
\end{equation}
while differentiation of the right-hand side gives
\begin{equation}\label{o_foc_R}
  \frac{\dd}{\dd x}\bigg(\varsigma\bigg(\int^{\bar x}_{x}e^{-\kappa y}(h(y))^{1/p-1}\dd y\bigg)^{-1}\bigg)=\frac{\varsigma e^{-\kappa x}(h(x))^{1/p-1}}{\big(\int^{\bar x}_{x}e^{-\kappa y}(h(y))^{1/p-1}\dd y\big)^{2}}=\varsigma^{-1}e^{-\kappa x}(h(x))^{1/p-1}(h'(x))^{2p/(p-1)},
\end{equation}
where the second equality uses that $\big(\int^{\bar x}_{x}e^{-\kappa y}(h(y))^{1/p-1}\dd y\big)^{-2}=\varsigma^{-2}(h'(x))^{2p/(p-1)}$ from (\ref{o_foc}). Substituting (\ref{o_foc_L}) and (\ref{o_foc_R}) into \eqref{o_foc} simplifying, we obtain the following second-order ODE for the function $h$:
\begin{equation}\label{h_ODE}
  h''(x)=\frac{p-1}{p\varsigma}e^{-\kappa x}(h(x))^{1/p-1}(h'(x))^{(2p-1)/(p-1)},
\end{equation}
which is the same as the first equation in (\ref{P3:DS}) and is subject to the immediate boundary condition $h(0)=0$. To determine a second boundary condition, we observe that in order for (\ref{h_ODE}) to admit a real-valued nonzero solution, it must be that $h'(0)=\chi$ for some constant $\chi>0$, which also ensures uniqueness by continuity. In this case, (\ref{h_ODE}) then admits a positive solution $h(x)$ for $x\in(0,\bar x)$, where $\bar x$ has to satisfy the transversality condition, with $\bar x=x(T)$ due to the increase of $x$ under optimality; more precisely, based on the integrand in (\ref{Lgn}), using that $\vartheta=(h')^{1/(p-1)}$, we have
\begin{equation}\label{Tr}
  e^{-\kappa \bar x}(v-\E[V]-\kappa C_{1}(b(h(\bar x))^{1/p}+c(v-\E[V])\bar x))-\mu(h'(\bar x))^{-1/(p-1)}=0.
\end{equation}
On the other hand, for the time constraint, if we define
\begin{equation}\label{H}
  H(x):=\int^{x}_{0}(h'(y))^{-1/(p-1)}\dd y,
\end{equation}
which corresponds to \eqref{o_rep_t'} parameterized with $x$ and satisfies $H'=(h')^{-1/(p-1)}$, then
\begin{equation}\label{TC}
  H(\bar x)=T.
\end{equation}
Together, \eqref{h_ODE}, \eqref{Tr}, and \eqref{TC} uniquely determine the three unknowns $\bar x>0$, $\mu\geq0$, and $\chi>0$.

Given the solved values $\bar x$, $\mu$, and $\chi$, as well as the corresponding unique solution $h$ to the ODE \eqref{h_ODE}, we then have (recalling (\ref{h})) $\vartheta(x)=(h'(x))^{1/(p-1)}$ for $x\in(0,\bar x)$, with $\vartheta(x)=\chi^{1/(p-1)}$. Note further that $H(x)=t$ by (\ref{H}), and $H$ is strictly increasing because $h'(x)>0$ for $x<\bar x$, and therefore by inverting it we can write
\begin{equation}\label{th_varth_x}
  \tl\theta^{\star v}_{t}=\vartheta(x(t))=(h'(H^{-1}(t)))^{1/(p-1)},\quad t\in[0,T),
\end{equation}
which is clearly positive.

Next, we verify that $\tl\theta^{\star v}$ from \eqref{th_varth_x} is square-integrable, namely $\tl\theta^{\star v}\in\tl\cA$. First, from the strict increase of the functions $h'$ and $H$ determined by (\ref{h_ODE}) and \eqref{TC}, \eqref{th_varth_x} implies that $\tl\theta^{\star v}_{t}$ is also strictly increasing in $t\in(0,T)$, besides obvious continuity, noting that $p>1$ here. Hence, for its square-integrability it suffices to look at its behavior as $t\nearrow T$, as in the proof of Proposition \ref{pro:4}.

First, note that based on (\ref{o_foc}), we have $\lim_{x\nearrow\bar x}h'(x)=\infty$, while $\lim_{x\nearrow\bar x}h(x)=h(\bar x)<\infty$, so we seek the asymptotic behavior $h'(x)\sim K(\bar x-x)^{-\delta}$ as $x\nearrow\bar x$, for $K,\delta>0$ independent of $x$. Then, by (\ref{h_ODE}), $h''(x)\sim K\delta(\bar x-x)^{-\delta-1}$, which plugged back in the ODE gives the asymptotic relation
\begin{equation}\label{asym}
  \delta(\bar x-x)^{-\delta-1}\sim \frac{p-1}{p\varsigma}e^{-\kappa \bar x}(h(\bar x))^{1/p-1}(\bar x-x)^{-\delta(2p-1)/(p-1)},\quad\text{as }x\nearrow\bar x.
\end{equation}
Matching the powers on both sides of \eqref{asym}, we must have
\begin{equation*}
  -\delta-1=-\frac{\delta(2p-1)}{p-1}\quad\Longrightarrow\quad \delta=\frac{p-1}{p},
\end{equation*}
so that $h'(x)\sim K(\bar x-x)^{-(p-1)/p}$, as $x\nearrow\bar x$, but since $\vartheta=(h')^{1/(p-1)}$, it follows that in the same limit,
\begin{equation}\label{varth_asym}
  \vartheta(x)\sim K^{1/(p-1)}(\bar x-x)^{-1/p}.
\end{equation}
Then, by \eqref{o_rep_t'} and \eqref{o_rep_T'}, we have that as $x\nearrow\bar x$,
\begin{equation*}
  T-t=\int^{\bar x}_{x}(\vartheta(y))^{-1}\dd y\sim K^{-1/(p-1)}\int^{\bar x-x}_{0}u^{1/p}\dd u=K^{-1/(p-1)}\frac{p}{p+1}(\bar x-x)^{1/p+1}.
\end{equation*}
Inverting the last relation and substituting it into (\ref{varth_asym}), we have that for some (generally different) $K>0$ independent of $t$,
\begin{equation*}
  \vartheta(x(t))\sim K((T-t)^{p/(p+1)})^{-1/p}=K(T-t)^{-1/(p+1)},\quad\text{as }t\nearrow T,
\end{equation*}
which implies the desired square-integrability because $p>1$. The proof is complete.
\end{proof}

\end{appendices}

\end{document}